\documentclass[a4paper,UKenglish,cleveref, autoref, thm-restate, numberwithinsect]{lipics-v2021}

\pdfoutput=1

\graphicspath{{./graphics/}}%

\title{The Parameterized Complexity of\\Extending Stack Layouts} %

\titlerunning{The Parameterized Complexity of Extending Stack Layouts} %

\author{Thomas Depian}{Algorithms and Complexity Group, TU Wien, Vienna, Austria}{tdepian@ac.tuwien.ac.at}{https://orcid.org/0009-0003-7498-6271}{}

\author{Simon D.~Fink}{Algorithms and Complexity Group, TU Wien, Vienna, Austria}{sfink@ac.tuwien.ac.at}{https://orcid.org/0000-0002-2754-1195}{}

\author{Robert Ganian}{Algorithms and Complexity Group, TU Wien, Vienna, Austria}{rganian@ac.tuwien.ac.at}{https://orcid.org/0000-0002-7762-8045}{}

\author{Martin N\"ollenburg}{Algorithms and Complexity Group, TU Wien, Vienna, Austria}{noellenburg@ac.tuwien.ac.at}{https://orcid.org/0000-0003-0454-3937}{}

\authorrunning{T.~Depian, S.~D.~Fink, R.~Ganian, M.~N\"ollenburg} %

\Copyright{Thomas Depian, Simon D.~Fink, Robert Ganian, Martin N\"ollenburg} %

\ccsdesc[500]{Theory of computation~Fixed parameter tractability}
\ccsdesc[500]{Mathematics of computing~Graphs and surfaces}

\keywords{Stack Layout, Drawing Extension, Parameterized Complexity, Book Embedding} %

\category{} %

\relatedversion{A preliminary version of this paper appeared in the proceedings of the 32nd International Symposium on Graph Drawing and Network Visualization (GD 2024). This version is published in the Journal of Graph Algorithms and Applications (JGAA).}
\relatedversiondetails[cite=GD-PAPER]{Conference Version}{https://doi.org/10.4230/LIPIcs.GD.2024.12}
\relatedversiondetails[cite=JGAA]{Journal Version}{https://doi.org/10.7155/jgaa.v29i3.3221} %

\hideLIPIcs

\funding{All authors acknowledge support from the Vienna Science and Technology Fund (WWTF) [10.47379/ICT22029]. Robert Ganian and Thomas Depian furthermore acknowledge support from the Austrian Science Fund (FWF) [10.55776/Y1329].}
\acknowledgements{The authors thank the reviewers for their helpful comments to improve the paper.}
\nolinenumbers

\usepackage{etoolbox}
\newtoggle{long}\toggletrue{long}
\iftoggle{long}{
	\NewDocumentEnvironment{prooflater}{m}{\begin{proof}}{\end{proof}}
	\NewDocumentEnvironment{proofsketch}{o +b}{}{}
	\newcommand{\restateref}[1]{}
	\NewDocumentEnvironment{statelater}{m}{}{}
}{
	\NewDocumentEnvironment{prooflater}{m +b}{ %
		\expandafter\global\expandafter\def\csname#1\endcsname{\begin{proof}#2\end{proof}}%
	}{}
	\NewDocumentEnvironment{proofsketch}{O{Proof sketch.}}{\begin{proof}[#1]}{\end{proof}}
	\usepackage{apptools}
	\newcommand{\restateref}[1]{[\IfAppendix{\hyperref[#1]{\AppendixSymbol}}{\hyperref[#1*]{\AppendixSymbol}}]}
	\NewDocumentEnvironment{statelater}{m +b}{%
		\expandafter\global\expandafter\def\csname#1\endcsname{#2}%
	}{\ignorespacesafterend}
}

\usepackage{graphicx}
\usepackage{amsmath}
\usepackage{amssymb}

\usepackage{cite}
\usepackage{complexity}
\usepackage{xparse}
\usepackage{xspace}
\usepackage{todonotes}
\usepackage{pifont}
\usepackage{mdframed}
\usepackage{mathtools}
\usepackage{thm-restate}

\usepackage{enumerate}
\usepackage[capitalise, noabbrev]{cleveref}

\graphicspath{{./graphics/}}

\usetikzlibrary{arrows.meta,patterns,ipe}

\crefname{property}{Property}{Properties}
\Crefname{property}{Property}{Properties}
\crefname{FP}{FP}{FPs}
\Crefname{FP}{FP}{FPs}
\crefname{EC}{EC}{ECs}
\Crefname{EC}{EC}{ECs}

\let\oldrestatable\restatable
\def\restatable{\expandafter\oldrestatable}

\newtheorem{reduction-rule}{Reduction Rule}
\newtheorem{property}{Property}

\newcommand{\Vadd}{\ensuremath{V_{\mathrm{add}}}\xspace}
\newcommand{\nadd}{\ensuremath{n_{\mathrm{add}}}\xspace}
\newcommand{\Vinc}{\ensuremath{V_{\mathrm{inc}}}\xspace}
\newcommand{\Eadd}{\ensuremath{E_{\mathrm{add}}}\xspace}
\newcommand{\madd}{\ensuremath{m_{\mathrm{add}}}\xspace}
\newcommand{\EaddH}{\ensuremath{E_{\mathrm{add}}^H}\xspace}

\NewDocumentCommand{\lSL}{o}{\IfNoValueTF{#1} 
	{\ensuremath{\langle\prec,\sigma\rangle}}%
	{\ensuremath{\langle\prec_{#1},\sigma_{#1}\rangle}}\xspace}%
\newcommand{\successor}[1]{\ensuremath{\text{succ}(#1)}}
\newcommand{\predecessor}[1]{\ensuremath{\text{pred}(#1)}}
\newcommand{\successorSpine}[2]{\successor{#1,\ #2}}
\newcommand{\predecessorSpine}[2]{\predecessor{#1,\ #2}}
\newcommand{\intervalPlacing}[1]{\ensuremath{\curlyvee(#1)}}
\newcommand{\pageWidth}[1]{\ensuremath{\omega(#1)}}%
\newcommand{\superInterval}{super interval}
\newcommand{\superIntervals}{super intervals}
\newcommand{\restrictedSpineOrder}{\ensuremath{\prec_{G \setminus H}}}
\newcommand{\vertexInside}[2]{\ensuremath{#1(#2)}}
\newcommand{\vertexOutside}[2]{\vertexInside{\overline{#1}}{#2}}

\newcommand{\assignedSuperInterval}[1]{\ensuremath{\curlyvee(#1)}}

\newcommand{\instance}{\ensuremath{\mathcal{I}}\xspace}
\newcommand{\instanceLong}{\ensuremath{\left(\ell, G, H, \lSL\right)}}

\newcommand{\probname}[1]{{\normalfont\textsc{#1}}}
\newcommand{\SLE}{\probname{Stack Layout Extension}\xspace}
\newcommand{\SLEShort}{\probname{SLE}\xspace}
\newcommand{\MCC}{\probname{McC}}
\newcommand{\ThreeSat}{\probname{3-Sat}}

\newcommand{\probdef}[3]{%
	\begin{mdframed}
		\probname{#1}
		\begin{description}
			\item[Given] #2
			\item[Question] #3
		\end{description}
	\end{mdframed}
}%

\newcommand{\Size}[1]{\ensuremath{\left\vert #1 \right\vert}}
\newcommand{\BigO}[1]{\ensuremath{\mathcal{O}(#1)}}
\newcommand{\AppendixSymbol}{\ding{72}}

\begin{document}
	
	\maketitle

	\begin{abstract}
		An $\ell$-page stack layout (also known as an $\ell$-page book embedding) of a graph is a linear order of the vertex set together with a partition of the edge set into $\ell$ stacks (or pages), such that the endpoints of no two edges on the same stack alternate.
		We study the problem of extending a given partial $\ell$-page stack layout into a complete one, which {%
			is} a natural generalization of the classical \NP-hard problem of computing a stack layout of an input graph from scratch. Given the inherent intractability of the problem, we focus on identifying tractable fragments through the refined lens of parameterized-complexity analysis. 
		Our results paint a detailed and surprisingly rich complexity-theoretic landscape of the problem which includes the identification of para\NP-hard, \W[1]-hard, and \XP-tractable, as well as fixed-parameter tractable fragments of stack layout extension via a natural sequence of parameterizations.
	\end{abstract}

	\section{Introduction}
	\label{sec:introduction}
	An $\ell$-page stack layout (or $\ell$-page book embedding) of a graph $G$ consists, combinatorially speaking, of (i) a linear order $\prec$ of its vertex set $V(G)$ and (ii) a partition $\sigma$ of its edge set $E(G)$ into $\ell \ge 1$ \emph{(stack-)}\emph{pages} such that for no two edges (with distinct endpoints) $uv$ and $wx$ with $u \prec v$ and $w \prec x$ that are assigned to the same page their endpoints alternate in $\prec$, i.e., we have %
	$u \prec w \prec v \prec x$. 
	When drawing a stack layout, the vertices are placed on a line called the \emph{spine} in the order given by $\prec$ and the edges of each page are drawn as pairwise non-crossing arcs in a separate half-plane bounded by the spine, see \cref{fig:example}a. 
	Stack layouts are a classic and well-studied topic in graph drawing and graph theory~\cite{Bernhart1979,Ollmann1973,Bilski1992}. 
	They have immediate applications in graph visualization~\cite{w-dvss-02,bb-crcl-04,gk-icl-06} as well as in bioinformatics, VLSI design, and parallel computing~\cite{Chung1987,Haslinger1999}; see also the overview by Dujmović and Wood~\cite{Dujmovic2004}.

	The minimum number $\ell$ such that a given graph $G$ admits an $\ell$-page stack layout is known as the \emph{stack number}, \emph{page number}, or \emph{book thickness} of $G$. 
	While the graphs with stack number $\ell=1$ are the outerplanar graphs, which can be recognized in linear time, the problem of computing the stack number is \NP-complete in general.
	Indeed, the class of graphs with stack number $\ell \le 2$ are precisely the subhamiltonian graphs (i.e., the subgraphs of planar Hamiltonian graphs) and recognizing them is \NP-complete~\cite{Bernhart1979,Chung1987,w-chcpmpg-82}.
	\begin{figure}
		\centering
		\includegraphics[page=1]{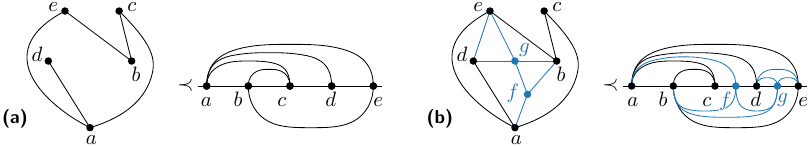}
		\caption{\textbf{\textsf{(a)}} A graph $H$ and a two-page stack layout of it. In~\textbf{\textsf{(b)}}, the graph $H$ and its two-page stack layout are extended by the new vertices and edges marked in blue.}
		\label{fig:example}
	\end{figure}
	Computing the stack number is known to also remain \NP-complete if the vertex order is provided as part of the input and $\ell=4$~\cite{Unger88}, and overcoming the intractability of these problems has been the target of several recent works in the field~\cite{Bhore2020,BhoreGMN19,Liu2021,GMOPR24}.
	Many other results on stack layouts are known--for instance, every planar graph has a $4$-page stack layout and this bound is tight~\cite{Yannakakis1989,KaufmannBKPRU20}. 
	For a comprehensive list of known upper and lower bounds for the stack number of different graph classes, we refer to the collection by 
	Pupyrev~\cite{Pupyrev.2023}.

	In this paper, we take a new perspective on stack layouts, namely the perspective of drawing extensions. 
	In drawing-extension problems, the input consists of a graph $G$ together with a partial drawing of $G$, i.e., a drawing of a subgraph $H$ of $G$. %
	The task is to insert the vertices and edges of~$G$ which are missing in $H$ in such a way that a desired property of the drawing is maintained; see \cref{fig:example}b for an example. 
	Such drawing-extension problems occur, e.g., when visualizing dynamic graphs in a streaming setting, where additional vertices and edges arrive over time and need to be inserted into the existing partial drawing. 
	Drawing-extension problems have been investigated for many types of drawings in recent years--including planar drawings~\cite{Angelini2015,jkr-ktppeg-13,p-epsd-06,Patrignani05a}, upward planar drawings~\cite{lbf-eupgd-20}, level planar drawings~\cite{Brueckner2017}, $1$-planar drawings~\cite{Eiben2020,Eiben2020a}, planar orthogonal drawings~\cite{Angelini2021,AngeliniRP20,Bhore2023a}, {%
		drawings with bounded geometric thickness~\cite{THICKNESS-SOFSEM}}{%
		, and, as a follow-up work to the conference version of this article~\cite{GD-PAPER}, for queue layouts~\cite{QLE}}---but until now, essentially nothing was known about the extension of stack layouts/book embeddings.

	Since it is \NP-complete to determine whether a graph admits an $\ell$-page stack layout (even when $\ell$ is a small fixed integer), the extension problem for $\ell$-page stack layouts is \NP-complete as well--after all, setting $H$ to be empty in the latter problem fully captures the former one.
	In fact, the extension setting can seamlessly also capture the previously studied \NP-complete problem of computing an $\ell$-stack layout with a prescribed vertex order~\cite{Chung1987,Unger88,Unger.1992,Bhore2020,BhoreGMN19}; indeed, this corresponds to the special case where $V(H) = V(G)$ and $E(H)=\emptyset$.
	Given the intractability of extending $\ell$-page stack layouts in the classical complexity setting, we focus on identifying tractable fragments of the problem through the more refined lens of parameterized-complexity analysis~\cite{DowneyF13,Cygan2015}, which considers both the input size of the graph and some additional parameter $k$ of the instance\footnote{We assume familiarity with the basic foundations of \emph{parameterized complexity theory}, notably including the notions of \emph{fixed-parameter tractability}, \XP, \W[1]-, and para\NP-hardness~\cite{Cygan2015}.}.

	\subparagraph*{Contributions.}
	A natural parameter in any drawing-extension problem is the size of the missing part of the graph, i.e., the missing number of vertices and/or edges. 
	We start our investigation by showing that the \textsc{Stack Layout Extension} problem (\SLEShort{}) for instances without any missing vertices, i.e., $V(G) = V(H)$, is fixed-parameter tractable when parameterized by the number of missing edges $|E(G)\setminus E(H)|$ (\cref{sec:only-edges}).
	
	The above result, however, only applies in the highly restrictive setting where no vertices are missing--generally, we would like to solve instances with missing vertices as well as edges. 
	A parameterization that has been successfully used in this setting is the \emph{vertex+edge deletion distance}, i.e., the number of vertex and edge deletion operations\footnote{As usual, we assume that deleting a vertex automatically also deletes all of its incident edges.} required to obtain~$H$ from~$G$. But while this parameter has yielded parameterized algorithms when extending, e.g., 1-planar drawings~\cite{Eiben2020,Eiben2020a} and orthogonal planar drawings~\cite{Bhore2023a}, we rule out any analogous result for \SLEShort\ by establishing its \NP-completeness even if $H$ can be obtained from $G$ by deleting only two vertices (\cref{sec:paranp-hardness}). This means that more ``restrictive'' parameterizations are necessary to achieve tractability for the problem of extending $\ell$-page stack layouts.
	\begin{figure}[t]
		\centering
		\input{graphics/complexity-landscape}
		\caption{The complexity landscape of \SLE{}. VEDD denotes the vertex+edge deletion distance, $\omega$ denotes the page width of the $\ell$-page stack layout of $H$, and $\kappa = \Size{V(G) \setminus V(H)} + \Size{E(G)\setminus E(H)}$.
			Boxes outlined in bold represent new results {%
				for the given parameterizations} that we show in the linked theorems and corollaries. {%
				The only results that are not depicted are \cref{thm:only-edges-fpt,thm:kappa-page-number-fpt}}.}
		\label{fig:complexity-landscape}
	\end{figure}
	Since the missing vertices in our hardness reduction have a high degree, we then consider parameterizations by the combined number of missing vertices and edges $\kappa = \Size{V(G) \setminus V(H)} + \Size{E(G)\setminus E(H)}$.
	We show that \SLEShort{} belongs to the class \XP\ when parameterized by $\kappa$ (\cref{sec:xp}) while being \W[1]-hard (\cref{sec:w-1}), which rules out the existence of a fixed-parameter tractable algorithm under standard complexity assumptions.
	The latter result holds even if we additionally bound the page width $\omega$ of the stack layout of $H$,  which measures the maximum number of edges that are crossed on a single page by a line perpendicular to the spine~\cite{Chung1987}.
	On our quest towards a fixed-parameter tractable fragment of the problem, we thus need to include another restriction, namely the number $\ell$ of pages of the stack layout. 
	So finally, when parameterizing \SLEShort{} by the combined parameter $\kappa + \omega + \ell$, we show that it becomes fixed-parameter tractable (\cref{sec:fpt}). 
	{%
		The last result gives rise to the question whether \SLEShort remains fixed-parameter tractable under the more general parameter $\kappa + \ell$.
		As our final result, we provide a positive answer to this question for the restricted setting where no two missing vertices are adjacent (\cref{sec:fpt-independent-set}).
		We summarize the complexity landscape of \SLEShort\ in \cref{fig:complexity-landscape}.
	}

	\iftoggle{long}{}{
		\smallskip
		\noindent
		\emph{Full proofs of statements marked by \AppendixSymbol{} can be found in the full version~\cite{ARXIV}.}
	}

	\section{Preliminaries}
	\label{sec:preliminaries}
	We assume the reader to be familiar with standard graph terminology~\cite{Diestel2012}.
	Throughout this paper, we assume standard graph representations that allow for efficient graph modifications.
	For two integers $p \leq q$ we denote with $[p, q]$ the set $\{p, p+1, \ldots, q\}$ and use $[p]_0$ and $[p]$ as abbreviations for $[0, p]$ and $[1, p]$, respectively.
	Let $G$ be a graph that is, unless stated otherwise, simple and undirected, with vertex set $V(G)$ and edge set $E(G)$.
	For $X\subseteq V(G)$, we denote by $G[X]$ the subgraph of $G$ induced on $X$. 

	\subparagraph*{Stack Layouts.}
	For an integer $\ell\geq 1$, an \emph{$\ell$-page stack layout} of $G$ is a tuple \lSL[G] where~$\prec_G$ is a linear order of $V(G)$ and $\sigma_G \colon E(G) \to [\ell]$ is a function that assigns each edge to a \emph{page} $p \in [\ell]$ such that  %
	for each pair of edges ${%
		e_1 =\ } u_1v_1$ and ${%
		e_2 =\ } u_2v_2$ with $\sigma_G(e_1) = \sigma_G(e_2)$ it does not hold $u_1 \prec_G u_2 \prec_G v_1 \prec_G v_2$.
	For the remainder of the paper, we write $\prec$ and $\sigma$ if the graph~$G$ is clear from context.
	We call~$\prec$ the \emph{spine (order)} and~$\sigma$ the \emph{page assignment}. %
	Observe that we can interpret a stack layout as a drawing of $G$ on $\ell$ different half-planes, one per page $p \in [\ell]$, each of which is bounded by the straight-line spine delimiting all half-planes.  
	{%
		On each page $p \in [\ell]$, all edges on $p$ are drawn as (circular) arcs that run monotonically in the direction of the spine, forming a planar drawing of the page.}
	One fundamental property of a stack layout is its \emph{page width}--denoted as \pageWidth{\lSL} or simply $\omega$ if \lSL\ is clear from context--which is the maximum number of edges that are crossed on a single page by a line perpendicular to the spine~\cite{Chung1987}.
	The properties of stack layouts with small page width have been studied, e.g., by St{\"{o}}hr~\cite{Stoehr.1988,Stohr91}.

	We say that two vertices $u$ and $v$ are \emph{consecutive on the spine} if they occur consecutively in~$\prec$.
	A vertex $u \in V(G)$ \emph{sees} a vertex $v \in V(G)$ on a page $p \in [\ell]$ if there does not exist an edge $e = xy \in E(G)$ with $\sigma(e) = p$ and $x \prec u \prec y \prec v$ or $u \prec x \prec v \prec y$.
	Note that if $u$ sees~$v$, then $v$ also sees $u$.
	For two vertices $u$ and $v$ which are consecutive in $\prec$, we refer to the  segment 
	on the spine between $u$ and $v$ as the \emph{interval} between $u$ and $v$, denoted as $[u, v]$.
	
	\subparagraph*{Problem Statement.}
	Let $H \subseteq G$ be a subgraph of a graph $G$.
	We say that \lSL[G] is an \emph{extension} of \lSL[H] if $\sigma_H\subseteq \sigma_G$ and $\prec_H\ \subseteq\  \prec_{G}$.
	We now formalize our problem of interest:

	\probdef{\SLE~(\SLEShort)}{Integer $\ell \geq 1$, graph $G$, subgraph $H$ of $G$, and $\ell$-page stack layout \lSL[H].}{Does there exist an $\ell$-page stack layout \lSL[G] of $G$ that extends \lSL[H]?}
	
	We remark that while \SLEShort is defined as a decision problem for complexity-theoretic reasons, every algorithm presented in this article is constructive and can be trivially adapted to also output a layout \lSL[G] as a witness (also called a \emph{solution}) for positive instances. For an instance {%
		$\instance = (\ell, G, H, \lSL[H])$} of \SLEShort, we use $\Size{\instance}$ as shorthand for $|V(G)|+|E(G)|+\ell$.
	{%
		Furthermore, we use $\instanceLong$ as a shorthand for $(\ell, G, H, \lSL[H])$ when denoting instances of \SLEShort.}
	
	In line with the terminology previously used for drawing-extension problems~\cite{Eiben2020}, we refer to the vertices and edges in $V(H)\cup E(H)$ as \emph{old} and call all other vertices and edges of~$G$ \emph{new}.
	Let $\Vadd$ and $\Eadd$ denote the sets of all new vertices and edges, respectively, and set $\nadd{} \coloneqq \Size{\Vadd{}}$ and $\madd{} \coloneqq \Size{\Eadd{}}$. Furthermore, we denote with \EaddH the set of new edges incident to two old vertices, i.e., $\EaddH{} \coloneqq \{e = uv \in \Eadd{} \mid u, v \in V(H)\}$.
	We consider the parameterized complexity of our extension problem by measuring how ``incomplete'' the provided partial solution is using the following natural parameters that have also been used in this setting before~\cite{Eiben2020,Eiben2020a,Ganian2021,Bhore2023,BhoreLMN21}:
	the \emph{vertex+edge deletion distance}, which is $\nadd + \Size{\EaddH}$, and the total \emph{number of missing vertices and edges}, i.e., $\nadd + \madd$. 

	\section{\textsc{SLE} with Only Missing Edges Is \textsf{FPT}}
	\label{sec:only-edges}
	We begin our investigation by first analyzing the special case where $V(G)=V(H)$, i.e., when only edges are missing from $H$. We recall that the problem remains \NP-complete even in this setting, as it generalizes the problem of computing the stack number of a graph with a prescribed vertex order~\cite{Chung1987,Unger88,Unger.1992,Bhore2020,BhoreGMN19}.
	Furthermore, both aforementioned measures of the incompleteness of $\lSL[H]$ are the same and equal $\madd = \Size{\EaddH}$. 
	As a ``warm-up'' result, we show that in this setting \SLEShort\ is fixed-parameter tractable parameterized by $\madd$.
	
	Towards this, consider the set $S(e) \subseteq [\ell]$ of pages on which we could place a new edge~$e$ without introducing a crossing with edges from $H$; formally, $p\in S(e)$ if and only if $\langle\prec_{H},\sigma_{H}\cup {(e,p)}\rangle$ is an $\ell$-page stack layout of $H\cup \{e\}$.
	Intuitively, if $\Size{S(e)}$ is large enough, then we are always able to find a %
	``free'' page to place $e$ independent of the placement of the remaining new edges.
	Formally, one can easily show:

	\begin{restatable}\restateref{lem:only-edges-remove-easy-edges}{lemma}{lemmaOnlyEdgesRemoveEasyEdges}
		\label{lem:only-edges-remove-easy-edges}
		Let $\instance = \instanceLong$ be an instance of \SLEShort{} with $\Vadd{} = \emptyset$ that contains an edge $e \in \Eadd{}$ with $\Size{S(e)} \geq \madd{}$.
		The instance {%
			$\instance' = \left(\ell, G', H, \lSL\right)$} with $G' = G \setminus \{e\}$ is a positive instance if and only if \instance{} is a positive instance.
	\end{restatable}
	\begin{prooflater}{plemmaOnlyEdgesRemoveEasyEdges}
		First, note that removing an edge from $G$ and adapting the page assignment $\sigma$ accordingly does not invalidate an existing solution \lSL\ to \SLEShort{} for \instance.
		Hence, the ``($\boldsymbol{\Leftarrow}$)-direction'' holds trivially, and we focus on the ``($\boldsymbol{\Rightarrow}$)-direction''.
		
		\proofsubparagraph*{($\boldsymbol{\Rightarrow}$)}
		Let $\instance'$ be a positive instance of \SLEShort{} with the solution \lSL.
		By our selection of~$e$, there exists a page $p \in S(e)$ such that we have for every edge $e' \in \Eadd{}$ with $e' \neq e$ that $\sigma(e') \neq p$ holds.
		We take \lSL and extend $\sigma$ by the page assignment $(e, p)$ to obtain \lSL[G].
		By the definition of $S(\cdot)$, this cannot introduce a crossing with edges from $H$ and by our selection of $p$, no crossings with other edges from $\Eadd{}$ are possible either.
		Hence, \lSL[G] is a valid stack layout of~$G$ that extends \lSL[H] as we did not alter \lSL except extending the page assignment.
		Thus, it witnesses that \instance{} is a positive instance of \SLEShort{}.
	\end{prooflater}
	\noindent
	With \cref{lem:only-edges-remove-easy-edges} in hand, we can establish the desired result:
	\begin{restatable}\restateref{thm:only-edges-fpt}{theorem}{theoremOnlyEdgesFPT}
		\label{thm:only-edges-fpt}
		Let $\instance = \instanceLong$ be an instance of \SLEShort{} with $\Vadd = \emptyset$.
		We can find an $\ell$-page stack layout of $G$ that extends \lSL or report that none exists in \BigO{\madd{}^{\madd{}}\cdot\Size{\instance{}}} time, where $\madd$ denotes the number of new edges.
	\end{restatable}
	\begin{proofsketch}
		We compute for a single edge $e \in \Eadd{}$ the set $S(e)$ in linear time by checking with which of the old edges $e$ would cross.
		If $S(e) \geq \madd{}$, then following \cref{lem:only-edges-remove-easy-edges}, we remove~$e$ from $G$.
		Overall, this takes $\BigO{\madd{} \cdot \Size{\instance{}}}$ time and results in a graph $G'$ with $H \subseteq G' \subseteq G$.
		Furthermore, each remaining new edge $e' \in E(G') \setminus E(H)$ can be put in fewer than~$\madd{}$ different pages.
		Hence, we can brute-force over all the at most $\BigO{\madd{}^{\madd{}}}$ page assignments $\sigma'$ that extend $\sigma_H$ for all edges in $E(G') \setminus E(H)$, and for each such assignment we check in linear time whether no pair of edges $e', e''\in E(G')\setminus E(H)$ cross each other.
	\end{proofsketch}
	\begin{prooflater}{ptheoremOnlyEdgesFPT}
		We compute for a single edge $e \in \Eadd{}$ the set $S(e)$ in linear time by checking with which of the old edges $e$ would cross.
		If $S(e) \geq \madd{}$, then we remove $e$ from $G$.
		Overall, this takes $\BigO{\madd{} \cdot\Size{\instance{}}}$ time and results in a graph $G'$ with $H \subseteq G' \subseteq G$.
		Furthermore, each edge $e' \in E(G') \setminus E(H)$ can be put in fewer than $\madd{}$ different pages.
		Hence, we can brute force the possible page assignments for each new edge $e'$.
		Each of the resulting \BigO{\madd{}^{\madd{}}} different $\ell$-page stack layouts \lSL is by construction an extension of \lSL[H].
		Creating \lSL can be done by copying \lSL[H] and then augmenting it with $\BigO{\madd{}}$ new edges.
		This amounts to \BigO{\Size{\instance{}}} time.
		{%
			For each created $\ell$-page stack layout, we can check in linear time whether it is crossing free.
			To this end, we can traverse the spine from left to right and maintain one stack per page that stores one copy of the first endpoint of every edge for which exactly one endpoint has already been visited.
			Checking for crossing-freeness now corresponds to verifying that whenever we visit a vertex $v$, exactly the vertices $\{u \mid uv \in E(G), u \prec v, \sigma(uv) = p\}$ are on the top of the stack for the page~$p$, for every page $p \in [\ell]$.
			Overall, this takes time linear in the number of edges of $G$.
		}

		If there exists a crossing free layout, then applying \cref{lem:only-edges-remove-easy-edges} iteratively tells us that we can extend it to a solution \lSL[G] for~\instance.
		If none of them is crossing free, we conclude by applying (iteratively) \cref{lem:only-edges-remove-easy-edges} that \instance{} does not admit the desired $\ell$-page stack layout.
		Combining all, the overall running time is \BigO{\madd{}^{\madd{}}\cdot\Size{\instance{}}}.
	\end{prooflater}

	\section{\textsc{SLE} with Two Missing Vertices Is \textsf{NP}-Complete}
	\label{sec:paranp-hardness}
	Adding only edges to a given linear layout is arguably quite restrictive.
	Therefore, we now lift this restriction and consider \SLEShort{} in its full generality, i.e., also allow adding vertices.
	Somewhat surprisingly, as our first result in the general setting we show that \SLEShort{} is \NP-complete even if the task is to merely add two vertices, i.e., for $\nadd{} = 2$ and $\EaddH{} = \emptyset$.
	This rules out not only fixed-parameter but also \XP\ tractability when parameterizing by the vertex+edge deletion distance, and represents--to the best of our knowledge--the first example of a drawing-extension problem with this behavior.%
	
	To establish the result, we devise a reduction from \ThreeSat~\cite{Garey.1979}.
	\iftoggle{long}{}{%
		In our reduction, we insert two new vertices into a partial layout derived from the given formula, and use the page assignment of their incident edges to encode a truth assignment and validate that it satisfies all clauses.
		For this, we will need to restrict the positions of the new vertices to a certain range along the spine.
		In \cref{sec:fixation-gadget}, we introduce the \emph{fixation gadget} that ensures this.
		We also reuse this gadget in the reduction shown in \cref{sec:w-1}.
		But first, we use it in this section to perform our reduction and prove \NP-completeness in \cref{sec:paranp-hardness-construction}.%
	}%
	\begin{statelater}{sectionParaNPIntuition}%
		Let $\varphi = (\mathcal{X}, \mathcal{C})$ be an instance of \ThreeSat{} consisting of $N$ variables $\mathcal{X} = \{x_1, \ldots, x_N\}$ and $M$ clauses $\mathcal{C} = \{c_1, \ldots, c_M\}$, each consisting of three different and pairwise non-complementary literals.
		\iftoggle{long}{Intuitively, the}{The}
		reduction constructs an instance $\instance{} = \instanceLong{}$ of \SLEShort{} which represents each variable $x_i$ and each clause $c_j$ of $\varphi$, respectively, by a corresponding vertex in $H$.
		The linear order~$\prec_H$ has the form $x_1 \prec x_2 \prec \ldots \prec x_N \prec c_1 \prec \ldots \prec c_M$; see \iftoggle{long}{\cref{fig:paranp-hardness-base-layout-edges}}{\cref{fig:paranp-hardness-base-layout-edges-no-labels}}.
		Furthermore, \instance contains two new vertices~$s$ and $v$.
		The vertex $s$ is adjacent to all variable-vertices and the construction will ensure that the page assignment for its incident edges represents, i.e., \emph{\textbf{s}}elects, a truth assignment $\Gamma$ for $\varphi$.
		The vertex $v$ is adjacent to all clause-vertices, and its purpose is to \emph{\textbf{v}}erify that the truth assignment satisfies all clauses.
		For the following\iftoggle{long}{ high-level}{} description of how this is achieved, we assume $s \prec v \prec x_1$ \iftoggle{long}{and will ensure later}{as we will use a fixation gadget to ensure} that every solution \lSL of \instance{} has this property.
		
		To each variable $x_i$, we associate two pages $p_i$ and $p_{\lnot i}$ corresponding to its possible truth states.
		We ensure that $s$ can see each variable-vertex {%
			$x_i$} only on its associated pages using edges {%
			between two dummy vertices $d_i$ and $d_{i + 1}$ with $s \prec d_i \prec x_i \prec d_{i + 1}$, for all $i \in [N]$.}%
		\iftoggle{long}{}{These dummy vertices are distributed as in \cref{fig:paranp-hardness-base-layout-edges-no-labels}.}
		Hence, a page assignment for the edges incident to $s$ induces a truth assignment.
		{%
			Similarly,} edges {%
			between two dummy vertices $d_{N + j}$ and $d_{N + j + 1}$ with $v \prec d_{N + j} \prec c_j \prec d_{N + j + 1}$} %
		ensure that $v$ can see a clause-vertex $c_j$ only on the pages that are associated to the negation of the literals the clause $c_j$ is composed of{%
			, for all $j \in [M]$;} see \iftoggle{long}{\cref{fig:paranp-hardness-base-layout-edges}}{the blue arcs in \cref{fig:paranp-hardness-base-layout-edges-no-labels} for an illustration}.
		We defer the full construction to \iftoggle{long}{\cref{sec:paranp-hardness-base-layout}}{the full version~\cite{ARXIV}}.
		\iftoggle{long}{}{%
			
			We now ensure that $s \prec v \prec x_1$ holds in every solution of \instance{} by using a fixation gadget on two vertices, i.e., for $F = 2$.
			In particular, we set $a_3 \prec d_1$, i.e., we place the fixation gadget at the beginning of the spine, and identify $s = f_1$ and $v = f_2$.
			The spine order $\prec_H$ is then the transitive closure of all the partial orders stated until now; see \cref{fig:paranp-hardness-base-layout-edges-no-labels}.
			Finally, we add the edge $d_1d_{N + M + 1}$ and set $\sigma(d_1d_{N + M + 1}) = p_d$ to ensure that our construction has \cref{prop:fixation-gadget-dummy-page}.
			
			Regarding the correctness of our reduction, we make the following observation.
		}%
		{%
			We observe that the edge $vc_{j}$ must be assigned to a page associated to the negated literals of~$c_j$.
			If an induced truth assignment does not satisfy a clause $c_j$, then $vc_j$ will cross another edge no matter
		}%
		which page we use\iftoggle{long}{, see also \cref{fig:paranp-hardness-base-layout-s-v}}{}.
		However, if a clause $c_j$ is satisfied, we can find a page for the edge $vc_j$ that does not introduce a crossing: the page associated to the negation of the literal that satisfies~$c_j$.
		Consequently, if~$\varphi$ is satisfiable, then there exists an extension \lSL[G].
		Similarly, the page assignment of an extension \lSL[G] induces a truth assignment $\Gamma$ that satisfies~$\varphi$.
		\iftoggle{long}{}{%
			An intuitive example of the reduction is provided in \cref{fig:paranp-hardness-example}, and we obtain the following theorem.
		}
		
	\end{statelater}
	
	\iftoggle{long}{
		Finally, recall that our approach hinges on some way to restrict the new vertices $s$ and $v$ to be placed within a certain range, i.e., before $x_1$.
		We realize this using the \emph{fixation gadget} that we describe in \cref{sec:fixation-gadget}.
		We also reuse this gadget in the reduction from \cref{sec:w-1}.
		In \cref{sec:paranp-hardness-construction}, we show how this can be build into our reduction to prove \cref{thm:paranp-hardness}.
	}{}
	
	The graph $H$ that we construct will have multi-edges to facilitate the presentation of the reduction. The procedure for removing multi-edges is detailed in \iftoggle{long}{\cref{app:removing-multi-edges}}{the full version~\cite{ARXIV}}.
	\begin{statelater}{sectionParaNPBaseLayout}
		\begin{figure}[t]
			\centering
			\includegraphics[page=5]{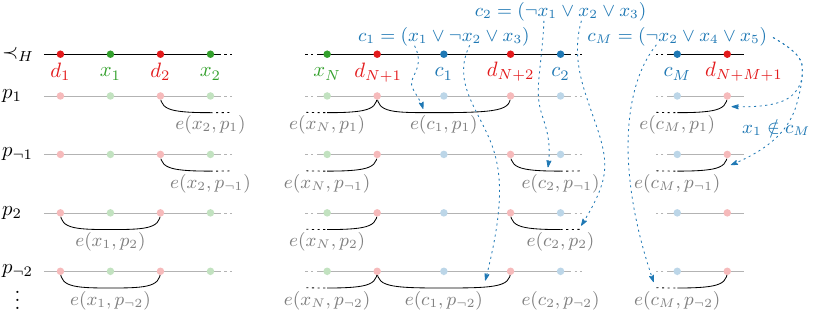}
			\caption{Parts of the spine order $\prec_{H}$. Green vertices represent variables, blue vertices clauses, and red vertices the dummy vertices. Furthermore, we visualize some edges in $H$ that are created for the variable-vertices (left) and clause-vertices (middle and right).
				If an edge is created due to the (non-)existence of a literal in the clause $c_1$, $c_2$, or $c_M$ it is indicated via a blue arc. \iftoggle{long}{}{Extended version of \cref{fig:paranp-hardness-base-layout-edges-no-labels}.}}
			\label{fig:paranp-hardness-base-layout-edges}
		\end{figure}
		\subsection{Encoding the Variables and Clauses: The Base Layout}
		\label{sec:paranp-hardness-base-layout}
		For each variable $x_i$ and each clause $c_j$ in $\varphi$ we introduce a vertex in $H$.
		In the following, we use the same symbol to address an element of $\varphi$ and its representation in $H$.
		Let us first fix the spine order $\prec_H$.
		For every $i \in [N - 1]$ and $j \in [M - 1]$, we set $x_i \prec x_{i + 1}$ and $c_{j} \prec c_{j + 1}$.
		Furthermore, we order the variables before the clauses on the spine, i.e., we set $x_N \prec c_1$.
		Next, we add $N + M + 1$ dummy vertices {%
			$d_1$, $d_2$, $\ldots$, $d_{N + M + 1}$} to $H$ and distribute them on the spine.
		More concretely, we set $d_i \prec x_i \prec d_{i + 1}$ and $d_{N + j} \prec c_j \prec d_{N + j + 1}$ for every $i \in [N]$ and $j \in [M]$.
		By taking the transitive closure of the above order, we obtain the following linear order, which we also visualize in \cref{fig:paranp-hardness-base-layout-edges}.
		\begin{align*}
			d_1 \prec x_1 \prec d_2 \prec x_2 \prec \ldots \prec x_N \prec d_{N + 1} \prec c_1 \prec \ldots \prec c_M \prec d_{N + M + 1}
		\end{align*}
		We now turn our attention to the page assignment $\sigma_{H}$ and create for each variable $x_i \in \mathcal{X}$ the two pages $p_{i}$ and $p_{\lnot i}$.
		Intuitively, the assignment of an edge incident to $x_i$ to either of these two pages will determine whether $x_i$ is true or false.
		For each $i, j \in [N]$ with $i \neq j$, we create the edges $e(x_i, p_j) = d_id_{i + 1}$ and $e(x_i, p_{\lnot j}) = d_id_{i + 1}$ in $H$.
		We assign the edge $e(x_i, p_j)$ to the page $p_j$ and the edge $e(x_i, p_{\lnot j})$ to $p_{\lnot j}$, i.e., we have $\sigma(e(x_i, p_j)) = p_j$ and $\sigma(e(x_i, p_{\lnot j})) = p_{\lnot j}$.
		\cref{fig:paranp-hardness-base-layout-edges} (left) visualizes this.
		Note that this introduces multi-edges in~$H$, but recall that we resolve them in \cref{app:removing-multi-edges}. In particular, we will address these multi-edges in \cref{app:multi-edges-paranp-hardness}.
		Next, we consider each combination of a clause $c_j \in \mathcal{C}$ and a variable $x_i \in \mathcal{X}$.
		If $x_i$ does not appear in $c_j$, we create the edges $e(c_j, p_i) = d_{N + j}d_{N + j + 1}$ and $e(c_j, p_{\lnot i}) = d_{N + j}d_{N + j + 1}$.
		Similar to before, we set $\sigma(e(c_j, p_i)) = p_i$ and $\sigma(e(c_j, p_{\lnot i})) = p_{\lnot i}$.
		If $x_i$ appears in $c_j$ without negation, we create the edge $e(c_j, p_i) = d_{N + j}d_{N + j + 1}$ and set $\sigma(e(c_j, p_i)) = p_{i}$.
		Symmetrically, if $x_i$ appears negated in $c_j$, we create the edge $e(c_j, p_{\lnot i}) = d_{N + j}d_{N + j + 1}$ and set $\sigma(e(c_j, p_{\lnot i})) = p_{\lnot i}$.
		We visualize this page assignment in \cref{fig:paranp-hardness-base-layout-edges} (middle and right).
		
		This completes the base layout of our reduction.
		\iftoggle{long}{However, we have only defined parts of~$H$ and its stack layout \lSL and complete the construction in \cref{sec:paranp-hardness-construction}. }{}%
		Next, we introduce two new vertices~$s$ and $v$ in $G$.
		The vertex $s$ is adjacent to each $x_i \in \mathcal{X}$ and the vertex $v$ is adjacent to each $c_j \in \mathcal{C}$.
		Let us assume for the moment that in every extension \lSL[G] of \lSL[H] we have $s \prec v \prec d_1$.
		Then, the vertex $s$ can see each $x_i \in \mathcal{X}$ only on the pages $p_i$ and $p_{\lnot i}$.
		Hence, the page assignment of $sx_i$ can be interpreted as the truth state of $x_i$.
		Similarly, from the perspective of the vertex $v$, each $c_j \in \mathcal{C}$ is only visible on the pages that correspond to the complementary literals of $c_j$.
		So intuitively, if the page assignment of the edges incident to $s$ induces a truth assignment that falsifies $c_j$, then these new edges will block the remaining available pages for the edge $vc_j$.
		Hence, $v$ verifies the truth assignment induced by the edges incident to $s$; {%
			recall \cref{fig:paranp-hardness-base-layout-s-v}}.
		\begin{figure}[t]
			\centering
			\includegraphics[page=7]{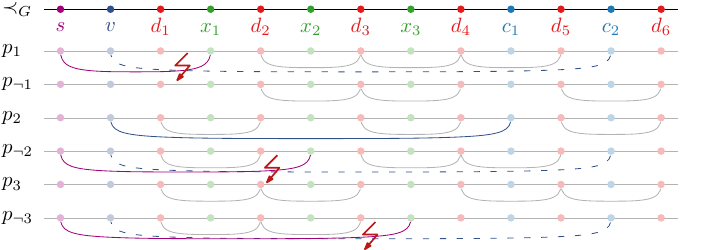}
			\caption{The edges incident to $s$ induce the truth assignment $x_1 = 1$, $x_2 = 0$, and $x_3 = 0$. This assignment satisfies the clause $c_1 =\left( x_1 \lor \lnot x_2 \lor x_3 \right)$, as we can set $\sigma_G(vc_1) = p_2$, but not the clause $c_2 = \left(\lnot x_1 \lor x_2 \lor x_3\right)$, as we cannot find a page for the edge $vc_2$.}
			\label{fig:paranp-hardness-base-layout-s-v}
		\end{figure}
		\iftoggle{long}{This indicates the intended semantics of our reduction but has still one caveat: 
			We have to ensure that we have $s \prec v \prec d_1$ in every solution to our created instance.}{Note that this relies on the assumption that we have $s \prec v \prec d_1$. However, we can satisfy this assumption using a fixation gadget.}
	\end{statelater}

	\subsection{Restricting the Placement of New Vertices: The Fixation Gadget}
	\label{sec:fixation-gadget}
	The purpose of the so-called \emph{fixation gadget} is to restrict the possible positions of new vertices to given intervals.
	As this gadget will also find applications outside this reduction, we describe in the following in detail its general construction for $F > 1$ new vertices $\mathcal{F} = \{f_1, \ldots, f_F\}$.

	First, we introduce $3(F + 1)$ new vertices $v_1, \ldots, v_{F + 1}$, $b_1, \ldots, b_{F + 1}$, and $a_1, \ldots, a_{F + 1}$. %
	{%
		For each $i \in [F +1]$, the vertex $b_i$ will be placed directly \emph{\textbf{b}}efore and the vertex $a_i$ directly \emph{\textbf{a}}fter the vertex $v_i$ on the spine.
		Therefore, we}
	fix the spine order $\prec_{H}$ among these vertices %
	to 
	$b_1 \prec v_1 \prec a_1 \prec b_2 \prec v_2 \prec a_2 \prec \ldots \prec b_{F + 1} \prec v_{F + 1} \prec a_{F + 1}$; see also \cref{fig:fixation-gadget-example}.
	Then, every new vertex $f_i$ is made adjacent to~$v_i$ and $v_{i + 1}$ and we aim to allow these new edges to be placed only in a dedicated further page $p_d$.
	To achieve this, we first introduce for every $i \in [F + 1]$ and every page $p \neq p_d$ an edge $e(b_i, a_i, p) = b_ia_i$ in $H$ and set $\sigma(e(b_i, a_i, p)) = p$; see \cref{fig:fixation-gadget-example}.
	Furthermore, we also introduce the edges $b_iv_i$ and $v_ia_i$ and set $\sigma(b_iv_i) = \sigma(v_ia_i) = p_d$ for {%
		every} $i \in [F + 1]$.
	For every $i \in [F]$, we add the edge $v_iv_{i + 1}$ and place it on the page~$p_d$, i.e., we have $\sigma(v_iv_{i + 1}) = p_d$.
	Finally, we also create the edge $b_1a_{F + 1}$ and set $\sigma(b_1a_{F + 1}) = p_d$.
	To complete the construction of the fixation gadget, we add the new edges $f_iv_i$ and $f_iv_{i + 1}$ for every $i \in [F]$ to~$G$.
	\cref{fig:fixation-gadget-example} shows an example of the fixation gadget for $F = 2$.
	\begin{figure}[t]
		\centering
		\includegraphics[page=4]{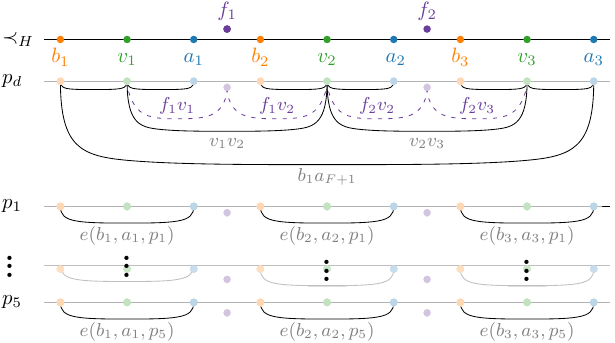}
		\caption{A fixation gadget for $F = 2$ with five other pages in the stack layout. We also highlight the intended position for $f_1$ and $f_2$ on the spine and the page assignment for their incident edges.}
		\label{fig:fixation-gadget-example}
	\end{figure}
	
	Next, we show that the fixation gadget forces $f_i$ to lie between $v_i$ and $v_{i + 1}$ on the spine and the edges $f_iv_i$ and $f_iv_{i + 1}$ to be on the page $p_d$ for every $i \in [F]$.
	\begin{restatable}\restateref{lem:fixation-gadget-properties}{lemma}{lemmaFixationGadgetProperties}
		\label{lem:fixation-gadget-properties}
		Let $\instance = \instanceLong$ be an instance of \SLEShort{} that contains a fixation gadget on $F$ vertices $\{f_1, \ldots, f_{F}\}$.
		In any solution \lSL[G] to \instance{} and for every $i \in [F]$, we have $v_i \prec f_i \prec v_{i + 1}$ and $\sigma(f_iv_i) = \sigma(f_{i}v_{i + 1}) = p_d$.
		Furthermore, the fixation gadget contributes $4F + 3$ vertices and $(\ell + 4)F + \ell + 2$ edges to the size of \instance.
	\end{restatable}
	\begin{proofsketch}
		Towards establishing $v_i \prec f_i \prec v_{i + 1}$, one can show that $f_i\prec v_i$ would prevent~$f_i$ from seeing $v_{i+1}$ on any page: As $f_i\prec v_i$ implies $f_i \prec b_{i + 1} \prec v_{i + 1} \prec a_{i + 1}$ and we have the edge $b_{i+1}a_{i + 1}$ on any page except $p_d$, only visibility on page $p_d$ would still be possible.
		However, the edges on the page $p_d$ prevent visibility to $v_{i+1}$ for any spine position left of $v_i$.
		By symmetric arguments, we can obtain that $v_{i+1}\prec f_i$ would prevent $v_i$ from seeing~$f_i$. 
		Using again the fact that we have the edge $b_{i}a_{i}$ on any page except $p_d$, in concert with the relation $v_i \prec f_i \prec v_{i + 1}$ shown above and the edges $v_ia_{i}$ and $b_{i + 1}v_{i + 1}$ on the page~$p_d$, one can deduce that $\sigma(f_iv_i) = \sigma(f_{i}v_{i + 1}) = p_d$ must hold.
		Finally, the bound on the size of the gadget can be obtained by a close analysis of the construction.
	\end{proofsketch}
	\begin{prooflater}{plemmaFixationGadgetProperties}
		Let \lSL[G] be a solution to $\instance= \instanceLong$.
		First, we will show that $v_i \prec f_i \prec v_{i + 1}$ must hold.
		Towards a contradiction, assume that there exists a solution \lSL[G] with $f_i \prec v_i$.
		Observe that $f_i \prec v_i$ also implies $f_i \prec v_{i + 1}$ and recall that we have in $H$ the edges $e(b_{i + 1}, a_{i + 1}, p)$ for every page $p \neq p_d$.
		As $b_{i + 1} \prec_H v_{i + 1} \prec_H a_{i + 1}$ and \lSL[G] is an extension of \lSL[H], $f_{i}$ can see $v_{i + 1}$ only on the page $p_d$.
		Hence, we must have $\sigma(f_iv_{i + 1}) = p_d$.
		We will now distinguish between the following two cases.
		On the one hand, there could exist a $j$ with $1 \leq j < i \leq F$, such that $v_j \prec f_i \prec v_{j + 1} \prec v_{i + 1}$.
		However, since we have $\sigma_{H}(v_{j}v_{j + 1}) = p_d$, this cannot be the case, as this would introduce a crossing on the page~$p_d$ between the edges $f_iv_{i + 1}$ and $v_{j}v_{j + 1}$.
		On the other hand, we could have $f_{i} \prec b_1$, i.e.,~$f_i$ is placed at the beginning of the fixation gadget.
		Observe that we have in this situation $f_{i} \prec b_1 \prec v_{i + 1} \prec a_{F + 1}$ and $\sigma_{H}(b_1a_{F + 1}) = p_d$.
		Hence, we would introduce a crossing on the page $p_d$ between the edges $f_iv_{i + 1}$ and $b_1a_{F + 1}$.
		As $\sigma(f_{i}v_{i + 1}) = p_d$ can therefore not hold, we have no page to which we could assign $f_{i}v_{i + 1}$ without introducing a crossing, contradicting the assumption that we have a solution with $f_{i} \prec v_{i}$.
		As the arguments that exclude $v_{i + 1} \prec f_{i}$ are symmetric, we obtain that $v_{i} \prec f_{i} \prec v_{i + 1}$ must hold in any solution to \instance.
		
		Secondly, we will show that $\sigma(f_iv_i) = \sigma(f_{i}v_{i + 1}) = p_d$ holds.
		Towards a contradiction, assume that there exists a solution \lSL[G] with $\sigma(f_iv_i) \neq p_d$ for some $i \in [F]$.
		Hence, $\sigma(f_iv_i) = p$ holds for some page $p \neq p_d$.
		Recall that we have the edge $e(b_i, a_i, p)$ with $\sigma_{H}(e(b_i, a_i, p)) = p$.
		This allows us to strengthen $v_{i} \prec f_{i} \prec v_{i + 1}$, which we have shown before, to $v_i \prec f_i \prec a_i$ under the assumption of $\sigma(f_iv_i) = p$, as we would otherwise have a crossing {%
			with the edge $b_ia_i$} on the page $p$.
		However, then we can conclude from $\sigma_{H}(v_ia_i) = p_d$ and the existence of the edges $e(b_{i + 1}, a_{i + 1}, p')$ with $\sigma_{H}(e(b_{i + 1}, a_{i + 1}, p')) = p'$, for any page $p' \neq p_d$, and ${f_i  \prec a_i \prec b_{i + 1}}  \prec v_{i + 1} \prec a_{i + 1}$ that there does not exist a feasible page assignment for the edge $f_iv_{i + 1}$.
		This contradicts our assumption of a solution with $\sigma(f_iv_i) \neq p_d$ and a symmetric argument rules out any solution with $\sigma(f_iv_{i + 1}) \neq p_d$.
		
		Thirdly, we analyze the size of the fixation gadget.
		Recall that $\mathcal{F}$ consists of $F$ vertices, and we introduce $3(F + 1)$ vertices in $H$.
		Furthermore, the fixation gadget contributes one page to an (existing) stack layout \lSL of $H$ on $\ell - 1$ pages.
		Regarding the number of edges, we create in~$H$ $(\ell - 1)(F + 1)$ edges of the form $b_ia_i$, $i \in [F + 1]$, $2(F + 1)$ edges of the form $b_iv_i$ or $v_ia_i$, $i \in [F + 1]$,~$F$ edges of the form $v_iv_{i + 1}$, $i \in [F]$, and the edge $b_1a_{F + 1}$.
		Together with the $2F$ new edges that we add to $G$, this sums up to 
		$(\ell + 4)F + \ell + 2$ edges.
	\end{prooflater}
	\noindent
	\cref{lem:fixation-gadget-properties} tells us that we can restrict the feasible positions for $f_i$ to a pre-defined set of consecutive intervals by choosing suitable positions for $v_i$ and $v_{i + 1}$ in the spine order~$\prec_{H}$.
	As the fixation gadget requires an additional page $p_d$, we must ensure that the existence of the (otherwise mostly empty) page $p_d$ does not violate the semantics of our reductions. 
	In particular, we will (have to) ensure that our full constructions satisfy the following property.
	\begin{property}
		\label{prop:fixation-gadget-dummy-page}
		Let $\instance = \instanceLong$ be an instance of \SLEShort{} that contains a fixation gadget on~$F$ vertices $\{f_1, \ldots, f_F\}$.
		In any solution \lSL[G] to \instance{} and for every new edge $e \in \Eadd{}$ with $\sigma(e) =~p_d$, we have $e \in \{f_iv_i, f_{i}v_{i + 1} \mid i \in [F]\}$.
	\end{property}
	
	\subsection{The Complete Reduction}
	\label{sec:paranp-hardness-construction}
	\iftoggle{long}{
		Recall the base layout of our reduction that we described in \cref{sec:paranp-hardness-base-layout} and illustrated with \cref{fig:paranp-hardness-base-layout-edges}.
		There, we created, for a given formula $\varphi$, one vertex for each variable $x_i$ and each clause~$c_j$.
		Furthermore, each $x_i$ should only be visible on two pages that correspond to its individual truth states and each $c_j$ should only be visible on three pages that correspond to the complementary literals in $c_j$.
		However, the intended semantics of our reduction rely on the assumption that new vertices can only be placed on a specific position, and in a specific order, on the spine.
		Equipped with the fixation gadget, we will now satisfy this assumption.
		
		For an instance $\varphi = (\mathcal{X}, \mathcal{C})$ of \ThreeSat, we take the base layout of our reduction as described in \cref{sec:paranp-hardness-base-layout} and incorporate in $H$ a fixation gadget on two vertices, i.e., for $F = 2$.
		We set $a_3 \prec d_1$, i.e., we place the fixation gadget at the beginning of the spine, and identify $s = f_1$ and $v = f_2$.
		Furthermore, we add the edge $d_1d_{N + M + 1}$ and set $\sigma(d_1d_{N + M + 1}) = p_d$.
		Observe that this ensures that our construction will have \cref{prop:fixation-gadget-dummy-page}, as this edge prevents connecting~$s$ with $x_i$ or $v$ with $c_j$ on page $p_d$ for any $i \in [N]$ and $j \in [M]$.
		Finally, we add to~$G$ the new edges $sx_i$ and $vc_j$ for every $i \in [N]$ and $j \in [M]$.
		This completes our reduction, and we establish with the following theorem its correctness; see also \cref{fig:paranp-hardness-example} for an example of our construction for a small formula.
	}{   
		\begin{figure}[t]
			\centering
			\includegraphics[page=6]{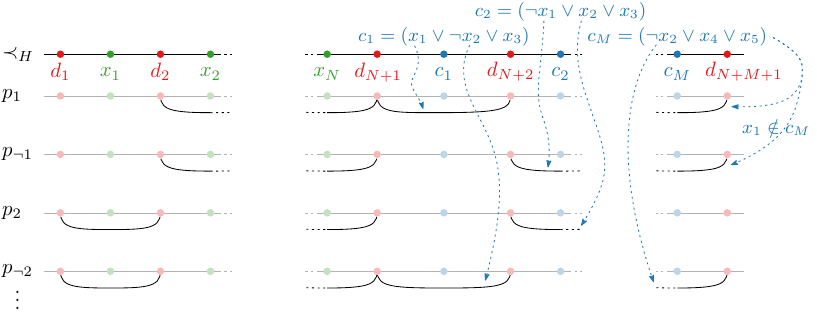}
			\caption{An overview of the created vertices and edges in our reduction. Green vertices represent variables, blue vertices clauses, and red vertices the dummy vertices $d_q$. Furthermore, we visualize some of the edges in $H$ that are created for the variable-vertices (left) and clause-vertices (middle and right) to block visibility on the respective pages.
				If an edge is created due to the (non-)existence of a literal in the clauses $c_1$, $c_2$, or $c_M$ it is indicated via a blue arc.}
			\label{fig:paranp-hardness-base-layout-edges-no-labels}
		\end{figure}
		\sectionParaNPIntuition
	}
	\begin{figure}[t]
		\centering
		\includegraphics[page=2]{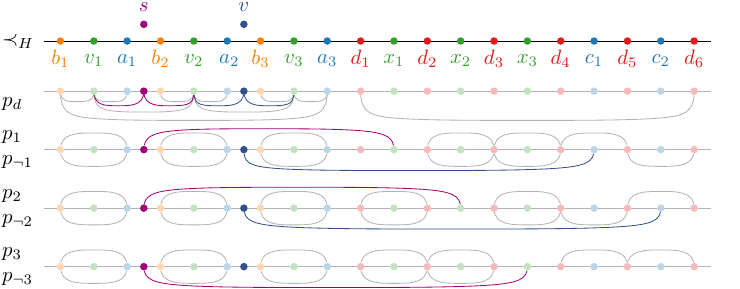}
		\caption{An example of our reduction for the formula $\varphi$ consisting of the clauses $c_1 =\left( x_1 \lor \lnot x_2 \lor x_3 \right)$ and $c_2 = \left(\lnot x_1 \lor x_2 \lor x_3\right)$. The extension indicated in saturated colors induces the truth assignment $\Gamma(x_1) = \Gamma(x_2) = 1$ and $\Gamma(x_3) = 0$, which satisfies $\varphi$.}
		\label{fig:paranp-hardness-example}
	\end{figure}
	\begin{restatable}\restateref{thm:paranp-hardness}{theorem}{theoremParaNPHardness}
		\label{thm:paranp-hardness}
		\SLEShort{} is \NP-complete even if we have just two new vertices and $\EaddH{} = \emptyset$.
	\end{restatable}
	\begin{prooflater}{ptheoremParaNPHardness}
		\NP-membership of \SLEShort{} follows immediately from the fact that we can encode a solution \lSL in \BigO{\Size{\instance}} space and can verify it in polynomial time.
		Thus, we focus in the remainder of the proof on showing \NP-hardness of \SLEShort{}.
		
		Let $\varphi = (\mathcal{X}, \mathcal{C})$ be an instance of \ThreeSat{} and let $\instance = \instanceLong$ be the obtained instance of \SLEShort{}.
		The number of vertices and edges in $H$ that we created in the base layout of \cref{sec:paranp-hardness-base-layout} is in \BigO{N + M}.
		Furthermore, in the base layout, we used $2N$ pages.
		Together with the page from the fixation gadget, we have $\ell = 2N + 1$.
		As $F = 2$ is constant, the contribution of the fixation gadget to the size of $H$ and $G$ is linear in $\varphi$ by \cref{lem:fixation-gadget-properties}.
		Hence, the overall size of \instance{} is polynomial in the size of $\varphi$.
		Clearly, the time required to create~\instance{} is also polynomial in the size of $\varphi$ and it remains to show the correctness of the reduction.

		\proofsubparagraph*{($\boldsymbol{\Rightarrow}$)}
		Assume that $\varphi$ is a positive instance of \ThreeSat{} and let $\Gamma: \mathcal{X} \to \{0,1\}$ be a truth assignment that satisfies every clause.
		To show that \instance{} is a positive instance of \SLEShort, we create a stack layout \lSL of $G$.
		We first ensure that it extends \lSL[H] by copying \lSL[H].
		Next, we set $a_1 \prec s \prec b_2$, $a_2 \prec v \prec b_3$ and $\sigma(sv_1) = \sigma(sv_2) = \sigma(vv_2) = \sigma(vv_3) = p_d$ as shown in \cref{fig:paranp-hardness-example}.
		Furthermore, for every variable $x_i \in \mathcal{X}$, we set $\sigma(sx_i) = p_i$ if $\Gamma(x_i) = 1$ and $\sigma(sx_i) = p_{\lnot i}$ otherwise.
		For every clause $c_j \in \mathcal{C}$, we identify a variable $x_i$ that satisfies the clause $c_j$ under $\Gamma$.
		As $\Gamma$ satisfies every clause, the existence of such an $x_i$ is guaranteed.
		Then, we set $\sigma(sx_i) = p_i$ if $\Gamma(x_i) = 0$ and $\sigma(sx_i) = p_{\lnot i}$ otherwise.
		
		To show that \lSL is crossing-free, we first observe that for the fixation gadget, our generated solution satisfies the necessary properties stated in \cref{lem:fixation-gadget-properties}.
		We can observe that in our page assignment, new edges cannot cross old edges.
		Hence, we only have to ensure that no two new edges cross.
		No two new edges on page $p_d$ can cross, so assume that there is a crossing on page $p_i$ for some $i \in [N]$.
		We observe that only edges of the form $sx_i$ and $vc_j$ for some $x_i \in \mathcal{X}$ and $c_j \in \mathcal{C}$ can cross, as they are otherwise incident to the same vertex.
		As there is a crossing on the page $p_i$, we must have by construction that the variable $x_i$ appears negated in the clause $c_j$ but we have $\Gamma(x_i) = 1$.
		Hence, $x_i$ does not satisfy $c_j$, which is a contradiction to our construction of \lSL, for which we only considered variables that satisfy the clause $c_j$.
		Therefore, a crossing on the page $p_i$ cannot exist.
		As the argument for a crossing on page $p_{\lnot i}$ is symmetric, we conclude that \lSL must be crossing-free and hence witnesses that \instance{} is a positive instance of \SLEShort{}.
		
		\proofsubparagraph*{($\boldsymbol{\Leftarrow}$)}
		Assume that \instance{} is a positive instance of \SLEShort{}.
		This implies that there exists a witness extension \lSL[G] of \lSL[H].
		As \instance{} contains the fixation gadget, we can apply \cref{lem:fixation-gadget-properties} and deduce that $s \prec v \prec d_1$ holds.
		Based on \lSL[G], we now construct a truth assignment $\Gamma\colon \mathcal{X} \to~\{0,1\}$ for $\varphi$.
		For each variable $x_i \in \mathcal{X}$, we consider the page assignment $\sigma(sx_i)$.
		Recall that we have $s \prec d_1 \prec x_i \prec d_{N + M + 1}$ and $\sigma_{H}(d_1d_{N + M + 1}) = p_d$.
		Together with $s \prec d_i \prec x_i \prec d_{i + 1}$ and $\sigma_{H}(d_id_{i + 1}) = p$ for any page $p \in[\ell] \setminus \{p_d, p_i, p_{\lnot i}\}$, we conclude that $\sigma(sx_i) \in \{p_i, p_{\lnot i}\}$ must hold.
		We set $\Gamma(x_i) = 1$ if $\sigma(sx_i) = p_i$ holds and $\Gamma(x_i) = 0$ if $\sigma(sx_i) = p_{\lnot i}$ holds and know by the above arguments that $\Gamma$ is well-defined.
		What remains to show is that $\Gamma$ satisfies $\varphi$.
		Let $c_j$ be an arbitrary clause of $\varphi$ and consider the page~$p$ with $\sigma(vc_j) = p$.
		As we have $\sigma_{H}(d_1d_{N + M + 1}) = p_d$ we know that $p \neq p_d$ must hold, i.e., the page~$p$ is associated to some variable $x_i$.
		For the remainder of the proof, we assume $p = p_i$ as the case $p = p_{\lnot i}$ is symmetric.
		Recall that we have $d_{N + j} \prec_{H} c_j \prec_{H} d_{N + j + 1}$, $\sigma_{H}(e(c_j, p_{i'})) = p_{i'}$, and $\sigma_{H}(e(c_j, p_{\lnot i'})) = p_{\lnot i'}$ for any variable $x_{i'} \in \mathcal{X}$ that does not occur in~$c_j$.
		Hence, we know that $x_i$ must occur in $c_j$.
		Furthermore, the same reasoning allows us to conclude that~$x_i$ must appear negated in $c_j$, as we would otherwise have a crossing.
		Using $s \prec v \prec x_i \prec c_j$ and our assumption of $\sigma(vc_j) = p_i$, we derive that $x_i$ will be visible from $s$ on page $p_{\lnot i}$ only, i.e., we have $\sigma(sx_i) = p_{\lnot i}$.
		By our construction of $\Gamma$, we conclude $\Gamma(x_i) = 0$.
		Hence, $x_i$ satisfies $c_j$ under $\Gamma$.
		As $c_j$ was selected arbitrarily, this holds for all clauses and therefore must $\Gamma$ satisfy the whole formula $\varphi$, i.e., it witnesses that~$\varphi$ is a positive instance of \ThreeSat.
	\end{prooflater}
	\noindent
	Finally, we want to remark that \cref{thm:paranp-hardness} is tight in the sense that \SLEShort with only one new vertex $v$ and $\EaddH{} = \emptyset$ can be solved in polynomial time.
	{%
		To this end, we can branch over all $\Size{V(H)} + 1 = \Size{V(G)}$ possible spine positions where $v$ can be placed.}
	For each of these, the observation that edges incident to the same vertex can never cross each other allows us to greedily assign a new edge $uv$ to the first page $p$ where $v$ can see $u$.
	Recall that we only add one new vertex~$v$.
	Hence, $u$ is an old vertex whose spine position is known.
	Clearly, an extension exists if and only if there exists a spine position for $v$ such that our greedy page assignment can find a page for all new edges.
	\begin{remark}
		\label{rem:sle-one-new-e-add-h-empty-poly}
		Let $\instance = \instanceLong$ be an instance of \SLEShort{} with $\nadd = 1$ and $\EaddH{} = \emptyset$.
		We can find an $\ell$-page stack layout of $G$ that extends \lSL or report that none exists in 
		{%
			$\BigO{\Size{V(G)} \cdot \madd \cdot \Size{\instance{}}}$ time, where $\madd$ denotes the number of new edges.}
	\end{remark}

	\section{\textsc{SLE} Parameterized by Missing Vertices and Edges Is in \textsf{XP}}
	\label{sec:xp}
	In the light of \cref{thm:paranp-hardness}, which excludes the use of the vertex+edge deletion distance as a pathway to tractability, we consider parameterizing by the total number of missing vertices and edges $\kappa\coloneqq\nadd+\madd$. 
	As our first result in this direction, we show that %
	parameterizing \SLEShort{} by~$\kappa$ makes it \XP-tractable.
	To this end, %
	we combine a branching procedure with the %
	fixed-parameter algorithm for the special case obtained in \cref{thm:only-edges-fpt}.
	\begin{theorem}
		\label{thm:xp-kappa}
		Let $\instance = \instanceLong$ be an instance of \SLEShort{}.
		We can find %
		an $\ell$-page stack layout of $G$ that extends \lSL or report that none exists in %
		{%
			$\BigO{\Size{V(G)}^{\nadd}\cdot{\madd}^{\madd}\cdot\Size{\instance}}$} time{%
			, where $\nadd$ and $\madd$ denote the number of new vertices and edges, respectively.}
	\end{theorem}
	\begin{proof}%
		We branch over the possible assignments of new vertices to the intervals in $\prec_{H}$.
		As a solution could assign multiple vertices to the same interval, %
		we also branch over the order in which the new vertices will appear in the spine order $\prec_{G}$.
		Observe that $\prec_{H}$ induces $\Size{V(H)} + 1$ different intervals, out of which we have to choose $\nadd{}$ with repetition.
		{%
			Recall that $\binom{n + k - 1}{k}$ denotes the number of ways one can choose $k$ out of $n$ elements with repetition.
			Plugging in the correct values for $n$ and $k$,%
		}
		together with the possible orders of the new vertices, we can bound the number of branches by $\nadd{}!\cdot\binom{\Size{V(H)} + \nadd{}}{\nadd{}}$. %
		We can simplify this expression to
		\begin{align*}
			\frac{\nadd{}!\cdot (\Size{V(H)} + \nadd{})!}{\nadd{}!\cdot ((\Size{V(H)} + \nadd{}) - \nadd{})!} &= \frac{(\Size{V(H)} + \nadd{})!}{\Size{V(H)}!}\\
			&= \Pi_{i = 1}^{\nadd{}}(\Size{V(H)} + i) \in  \BigO{\Size{V(G)}^{\nadd{}}}.
		\end{align*}
		In each branch, the spine order $\prec_{G}$ is fixed and extends $\prec_{H}$.
		Hence, it only remains to check whether $\prec_{G}$ allows for a valid page assignment $\sigma_G$.
		As each branch corresponds to an instance %
		of \SLEShort{} where only edges are missing,
		we use \cref{thm:only-edges-fpt} to check in \BigO{{\madd{}}^{\madd{}}\cdot\Size{\instance}} time whether such an assignment $\sigma_G$ exists.
		The overall running time now follows readily.
	\end{proof}
	\noindent
	The running time stated in \cref{thm:xp-kappa} not only proves that \SLEShort{} is in \XP{} when parameterized by~$\kappa$, but also \FPT{} when parameterized by $\madd{}$ for constant $\nadd{}$.
	However, common complexity assumptions rule out an efficient algorithm parameterized by $\kappa$, as we show next.

	\section{\textsc{SLE} Parameterized by Missing Vertices and Edges Is \textsf{W}[1]-hard}
	\label{sec:w-1}
	\newcommand{\statemccXiViProperty}{\begin{restatable}{property}{mccXiViProperty}
			\label{property:w-1-x-i-v-i}
			In a solution \lSL to \SLEShort{} we have $u_{\alpha}^0 \prec x_{\alpha} \prec u_{\alpha + 1}^0$ for every~$\alpha \in [k]$.
	\end{restatable}}
	\newcommand{\statemccLayerProperty}{\begin{restatable}{property}{mccLayerProperty}
			\label{property:w-1-layer}
			Let \lSL be a solution to an instance of \SLEShort{} that fulfills \cref{property:w-1-x-i-v-i} and for which we have $e = v_{\alpha}^iv_{\beta}^{j} \in E(G_C)$, $1 \leq \alpha < \beta \leq k$, and $x_{\alpha}, x_{\beta} \in \mathcal{X}$.
			If $\sigma(x_{\alpha}x_{\beta}) = p_e$ then~$x_{\alpha}$ is in $\intervalPlacing{v_{\alpha}^i}$ and $x_{\beta}$ is in $\intervalPlacing{v_{\beta}^j}$.
	\end{restatable}}
	\newcommand{\plotWOneLayerFigure}{\begin{figure}
			\centering
			\includegraphics[page=5]{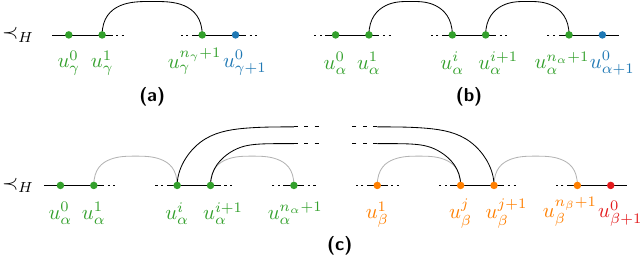}
			\caption{Edges of $H$ that model the adjacency given by the edge $e = v_{\alpha}^{i}v_{\beta}^j \in E(G_C)$. All of these edges are placed on the page $p_e$. Intuitively, we span the intervals induced~\textbf{\textsf{(a)}} by all vertices for each color $\gamma \in [k] \setminus \{\alpha, \beta\}$ %
				and~\textbf{\textsf{(b)}} by vertices of the colors $\alpha$ and $\beta$ %
				that are not incident to $e$, here visualized for the color $\alpha$.~\textbf{\textsf{(c)}} Furthermore, we create a tunnel that connects $\intervalPlacing{v_{\alpha}^{i}}$ with $\intervalPlacing{v_{\beta}^{j}}$. The gray edges in~\textbf{\textsf{(c)}} are from~\textbf{\textsf{(a)}} and~\textbf{\textsf{(b)}}.}
			\label{fig:w-1-layer}
	\end{figure}}
	
	In this section, we show that \SLEShort{} parameterized by the number $\kappa$ of missing vertices and edges is \W[1]-hard.
	To show \W[1]-hardness, we reduce from the \probname{Multi-colored Clique} (\MCC{}) problem.
	Here, we are given a graph $G_C$, an integer $k > 0$, and a partition of~$V(G_C)$ into $k$ independent subsets $V_1, \ldots, V_{k}$, and ask whether there exists a \emph{colorful $k$-clique} $\mathcal{C} \subseteq V(G_C)$ in $G_C$, i.e., a clique on $k$ vertices that contains exactly one vertex of every set $V_i$, $i \in [k]$.
	It is known that \MCC{} is \W[1]-hard when parameterized by $k$~\cite{Cygan2015}.
	In the following, we will use Greek letters for the indices of the partition and denote with $n_{\alpha}$ the number of vertices in $V_{\alpha}$, i.e., $n_{\alpha} = \Size{V_{\alpha}}$.
	Observe that $\sum_{\alpha \in [k]} n_{\alpha} = N$ with $N = \Size{V(G_C)}$.
	As we can interpret the partitioning of the vertices into $V_1, \ldots, V_{k}$ as assigning to them one of $k$ colors, we will call a vertex~$v_{\alpha}^i\ {%
		\in V_{\alpha}}$ %
	with $\alpha \in [k]$ and $i \in [n_{\alpha}]$ a vertex with \emph{color} $\alpha$.
	Our construction will heavily use the notion of a successor and predecessor of a vertex in a given spine order~$\prec$.
	For a vertex~$u$, the function $\successorSpine{\prec}{u}$ returns the \emph{successor} of $u$ in the spine order $\prec$, %
	i.e., the consecutive vertex in $\prec$ after $u$.
	Note that $\successorSpine{\prec}{u}$ is undefined if there is no  vertex $v \in V(G)$ with $u \prec v$.
	We write $\successor{u}$ if~$\prec$ is clear from context.
	The \emph{predecessor} function $\predecessorSpine{\prec}{u}$ is defined analogously.
	\begin{statelater}{mccIntro}
		
		Let $(G_C, k, (V_1, \ldots, V_{k}))$ be an instance of \MCC{}.
	\end{statelater}%
	\iftoggle{long}{%
		Our instance contains for every \emph{original} vertex $v_{\alpha}^i \in V(G_C)$ a \emph{copy} $u_{\alpha}^i \in V(H)$.
		Furthermore, we add for each color $\alpha \in [k]$ to $H$ two additional vertices and, overall, three further dummy vertices that we use to ensure correctness of the reduction.
		We place the vertices on the spine based on their color $\alpha$ and index $i$; see \cref{fig:w-1-base-layout} and \cref{sec:w-1-base}, where we give the full details of the base layout.
		Observe that every vertex $v_{\alpha}^i \in V(G_C)$ \emph{induces} the interval $[u_{\alpha}^i, u_{\alpha}^{i + 1}]$ in $\prec_H$, which we denote with \intervalPlacing{v_{\alpha}^i}.
		The equivalence between the two problems will be obtained by adding a $k$-clique to $G$ that consists of the $k$ new vertices $\mathcal{X} = \{x_1, \ldots, x_{k}\}$.
		Placing $x_{\alpha} \in \mathcal{X}$ in \intervalPlacing{v_{\alpha}^i} indicates that $v_{\alpha}^i$ will be part of the colorful $k$-clique in~$G_C$.

		To establish the correctness of our reduction, we have to ensure two things.
		First, we have to model the adjacencies in $G_C$.
		In particular, two new vertices $x_{\alpha}$ and $x_{\beta}$, with $\alpha < \beta$, should only be placed in intervals induced by vertices adjacent in $G_C$.
		We enforce this by adding for every edge $e = v_{\alpha}^iv_{\beta}^j \in E(G_C)$ a page $p_e$ that contains a set of edges creating a \emph{tunnel} on~$p_e$, see \cref{fig:w-1-layer}, and thereby allowing us to place the edge $x_{\alpha}x_{\beta} \in E(G)$ in the page $p_e$ if and only if $x_{\alpha}$ is placed in \intervalPlacing{v_{\alpha}^i} and~$x_{\beta}$ in \intervalPlacing{v_{\beta}^j}. 
		Hence, the page assignment verifies that only pairwise adjacent vertices are in %
		the solution, i.e., new vertices can only be placed in intervals induced by a clique in $G_C$.
		We %
		describe the tunnel further in \cref{sec:w-1-layer}.
		
		Second, we have to ensure that we select exactly one vertex $v_{\alpha}^i \in V_{\alpha}$ for every color $\alpha \in [k]$.
		In particular, the new vertex $x_{\alpha}$ should only be placed in intervals that are induced by vertices from~$V_{\alpha}$.
		To this end, we modify $H$ to include an appropriate fixation gadget by %
		re-using some vertices of the base layout; see \cref{sec:w-1-fixation-gadget} for details.
		As the whole base layout thereby forms the fixation gadget, our construction trivially satisfies \cref{prop:fixation-gadget-dummy-page}.
		
		The above two ideas are formalized in \cref{property:w-1-x-i-v-i,property:w-1-layer}.
		With these properties at hand, we show at the end of \cref{sec:w-1-correctness} that \SLEShort{} is \W[1]-hard when parameterized by $\kappa$.
		As in the reduction from \cref{sec:paranp-hardness}, we will allow multi-edges in the graph $H$ to facilitate presentation and understanding.
		In \cref{app:multi-edges-w-1} we will discuss a way to remove the multi-edges by distributing the individual edges over auxiliary vertices.
		
	}{
		In the following, we first give an overview of and intuition behind our reduction in \cref{sec:w-1-overview}, before we show its correctness in \cref{sec:w-1-correctness}.
		Note that the full details of the construction can be found in the full version~\cite{ARXIV}.
		Furthermore, as in the reduction from \cref{sec:paranp-hardness}, we will allow multi-edges in the graph $H$ to facilitate the presentation of the reduction. The procedure for removing multi-edges by distributing the individual edges over auxiliary vertices is also detailed in the full version~\cite{ARXIV}.
		
		\subsection{An Overview of the Construction}
		\label{sec:w-1-overview}
		\mccIntro
		
		First, we define the base layout of our reduction.
		In the base layout, we create the $N + 2k + 3$ vertices $\{u_{\alpha}^j \mid {\alpha} \in [k],\ j \in [n_{\alpha} + 1]_0\} \cup \{u_0^0, u_{\bot}^0, u_{\bot}^1\}$ in $H$.
		Note that for each \emph{original} vertex $v_{\alpha}^i \in V(G_C)$, we have a \emph{copy} $u_{\alpha}^i$.
		We will refer to the vertices $u_0^0$, $u_{\bot}^0$, and $u_{\bot}^1$ as \emph{dummy vertices} and set, for ease of notation, $\bot = k + 1$ and $n_{\bot} = 1$.
		The vertices are placed on the spine based on their color $\alpha$ and index $i$; see \cref{fig:w-1-base-layout}.
		Finally, observe that $\successor{u_{\alpha}^i} =  u_{\alpha}^{i + 1}$ for every $v_{\alpha}^i \in V(G_C)$.
		Furthermore, every vertex $v_{\alpha}^i \in V(G_C)$ \emph{induces} the interval $[u_{\alpha}^i, u_{\alpha}^{i + 1}]$ in $\prec_H$, which we denote with~\intervalPlacing{v_{\alpha}^i}.
		The equivalence between the two problems will be obtained by adding a $k$-clique to $G$ that consists of the $k$ new vertices $\mathcal{X} = \{x_1, \ldots, x_{k}\}$.
		Placing $x_{\alpha} \in \mathcal{X}$ in \intervalPlacing{v_{\beta}^i} indicates that $v_{\beta}^i$ will be part of the colorful $k$-clique in~$G_C$, i.e., we will have the equivalence $u_{\alpha}^i \prec x_{\alpha} \prec 
		\successorSpine{\prec_{H}}{u_{\alpha}^i}\ 
		\overset{\text{p.~d.}}{\Longleftrightarrow}\ x_{\alpha}\ \text{is placed in}\ \intervalPlacing{v_{\alpha}^i}\ \Longleftrightarrow\ v_{\alpha}^i \in \mathcal{C}$ between a solution \lSL[G] to \SLEShort{} and a solution $\mathcal{C}$ to \MCC.
		To guarantee that $\mathcal{C}$ is colorful, i.e., contains exactly one vertex from each color, we will ensure the following property with our construction.%
		\statemccXiViProperty
		
		\begin{figure}
			\centering
			\includegraphics[page=2]{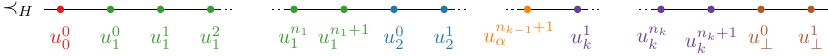}
			\caption{The base layout of our reduction. We use colors to additionally differentiate vertices that originate from different vertex sets $V_{\alpha}$, for $\alpha \in [k]$, and the dummy vertices $u_0^0$, $u_{\bot}^0$ and $u_{\bot}^1$.}
			\label{fig:w-1-base-layout}
		\end{figure}

		To establish the correctness of our reduction, we have to ensure two things.
		First, we have to model the adjacencies in $G_C$.
		In particular, two new vertices $x_{\alpha}$ and $x_{\beta}$, with $\alpha < \beta$, should only be placed in intervals induced by vertices adjacent in $G_C$.
		We enforce this by adding for every edge $e = v_{\alpha}^iv_{\beta}^j \in E(G_C)$ a page $p_e$.
		On this page $p_e$, we create the following edges in $H$; see also \cref{fig:w-1-layer} for a visualization.
		Firstly, we create for every color $\gamma \in [k] \setminus \{\alpha, \beta\}$ an edge that spans exactly the intervals induced by vertices of color $\gamma$, thereby intuitively blocking visibility to any interval induced by a vertex of a color different to $\alpha$ and $\beta$; see \cref{fig:w-1-layer}a.
		Secondly, we create up to two edges that span all intervals induced by vertices of color $\alpha$ except~\intervalPlacing{v_{\alpha}^i}; see \cref{fig:w-1-layer}b.
		We do so similarly for color $\beta$.
		These edges in concert with a \emph{tunnel} that we create on page~$p_e$, see \iftoggle{long}{\cref{fig:w-1-layer}c and \cref{sec:w-1-layer} for a formal description}{\cref{fig:w-1-layer}c}, allow us to place the edge $x_{\alpha}x_{\beta} \in E(G)$ in the page~$p_e$ if and only if $x_{\alpha}$ is placed in \intervalPlacing{v_{\alpha}^i} and~$x_{\beta}$ in \intervalPlacing{v_{\beta}^j}.
		More formally, our construction will ensure the following property.
		\statemccLayerProperty
		\plotWOneLayerFigure

		Second, we have to ensure that we select exactly one vertex $v_{\alpha}^i \in V_{\alpha}$ for every color $\alpha \in [k]$.
		In particular, the new vertex $x_{\alpha}$ should only be placed in intervals that are induced by vertices from $V_{\alpha}$.
		To this end, we modify $H$ to include a fixation gadget on $F = k$ vertices by %
		re-using some vertices of the base layout.
		Most importantly, we identify $v_{\alpha} = u_{\alpha}^0$ for every $\alpha \in[k + 1]$ and $f_{\alpha} = x_{\alpha}$ for every $\alpha \in[k]$; see \iftoggle{long}{\cref{sec:w-1-fixation-gadget}}{the full version~\cite{ARXIV}} for details.
		As the whole base layout thereby forms the fixation gadget, our construction trivially satisfies \cref{prop:fixation-gadget-dummy-page}.

		\iftoggle{long}{
			Throughout the description of the reduction, we explicitly highlight crucial places where our construction should fulfill certain properties to ensure its correctness.
			While, at the time of stating the property, our construction might not yet fulfill it, we argue in \cref{sec:w-1-correctness}\iftoggle{long}{}{ (and \cref{app:w-1-properties})} that it in the end indeed has the desired properties.
			With these properties at hand, we show at the end of \cref{sec:w-1-correctness} that \SLEShort{} is \W[1]-hard when parameterized by $\kappa$.
		}{}
		
	}

	\begin{statelater}{sectionWOneBaseLayout}
		\subsection{Creating Intervals on the Spine: Our Base Layout}
		\label{sec:w-1-base}

		\iftoggle{long}{\begin{figure}
				\centering
				\includegraphics[page=2]{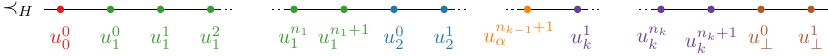}
				\caption{The base layout of our reduction. We use colors to additionally differentiate vertices that originate from different vertex sets $V_{\alpha}$, for $\alpha \in [k]$, and the dummy vertices $u_0^0$, $u_{\bot}^0$ and $u_{\bot}^1$.}
				\label{fig:w-1-base-layout}
		\end{figure}}{}
		
		Recall that $G_C$ has the vertex set $V(G_C) = \{v_1^1, \ldots, v_{k}^{n_{k}}\}$ partitioned into $V_1, \ldots, V_{k}$ with $V_{\alpha} = \{v_{\alpha}^1, \ldots, v_{\alpha}^{n_{\alpha}}\}$ for ${\alpha} \in [k]$.
		We create the $N + 2k + 3$ vertices $\{u_{\alpha}^j \mid {\alpha} \in [k],\ j \in [n_{\alpha} + 1]_0\} \cup \{u_0^0, u_{\bot}^0, u_{\bot}^1\}$ in $H$.
		Note that for each \emph{original} vertex $v_{\alpha}^i \in V(G_C)$, we have a \emph{copy} $u_{\alpha}^i$.
		We will refer to the vertices $u_0^0$, $u_{\bot}^0$, and $u_{\bot}^1$ as \emph{dummy vertices} and set, for ease of notation, $\bot = k + 1$ and %
		{%
			$n_{\bot} = 0$.}
		We order the vertices of $H$ on the spine by setting $u_{\alpha}^i \prec u_{\alpha}^{i + 1}$ and $u_{\alpha}^{n_{\alpha} + 1} \prec u_{\alpha+1}^0$ for every $\alpha \in [k]$ and $i \in [n_{\alpha}]_0$.
		Furthermore, we set $u_{0}^0 \prec u_{1}^{0}$ and $u_{\bot}^0 \prec u_{\bot}^{1}$.
		The spine order is then the transitive closure of the above partial orders.
		We visualize it in \cref{fig:w-1-base-layout} and observe that $\successor{u_{\alpha}^i} =  u_{\alpha}^{i + 1}$ for every $v_{\alpha}^i \in V(G_C)$.

		As already indicated, we define the set $\mathcal{X} = \{x_1, \ldots, x_{k}\}$ to contain $k$ (new) vertices, which we add to $G$.
		Furthermore, we form a $k$-clique on $\mathcal{X}$, i.e., we add the edges $x_{\alpha}x_{\beta}$ for {%
			all} $1 \leq \alpha < \beta \leq k$ to $E(G)$.
		Recall that \intervalPlacing{v_{\alpha}^i} denotes the interval $[u_{\alpha}^i, \successorSpine{\prec_{H}}{u_{\alpha}^i}]$.
		We use the following equivalence between a solution \lSL[G] to \SLEShort{} and a solution $\mathcal{C}$ to \MCC.
		\begin{align}
			\label{eq:w-1-equivalence}
			u_{\alpha}^i \prec x_{\alpha} \prec \successorSpine{\prec_{H}}{u_{\alpha}^i}\ \overset{\text{per~def.}}{\Longleftrightarrow}\ x_{\alpha}\ \text{is placed in}\ \intervalPlacing{v_{\alpha}^i}\ \Longleftrightarrow\ v_{\alpha}^i \in \mathcal{C}
		\end{align}
		To guarantee that $\mathcal{C}$ is colorful, i.e., contains exactly one vertex from each color, we will ensure the following property with our construction.
		\iftoggle{long}{\statemccXiViProperty}{\mccXiViProperty*}
		\iftoggle{long}{\noindent
			Of course, our construction does not yet fulfill \cref{property:w-1-x-i-v-i}.
			We will show in \cref{sec:w-1-correctness} that the finished construction indeed does fulfill \cref{property:w-1-x-i-v-i}.}{}
		
	\end{statelater}
	
	\begin{statelater}{sectionWOneEdges}
		\subsection{Creating One Page per Edge: Encoding the Adjacencies}
		\label{sec:w-1-layer}
		Having fixed the base order on the spine, we now ensure that we only select vertices that are adjacent in $G_C$, i.e., we encode the edges of $G_C$ in our stack layout \lSL[H].
		Let $e = v_{\alpha}^{i}v_{\beta}^{j}$ be an edge of $G_C$ and recall that, by our assumption, we have $\alpha \neq \beta$.
		Furthermore, we assume for ease of presentation that $\alpha < \beta$ holds, which implies $u_{\alpha}^{i} \prec u_{\beta}^{j}$.
		We create the following edges in $H$; see \cref{fig:w-1-layer}.
		Note that each of them is assigned to the page $p_e$ in~$\sigma$, where $p_e$ is a new empty page that we associate with the edge $e$.
		Firstly, we create the edge $u_{\gamma}^{1}u_{\gamma}^{n_{\gamma} + 1}$ for each $\gamma \in [k] \setminus \{\alpha, \beta\}$; see \cref{fig:w-1-layer}a.
		Secondly, we create the edges $u_{\alpha}^{1}u_{\alpha}^{i}$ and $u_{\alpha}^{i + 1}u_{\alpha}^{n_{\alpha} + 1}$ as shown in \cref{fig:w-1-layer}b.
		Similarly, we add the edges $u_{\beta}^{1}u_{\beta}^{j}$ and $u_{\beta}^{j + 1}u_{\beta }^{n_{\beta} + 1}$.
		If $i \in \{1, n_{\alpha}\}$ or $j \in \{1, n_{\beta}\}$, we omit creating the respective edge to not introduce self-loops in~$H$.
		Thirdly, we create the edges $u_{\alpha}^{i}u_{\beta}^{j + 1}$ and $u_{\alpha}^{i + 1}u_{\beta}^{j}$, which we mark in \cref{fig:w-1-layer}c in black.
		\iftoggle{long}{\plotWOneLayerFigure}{}
		One can readily verify that these edges in the page $p_e$ do not cross.
		Furthermore, observe that the edges $u_{\alpha}^{i}u_{\beta}^{j + 1}$ and $u_{\alpha}^{i + 1}u_{\beta}^{j}$ create a \emph{tunnel} on the page $p_e$ connecting~$\intervalPlacing{v_{\alpha}^{i}}$ and~$\intervalPlacing{v_{\beta}^{j}}$.
		Intuitively, the edges on the page $p_e$ ensure that if $x_i$ is placed in \intervalPlacing{v_{\alpha}^i} then $x_i$ sees {%
			on page $p_e$ only vertices that are placed in the interval}~\intervalPlacing{v_{\beta}^j}.
		More formally, we will have the following property.
		\iftoggle{long}{\statemccLayerProperty}{\mccLayerProperty*}
	\end{statelater}
	
	\begin{statelater}{sectionWOneFixationGadget}
		\subsection{Restricting the Placement of New Vertices: Encoding the Colors}
		\label{sec:w-1-fixation-gadget}
		Until now, the new vertices in $\mathcal{X}$ can be placed in any interval on the spine.
		However, in \MCC{} we must select exactly one vertex from each color.
		Recall \cref{property:w-1-x-i-v-i}, which intuitively states that each new vertex should only be placed in intervals that correspond to its color.
		We now use the fixation gadget to ensure that our construction fulfills \cref{property:w-1-x-i-v-i}.
		Observe that we already introduced in \cref{sec:w-1-base} the required vertices of the fixation gadget when creating the base layout of our reduction.
		More specifically, we (implicitly) create in $H$ and our construction the fixation gadget on $F = k$ vertices $\mathcal{F} = \mathcal{X}$ by identifying, for $\alpha\in [k + 1]$, $v_{i} = u_{\alpha}^0$, $b_{i} = \predecessor{u_{\alpha}^0}$, and $a_{i} = \successor{u_{\alpha}^0}$, where we use $i = \alpha$ to differentiate between the vertices from the fixation gadget and the graphs $G_C$ and $H$.
		We also introduce the corresponding edges of the fixation gadget, which we visualize for the vertices in this construction in \cref{fig:w-1-fixation-gadget}.
		\begin{figure}
			\centering
			\includegraphics[page=5]{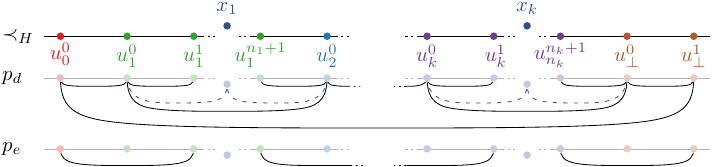}
			\caption{We use a fixation gadget to encode the color constraints. This figure is an adapted version of \cref{fig:fixation-gadget-example} that shows the fixation gadget on $k$ vertices with the vertices used in this reduction. We visualize the dummy page $p_d$ and a page $p_e$ created for an edge $e \in E(G_C)$.}
			\label{fig:w-1-fixation-gadget}
		\end{figure}
		Recall that when introducing the fixation gadget in \cref{sec:fixation-gadget}, we required that our instance must fulfill \cref{prop:fixation-gadget-dummy-page}, which states that $p_d$ must only be used by new edges that were introduced in the fixation gadget.
		However, observe that any new edge $e \in E(G) \setminus (E(H) \cup \{x_{\alpha}u_{\alpha}^0, x_{\alpha}u_{\alpha + 1}^0 \mid \alpha \in [k]\})$, i.e., that was not introduced in the fixation gadget, is of the form $e = x_{\alpha}x_{\beta}$ for $1 \leq \alpha < \beta \leq k$, i.e., one of the edges in the $k$-clique.
		\cref{lem:fixation-gadget-properties} tells us that for every new vertex $x_{\alpha} \in \mathcal{X}$ we have $u_{\alpha}^0 \prec x_{\alpha} \prec u_{\alpha + 1}^0$ in any extension \lSL[G] of \lSL[H].
		This implies that we have $u_{\alpha}^0 \prec x_{\alpha} \prec u_{\alpha + 1}^0 \preceq u_{\beta}^0 \prec x_{\beta} \prec u_{\beta + 1}^0$.
		Together with having $u_{\alpha}^0u_{\alpha + 1}^0, u_{\beta}^0u_{\beta + 1}^0 \in E(H)$ and $\sigma_H(u_{\alpha}^0u_{\alpha + 1}^0) = \sigma(u_{\beta}^0u_{\beta + 1}^0) = p_d$, this rules out $\sigma_G(x_{\alpha}x_{\beta}) = p_d$.
		Hence, we observe that the fixation gadget is formed on the base layout and our construction thus (trivially) fulfills \cref{prop:fixation-gadget-dummy-page}.
	\end{statelater}
	
	\subsection{Bringing It Together: Showing Correctness of the Reduction}
	\label{sec:w-1-correctness}
	\iftoggle{long}{%
		We start by shortly summarizing our construction.
		Recall that we want to insert $k$ new vertices that form a clique.
		On a high level, we first created for each vertex $v_{\alpha}^i \in V(G_C)$ a copy $u_{\alpha}^i \in V(H)$ and ordered the latter vertices depending on the color $\alpha$ and the index~$i$.
		Then, we created for each edge $e \in E(G_C)$ a page $p_e$ on which we formed a tunnel that will enforce for every new edge assigned to $p_e$ that its endpoints lie in specific intervals via \cref{property:w-1-layer}.
		Finally, we used the fixation gadget to ensure that a new vertex can only be placed in intervals for its color, i.e., that our construction will enforce \cref{property:w-1-x-i-v-i}.
		Overall, we obtain for an instance $(G_C, k, (V_1, \ldots, V_{k}))$ of \MCC{} an instance $\instance{} = \instanceLong{}$ of \SLEShort{} that we parameterize by the number $\kappa$ of missing vertices and edges.
		Before we show correctness of the reduction, we first argue in \cref{obs:w-1-fulfills-property-x-i-v-i,lem:w-1-fulfills-property-layer} that \instance{} fulfills \cref{property:w-1-x-i-v-i,property:w-1-layer}, respectively.
	}{%
		With the overview of the construction and the intuition behind the reduction settled, we now proceed to show its correctness in \cref{thm:w-1}.
		In the proof, we make use of \cref{property:w-1-x-i-v-i,property:w-1-layer}.
		Therefore, on our path to obtain \cref{thm:w-1},
		we first have to show that our construction fulfills them.
	}%
	Recall that \cref{property:w-1-x-i-v-i} is defined as follows.
	\mccXiViProperty*
	\noindent
	\iftoggle{long}{In \cref{sec:w-1-fixation-gadget}, when}{When}
	incorporating the fixation gadget on $F = k$ vertices in our construction, we identified%
	\iftoggle{long}{
		$v_{i} = u_{\alpha}^0$ and $f_{i} = x_{\alpha}$ for every $\alpha \in [k]$ and $i = \alpha$.}{
		$v_{\alpha} = u_{\alpha}^0$ and $f_{\alpha} = x_{\alpha}$ for every $\alpha \in [k]$.}
	Similarly, we identified $v_{F + 1} = u_{k + 1}^0$.
	The fixation gadget now guarantees thanks to \cref{lem:fixation-gadget-properties} that we have $v_\alpha \prec_{G} f_\alpha \prec_{G} v_{\alpha + 1}$, i.e., $u_{\alpha}^0 \prec x_{\alpha} \prec u_{\alpha + 1}^0$, in any solution \lSL[G].
	Hence, we can observe the following.
	\begin{observation}
		\label{obs:w-1-fulfills-property-x-i-v-i}
		Our instance \instance{} of \SLEShort{} fulfills \cref{property:w-1-x-i-v-i}.
	\end{observation}
	Recall that \cref{lem:fixation-gadget-properties} furthermore tells us that we have in any solution \lSL[G] the page assignment $\sigma(x_{\alpha}u_{\alpha}^0) = \sigma(x_{\alpha}u_{\alpha + 1}^0) = p_d$ for every $\alpha \in [k]$.
	As we have by \cref{property:w-1-x-i-v-i} $u_{\alpha}^0 \prec x_{\alpha} \prec u_{\alpha + 1}^0$ and furthermore by the construction of the fixation gadget
	$\sigma_{H}(\predecessor{u_{\alpha}^0}u_{\alpha}^0) = \sigma_{H}(u_{\alpha}^0\successor{u_{\alpha}^0}) = p_d$
	for every $\alpha \in [k]$, we cannot have in \lSL[G] that $u_{\alpha}^0 \prec x_{\alpha} \prec \successorSpine{\prec_{H}}{u_{\alpha}^0}$ or $\predecessorSpine{\prec_{H}}{u_{\alpha + 1}^0} \prec x_{\alpha} \prec u_{\alpha + 1}^0$ holds, as this would introduce a crossing on page $p_d$.
	As we have in $\prec_H$ the equality $\successor{u_{\alpha}^0} = u_{\alpha}^1$ and $\predecessor{u_{\alpha + 1}^0} = u_{\alpha}^{n_{\alpha} + 1}$ for every $\alpha \in [k]$, we can strengthen \cref{property:w-1-x-i-v-i} and obtain the following.
	\begin{property}
		\label{prop:w-1-x-i-v-i-strong}
		In a solution \lSL to \SLEShort{} we have $u_{\alpha}^1 \prec x_{\alpha} \prec u_{\alpha}^{n_{\alpha} + 1}$ for every~$\alpha \in [k]$.
	\end{property}
	Finally, we now show that our construction fulfills \cref{property:w-1-layer}, which was defined as follows.
	
	\mccLayerProperty*
	\begin{restatable}%
		{lemma}{lemmaWOneFulfillsPropertyLayers}
		\label{lem:w-1-fulfills-property-layer}
		Our instance \instance{} of \SLEShort{} fulfills \cref{property:w-1-layer}.
	\end{restatable}
	\begin{proof}%
		First, recall that we made \cref{obs:w-1-fulfills-property-x-i-v-i}, i.e., our construction fulfills \cref{property:w-1-x-i-v-i}.
		Let \lSL[G] be a solution to \SLEShort{} with $\sigma(x_{\alpha}x_{\beta}) = p_e$, for $e = v_{\alpha}^{i}v_{\beta}^{j} \in E(G_C)$, $1 \leq \alpha < \beta \leq k$.
		\cref{prop:w-1-x-i-v-i-strong} tells us that $u_{\alpha}^1 \prec x_{\alpha} \prec u_{\alpha}^{n_{\alpha} + 1}$ and $u_{\beta}^1 \prec x_{\beta} \prec u_{\beta}^{n_{\beta} + 1}$ holds.
		\cref{prop:w-1-x-i-v-i-strong} also holds for any new vertices $x_{\gamma}$ and $x_{\delta}$ with $\gamma, \delta \in [k] \setminus \{\alpha, \beta\}$ and $\gamma \neq \delta$.
		Furthermore, we have the edges $u_{\gamma}^{1}u_{\gamma}^{n_{\gamma} + 1}$ and $u_{\delta}^{1}u_{\delta}^{n_{\delta} + 1}$ on page $p_e$.
		Hence, all new edges on the page $p_e$ must be among new vertices placed in intervals induced by vertices of color $\alpha$ or $\beta$. 
		
		Now assume that we have $u_{\alpha}^1 \preceq x_{\alpha} \preceq u_{\alpha}^{i}$.
		Using $\sigma_{H}(u_{\alpha}^1u_{\alpha}^{i}) = p_e$ together with $u_{\alpha}^1 \preceq x_{\alpha} \preceq u_{\alpha}^i \prec x_{\beta}$, we derive that $u_{\alpha}^1 \preceq x_{\alpha} \preceq u_{\alpha}^{i}$ results in a crossing on page~$p_e$.
		Hence, $u_{\alpha}^1 \preceq x_{\alpha} \preceq u_{\alpha}^{i}$ cannot hold.
		Now assume that we have $u_{\alpha}^{i + 1} \preceq x_{\alpha} \preceq u_{\alpha}^{n_{\alpha} + 1}$.
		From $\sigma_{H}(u_{\alpha}^{i + 1}u_{\alpha}^{n_{\alpha} + 1}) = p_e$ and $u_{\alpha}^{i + 1} \preceq x_{\alpha} \preceq u_{\alpha}^{n_{\alpha} + 1} \prec x_{\beta}$ we get that $u_{\alpha}^{i + 1} \preceq x_{\alpha} \preceq u_{\alpha}^{n_{\alpha} + 1}$ results in a crossing on page~$p_e$.
		Hence, $u_{\alpha}^{i + 1} \preceq x_{\alpha} \preceq u_{\alpha}^{n_{\alpha} + 1}$ cannot hold.
		Since we can exclude $u_{\alpha}^1 \preceq x_{\alpha} \preceq u_{\alpha}^{i}$ and $u_{\alpha}^{i + 1} \preceq x_{\alpha} \prec u_{\alpha}^{n_{\alpha} + 1}$ by the construction of the tunnel on the page $p_e$, we can derive that $x_{\alpha}$ must be placed in $\intervalPlacing{v_{\alpha}^i}$.
		As similar arguments can be made for $x_{\beta}$, we can conclude that we get a crossing on page $p_e$ unless~$x_{\alpha}$ is placed in $\intervalPlacing{v_{\alpha}^i}$ and $x_{\beta}$ in $\intervalPlacing{v_{\beta}^i}$.
	\end{proof}%
	\iftoggle{long}{%
		\noindent
		We are now ready to show correctness of our reduction, i.e., show the following theorem.}{}
	\begin{restatable}\restateref{thm:w-1}{theorem}{theoremWOne}
		\label{thm:w-1}
		\SLEShort{} parameterized by the number $\kappa$ of missing vertices and edges is \W\textup{[1]}-hard.
	\end{restatable}
	\begin{proofsketch}
		Let $(G_C, k, (V_1, \ldots, V_{k}))$ be an instance of \MCC{} with $N = \Size{V(G_C)}$ and $M = \Size{E(G_C)}$ and let $\instance{} = \instanceLong{}$ be the instance of \SLEShort{} parameterized by the number $\kappa$ of missing vertices and edges created by our construction described above. Closer analysis %
		reveals that the size of $\instance$ is bounded by $\BigO{N + Mk + k^2}$, and we have $\kappa = 3k + \binom{k}{2}$ as $\nadd{} = k$ and $\madd{} = \binom{k}{2} + 2k$; recall that the fixation gadget contributes $2k$ new edges.
		
		Towards arguing correctness, assume that $(G_C, k, (V_1, \ldots, V_{k}))$ contains a colorful $k$-clique~$\mathcal{C}$. We construct a solution to $\instance$ by considering, for every new vertex $x_{\alpha} \in \mathcal{X}$, the vertex $v_{\alpha}^i \in \mathcal{C}$ and placing $x_{\alpha}$ immediately to the right of the copy $u_{\alpha}^i$ of $v_{\alpha}^i$ in $H$. The fact that $\mathcal{C}$ is a clique then guarantees that, for each edge $e\in E(G_C[\mathcal{C}])$, there exists the page $p_e$ in which the corresponding edge $e' \in E(G[\mathcal{X}])$ can be placed in. For the remaining edges from the fixation gadget, we can use the page assignment from \cref{lem:fixation-gadget-properties}.
		
		For the converse (and more involved) direction, assume that \SLEShort\ admits a solution \lSL[G]. By \cref{property:w-1-x-i-v-i}, we have that each $x_\alpha\in \Vadd$ must be placed between $u_{\alpha}^0$ and~$u_{\alpha + 1}^0$. 
		Moreover, our construction together with the page assignment forced by \cref{lem:fixation-gadget-properties} guarantees that $x_\alpha$ is placed between precisely one pair of consecutive vertices $u_{\alpha}^{i_\alpha}$ and~$u_{\alpha}^{i_\alpha+1}$, for some $i_\alpha \in [n_\alpha]$; recall \cref{prop:w-1-x-i-v-i-strong}.
		Our solution $\mathcal{C}$ to the instance of \MCC{} will consist of the vertices $v_{\alpha}^{i_{\alpha}}$, i.e., exactly one vertex per color $\alpha$.
		Moreover, each new edge $x_\alpha x_\beta \in E(G[\mathcal{X}])$ must be placed by $\sigma_G$ on some page, and as our construction satisfies \cref{prop:fixation-gadget-dummy-page,property:w-1-layer}, this page must be one that is associated to one edge $e = v_\alpha^{i_\alpha}v_\beta^{i_\beta}$ of $G_C$. 
		\cref{property:w-1-layer} now also guarantees that this page assignment enforces that $x_{\alpha}$ and $x_{\beta}$ are placed precisely between the consecutive vertices $u_{\alpha}^{i_\alpha}$ and $u_{\alpha}^{i_\alpha+1}$ and $u_{\beta}^{i_\beta}$ and $u_{\beta}^{i_\beta+1}$ of $H$, respectively.
		This means that the vertices in~$\mathcal{C}$ are pairwise adjacent, which implies that $\mathcal{C}$ is a colorful $k$-clique.
	\end{proofsketch}
	\begin{prooflater}{ptheoremWOne}
		Let $(G_C, k, (V_1, \ldots, V_{k}))$ be an instance of \MCC{} with $N = \Size{V(G_C)}$ and $M = \Size{E(G_C)}$.
		Furthermore, let $\instance{} = \instanceLong{}$ be the instance of \SLEShort{} parameterized by the number $\kappa$ of missing vertices and edges created by our construction described above.
		
		We first bound the size of \instance{} and note that we create $N + 2k + 3$ vertices in $H$ in \cref{sec:w-1-layer} and $k$ additional new vertices in $G$ in \cref{sec:w-1-base}.
		We enrich $H$ by $k + 4$ edges for each edge $e \in E(G_C)$.
		This gives us $M\left(k + 4\right)$ edges in $H$ so far.
		Furthermore, in \cref{sec:w-1-fixation-gadget}, we use a fixation gadget to keep the new vertices in place.
		As we introduce for the fixation gadget no new vertex but rather identify vertices of the fixation gadget with already introduced ones from $H$, it only remains to account for the edges of the gadget, which are $(M + 5)k + M + 3$; see \cref{lem:fixation-gadget-properties} and note that \lSL[H] has $\ell = M + 1$ pages as we create one page for each edge in $G_C$ and have the dummy page from the fixation gadget.
		Finally, we also add a clique among the $k$ new vertices, which are $\binom{k}{2}$ additional edges.
		Overall, the size of $H$ and $G$ is therefore in \BigO{N + Mk + k^2}.
		Hence, the size of the constructed instance is polynomial in the size of $G_C$ and the new {%
			parameter is} bounded by a (computable) function of the old parameter, more specifically, we have $\nadd{} = k$ and $\madd{} = \binom{k}{2} + 2k$, thus $\kappa = 3k + \binom{k}{2}$. 
		The instance \instance{} can trivially be created in \FPT($\kappa$)-time.
		We conclude with showing the correctness of our construction.
		
		\proofsubparagraph*{($\boldsymbol{\Rightarrow}$)}
		Let $(G_C, k, (V_1, \ldots, V_{k}))$ be a positive instance of \MCC{} with solution $\mathcal{C} = \{v_{1}^{i}, \ldots, v_{k}^{j}\}$.
		We construct a witness extension \lSL[G] of \lSL[H] to show that \instance{} is a positive instance of \SLEShort.
		First, we copy \lSL[H] to ensure that 
		\lSL[G] extends \lSL[H].
		Then, we extend \lSL[G] as follows.
		
		For every $v_{\alpha}^{i} \in~\mathcal{C}$, we set $u_{\alpha}^{i} \prec x_{\alpha} \prec u_{\alpha}^{i + 1}$.
		We also set $\sigma(x_{\alpha}u_{\alpha}^0) = \sigma(x_{\alpha}u_{\alpha + 1}^0) =~p_d$.
		For every $x_{\alpha}, x_{\beta} \in \mathcal{X}$ with $1 \leq \alpha < \beta \leq k$ let $x_{\alpha}$ be placed in~$\intervalPlacing{v_{\alpha}^i}$ and $x_{\beta}$ be placed in~$\intervalPlacing{v_{\beta}^j}$.
		We set $\sigma(x_{\alpha}x_{\beta}) = p_e$ for the edge $e = v_{\alpha}^{i}v_{\beta}^{j}$.
		As $\mathcal{C}$ is a clique, we must have $v_{\alpha}^{i}v_{\beta}^{j} \in E(G_C)$ and thus we have the page $p_e$ in \instance{}, i.e., this page assignment is well-defined.
		This completes the creation of \lSL[G].
		As it is an extension of \lSL[H] by construction, we only show that no two edges on the same page cross.
		
		It is trivial that no two new edges, i.e., edges from \Eadd{}, can cross as they are all put on different pages.
		For the edges $x_{\alpha}u_{\alpha}^0$ and $x_{\alpha}u_{{\alpha} + 1}^0$ it is sufficient to observe that we assemble the necessary page assignment from \cref{lem:fixation-gadget-properties}.
		What remains to do is to analyze the edges of the form $x_{\alpha}x_{\beta}$.
		Recall that we set $\sigma(x_{\alpha}x_{\beta}) = p_e$ for $e = v_{\alpha}^{i}v_{\beta}^{j}$ and $x_{\alpha}$ is placed in \intervalPlacing{v_{\alpha}^i} and $x_{\beta}$ is placed in \intervalPlacing{v_{\beta}^j}.
		To see that there does not exist an old edge $e' = uv \in E(H)$ that crosses~{%
			$x_{\alpha}x_{\beta}$}, i.e., with $\sigma(e') = p_e$ and $u \prec x_{\alpha} \prec v \prec x_{\beta}$, recall that the only old edge $e' = uv$ on the page~$p_e$ for which we have $u \prec x_{\alpha} \prec v$ is the edge $e' = u_{\alpha}^iu_{\beta}^{j + 1}${%
			. However,} we have $x_{\beta} \prec u_{\beta}^{j + 1}$, i.e., the edge $e'$ ``spans over'' the edge $x_{\alpha}x_{\beta}$.
		As a similar argument can be made to show that there cannot exist an edge $e' = u v \in E(H)$ with $\sigma(e') = p_e$ and $x_{\alpha} \prec u \prec x_{\beta} \prec v$, we conclude that there are no crossings on the page $p_e$.
		
		As all edges are covered by these cases, we conclude that no two edges of the same page in \lSL[G] can cross, i.e., \lSL[G] is a witness that \instance{} is a positive instance of \SLEShort. 
		
		\proofsubparagraph*{($\boldsymbol{\Leftarrow}$)}
		Let \instance{} be a positive instance of \SLEShort{}.
		Hence, there exists a stack layout \lSL[G] that extends \lSL[H].
		We now construct, based on \lSL[G], a set $\mathcal{C}$ of $k$ vertices and show that it forms a colorful clique in $G_C$.
		Recall that \instance{} fulfills \cref{property:w-1-x-i-v-i,property:w-1-layer} and contains the fixation gadget.
		From \cref{property:w-1-x-i-v-i} and \cref{prop:w-1-x-i-v-i-strong}, we conclude that we have $u_{\alpha}^1 \prec x_{\alpha} \prec u_{\alpha}^{n_{\alpha} + 1}$ for each $\alpha \in [k]$.
		Let $x_{\alpha}$ be placed in some \intervalPlacing{v_{\alpha}^i}.
		We now employ our intended semantics and add $v_{\alpha}^i$ to~$\mathcal{C}$.
		\cref{property:w-1-x-i-v-i} ensures that each vertex in $\mathcal{C}$ will have a different color, i.e., for each $\alpha \in [k]$ there exists exactly one vertex $v_{\alpha}^i \in \mathcal{C}$ such that $v_{\alpha}^i \in V_{\alpha}$.
		Hence, it remains to show that $\mathcal{C}$ forms a clique in $G_C$.
		
		Let $x_{\alpha}, x_{\beta} \in \Vadd{}$ be two arbitrary new vertices placed in \intervalPlacing{v_{\alpha}^i} and \intervalPlacing{v_{\beta}^j}, respectively.
		Assume without loss of generality $x_{\alpha} \prec x_{\beta}$.
		To show $v_{\alpha}^{i}v_{\beta}^{j} \in E(G_C)$, let us consider the edge $x_{\alpha}x_{\beta} \in E(G)$.
		We have $\sigma(x_{\alpha}x_{\beta}) = p$ for some page $p \in [\ell]$.
		Trivially, $p \neq p_d$ because $u_{\alpha}^0 \prec x_{\alpha} \prec u_{\alpha + 1}^0 \prec x_{\beta}$ and $\sigma_{H}(u_{\alpha}^0u_{\alpha + 1}^0) = p_d$.
		Furthermore, for any $e  = uv \in E(G_C)$ with $u \not\in V_{\alpha}$ and $v \notin V_{\alpha}$ we get $p \neq p_e$.
		This follows from $\sigma_{H}(u_{\alpha}^1u_{\alpha}^{n_{\alpha} + 1}) = p_e$ and $u_{\alpha}^1 \prec x_{\alpha} \prec  u_{\alpha}^{n_{\alpha} + 1} \prec x_{\beta}$.
		Similar arguments also hold if $u \not\in V_{\beta}$ and $v \not\in V_{\beta}$.
		Hence, $p = p_e$ for an edge $e \in E(G_C) \cap (V_{\alpha} \times V_{\beta})$ must hold.
		However, now all prerequisites for \cref{property:w-1-layer} are fulfilled.
		Thus, we can conclude that the only possible edge $e$ is $e = v_{\alpha}^{i}v_{\beta}^{j}$.
		For any other edge $e' \in E(G_C) \cap (V_{\alpha} \times V_{\beta})$, either $x_{\alpha}$ or $x_{\beta}$ are not positioned in the right interval with respect to $\prec_{G}$.
		Thus, \cref{property:w-1-layer} tells us (indirectly) that we cannot use {%
			the page $p_{e'}$} for the edge $x_{\alpha}x_{\beta}$.
		As the edge $x_{\alpha}x_{\beta}$ has to be placed in some page, and we ruled out every possibility but the page that would be created for the edge $v_{\alpha}^{i}v_{\beta}^{j}$, we conclude that $v_{\alpha}^{i}v_{\beta}^{j} \in E(G_C)$ must hold.
		As $x_{\alpha}$ and $x_{\beta}$ are two arbitrary new vertices from $\Vadd{}$, we derive that $\mathcal{C}$ forms a (colorful) clique in $G_C$, i.e., $(G_C, k, (V_1, \ldots, V_{k}))$ is a positive instance of \MCC{}.
	\end{prooflater}
	\begin{figure}[t]
		\centering
		\includegraphics[page=1]{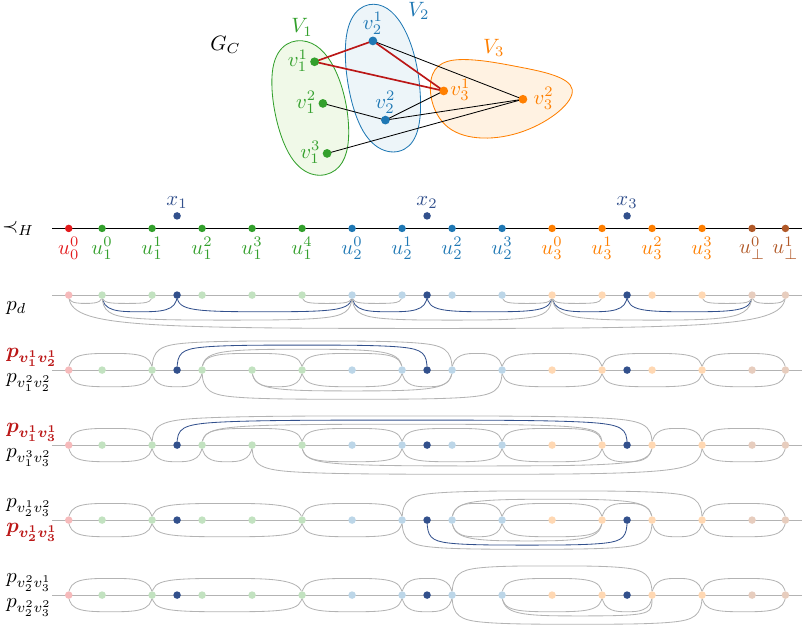}
		\caption{An instance $(G_C, 3, (V_1, V_2, V_3))$ of \MCC{} (top) and the \SLEShort{} instance resulting from our construction (bottom). %
			Colors indicate (correspondence to) the partition. The extension \lSL indicated in saturated colors induces the colorful 3-clique $\mathcal{C} = \{v_1^1, v_2^1, v_3^1\}$ in~$G_C$. The edges in $G_C[\mathcal{C}]$ and their corresponding pages are highlighted in red.}
		\label{fig:w-1-hardness-example}
	\end{figure}
	\noindent
	\cref{fig:w-1-hardness-example} shows an example of the reduction for a small graph $G_C$ with three colors.
	\iftoggle{long}{%
		Taking a closer look at our construction for \cref{thm:w-1} (and \cref{fig:w-1-hardness-example}), we make the following observation.
		Consider a line $l$ perpendicular to the spine.
		On the page $p_e$ for an edge $e = v_{\alpha}^iv_{\beta}^j \in E(G_C)$, the line $l$ intersects at most one edge if placed in the interval for a color $\gamma$ with $\gamma < \alpha$ or $\gamma > \beta$.
		If $\alpha < \gamma < \beta$, the line $l$ can in addition intersect the edges $u_{\alpha}^iu_{\beta}^{j + 1}$ and $u_{\alpha}^{i + 1}u_{\beta}^{j}$ of \cref{fig:w-1-layer}c.
		Finally, if $\gamma \in \{\alpha, \beta\}$, the line $l$ intersects at most the full span of the tunnel, which has a width of three, see again \cref{fig:w-1-layer}.
		Hence, the width of the page $p_e$ is at most three.
		Similarly, for the page $p_d$, \cref{fig:fixation-gadget-example,fig:w-1-fixation-gadget} show that its width is also at most three. Hence, the page width $\omega$ of \lSL[H] is constant and we obtain \cref{cor:w-1}.
	}{%
		Since in a stack layout constructed by our reduction each line perpendicular to the spine intersects a constant number of edges, see also \cref{fig:w-1-hardness-example}, we also obtain:
	}%
	
	\begin{corollary}
		\label{cor:w-1}
		\SLEShort{} parameterized by the number $\kappa$ of missing vertices and edges and the page width~$\omega$ of the given layout, i.e., by $\kappa + \omega$, is \W\textup{[1]}-hard.
	\end{corollary}

	\section{Adding the Number of Pages as Parameter for \textsc{SLE}}
	\label{sec:fpt}
	\newcommand{\stateFPTTheorem}{\begin{restatable}\restateref{thm:kappa-page-number-page-width-fpt}
			{theorem}{theoremPageWidthFPT}
			\label{thm:kappa-page-number-page-width-fpt}
			Let $\instance = \instanceLong$ be an instance of \SLEShort{}.
			We can find %
			an $\ell$-page stack layout of~$G$ that extends \lSL or report that none exists in
			{%
				$\BigO{\ell^{\madd} \cdot {\nadd{}}! \cdot \madd{}^{\nadd{}+1} \cdot \omega^{\madd{}} \cdot \nadd{}\cdot \Size{\instance{}}}$}
			time{%
				, where $\nadd$ and $\madd$ denote the number of new vertices and edges, respectively, and $\omega$ is the page width of the given layout \lSL.}
	\end{restatable}}
	
	In this section, we complete the landscape of \cref{fig:complexity-landscape} by showing that \SLEShort becomes fixed-parameter tractable once we add $\ell$ to the parameterization considered by \cref{cor:w-1}\iftoggle{long}{, i.e., we show the following theorem.
		\stateFPTTheorem}{.}
	We will make use of the following concepts.
	\iftoggle{long}{}{
		
	}%
	Consider a page $p$ of a stack layout \lSL of $G$ and recall that we can interpret it as a plane drawing of the graph $G'$ with $V(G') = V(G)$ and $E(G') = \{e \in E(G) \mid \sigma(e) = p\}$ on a half-plane, where the edges are drawn as (circular) arcs.
	A \emph{face} on the page $p$ in \lSL coincides with the notion of a face in the drawing (on the half-plane~$p$) of $G'$.
	This also includes the definition of the \emph{outer face}.
	See \cref{fig:faces} for a visualization of \iftoggle{long}{these and the following concepts}{this concept} and observe that we can identify every face, except the outer face, by the unique edge $e = uv \in E(G)$ with $u \prec v$ and $\sigma(e) = p$ that bounds it from {%
		above}.
	\begin{figure}[t]
		\centering
		\includegraphics[page=1]{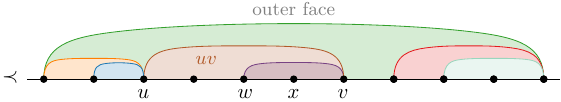}
		\caption{A stack layout \lSL and the faces on page $p$. Note that each edge has the same color as the face it identifies.
			\iftoggle{long}{%
				While the vertex $w$ is incident to the face $uv$, the vertex $x$ and the interval $[w, x]$ are not. However, the face $uv$ still spans the interval $[w, x]$.}{}}
		\label{fig:faces}
	\end{figure}
	\iftoggle{long}{%
		In the following, we will address a face by the edge it is identified with.
		Similarly, we say that an edge induces the face it identifies.
		We say that a vertex $w$ is \emph{incident} to the face $uv$ (on some page $p$) if $u \preceq w \preceq v$ holds and there does not exist a different face $u'v'$ (on the {%
			same} page $p$) with %
		{%
			$u \preceq u' \prec w \prec v' \preceq v$.}
		{%
			Similarly, an interval $[w, x]$ is incident to a face if the (non-existing) vertex $y$ with $w \prec y \prec x$ is incident to the face; note that every interval is incident to precisely one face on each page.
		}%
		Finally, we say that a face~$uv$ \emph{spans} an interval $[w, x]$ if $u \preceq w \prec x \preceq v$ holds; note that $[w, x]$ might not be incident to the face~$uv$.}{}
	
	Let $\Vinc{} \subseteq V(H)$ be the vertices of $H$ that are incident to new edges, i.e., $\Vinc{} \coloneqq\{u \in V(H) \mid \text{there is an edge}\ e = uv \in \Eadd{}\}$.
	The size of \Vinc{} is upper-bounded by $2\madd{}$. 
	We will define an equivalence class on the intervals of~$\prec_{H}$ based on the location of the vertices from~\Vinc{}.
	Consider the two intervals $[u_1, v_1]$ and $[u_2, v_2]$ defined by the old vertices $u_1$, $v_1$,~$u_2$ and $v_2$, respectively.
	These two intervals are in the same equivalence class if and only if 
	$\{w \in \Vinc{} \mid w \preceq u_1\} = \{w \in \Vinc{} \mid w \preceq u_2\}$ and $\{w \in \Vinc{} \mid v_1 \preceq w\} = \{w \in \Vinc{} \mid v_2 \preceq w \}$
	holds.
	\begin{figure}
		\centering
		\includegraphics[page=1]{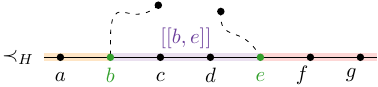}
		\caption{Visualization of \superIntervals{}. Each color represents one \superInterval{}. Vertices from~\Vinc{} are marked in green.}
		\label{fig:super-intervals}
	\end{figure}
	Each equivalence class, which we call \emph{\superInterval{}}, consists of a set of consecutive intervals delimited by (up to) two old vertices; see \cref{fig:super-intervals}.
	Note that the first and last \superInterval{} are defined by a single vertex~$v \in \Vinc{}$.
	The number of \superIntervals{} is bounded by $2\madd{} + 1$.
	\begin{statelater}{notationSuperInterval}%
		We denote the \superInterval{} delimited by the two vertices~$u, v \in \Vinc{}$ with $u \prec v$ by $[[u, v]]$. 
		For the remainder of this paper, we assume that every \superInterval{} is bounded by two vertices.
		This is without loss of generality, since we can place dummy vertices at the beginning and end of the spine and assume that they bound the first and the last interval.
		Furthermore, we write $\assignedSuperInterval{w} = [[u, v]]$ to denote that the \emph{new} vertex $w \in \Vadd{}$ is placed, with respect to a given spine order~$\prec_{G}$, in the \superInterval{} $[[u, v]]$.
	\end{statelater}
	Furthermore, for a given $\prec_{G}$, we define \restrictedSpineOrder{} to be its restriction to new vertices, i.e., for every two vertices $u, v \in \Vadd{}$ we have that $u \prec v$ implies $u \restrictedSpineOrder{} v$.
	\iftoggle{long}{
		\subparagraph*{A Helpful Lemma Towards the Fixed-Parameter Algorithm.}
	}{    
		\subparagraph*{The Algorithm.}
	}
	With the above concepts at hand, we can now describe our algorithm.
	It consists of a branching step, where we consider all possible page assignments for the new edges, all relative orders among the new vertices, all their possible assignments to \superIntervals{}, and all {%
		the \emph{distances}} new edges can have from the outer face {%
		with respect to $H$, i.e., for each new edge $e$, how many edges of $H$ we need to remove until $e$ lies in the outer face of the stack layout for $H$ and $e$; see \cref{fig:semantics-omega}.}
	\iftoggle{long}{%
		The core of our approach is a (dynamic programming) algorithm that we apply in each branch.
		In particular, we aim to show the following lemma.
	}{%
		In the following, we show that we can verify in polynomial-time whether a branch can be extended to a solution \lSL or not.
		The core of our approach is a dynamic program that we apply in each branch.}%
	\begin{restatable}\restateref{lem:dp}{lemma}{lemmaDP}
		\label{lem:dp}
		Given an instance $\instance = \instanceLong$ of \SLEShort{},
		\textbf{\textup{(i)}} a page assignment~$\sigma_G$ for all edges,
		\textbf{\textup{(ii)}} an order \restrictedSpineOrder{} in which the new vertices will appear along the spine,
		\textbf{\textup{(iii)}} for every new vertex $v \in \Vadd{}$ an assignment to a \superInterval{}, and
		\textbf{\textup{(iv)}} for every new edge~$e$ an assigned distance~$\omega_e$ to the outer face with respect to~$H$ and \lSL.
		In 
		\BigO{\nadd{}\cdot\madd{}\cdot \Size{\instance{}}}
		time, we can compute an $\ell$-page stack layout of~$G$ that extends \lSL and respects the given assignments \textbf{\textup{(i)}}--\textbf{\textup{(iv)}} or report that no such layout exists.
	\end{restatable}
	\begin{proofsketch}
		We first observe that assignments (i)--(iv) fix everything, except for the actual position of the new vertices within their \superInterval{}.
		Especially, assignment (i) allows us to check whether an edge $e\in\EaddH{}$ incident to two old vertices crosses any old edge or another new edge from $\EaddH{}$.
		Furthermore, assignments (i) and (ii) allow us to check whether two new edges $e = ua, e' = vb \in \Eadd{}$ with $u,a,v,b\in\Vadd$ will cross. 
		Adding assignment~(iii), we can also check this for new edges with some endpoints in $V(H)$, i.e., extend this to all $u,a,v,b\in V(G)$.
		If the assignments imply a crossing or contradict each other, we can directly return that no desired layout exists.
		These checks can be performed in %
		\BigO{{\nadd{}}^2 + \madd{}\cdot \Size{\instance{}}} time.
		It remains to check whether there exists a stack layout in which no edge of $\Eadd \setminus \EaddH$ intersects an old edge.
		This depends on the exact intervals new vertices are placed in.
		
		To do so, we need to assign new vertices to faces such that adjacent new vertices are in the exact same face and not two different faces with the same distance to the outer face.
		We will find this assignment using a dynamic program that models whether there is a solution that places the first $j$ new vertices (according to \restrictedSpineOrder{}) within the first $i$ intervals in $\prec_{H}$.
		When placing the vertex $v_{j + 1}$ in the $i$th interval, we check that all preceding neighbors are visible in the faces assigned by~(iv).
		When advancing to the interval $i + 1$, we observe that when we leave a face, all edges with the same or a higher distance to the outer face need to have both endpoints placed or none.
		We thus ensure that for no edge only one endpoint has been placed; see also \cref{fig:kappa-page-number-page-width-admissible-predecessor-sketch}.
		These checks require \BigO{\madd{}} time for each of the \BigO{\nadd{}\cdot \Size{V(H)}} combinations of $j$ and $i$.
		Once we reach the interval $\Size{V(H)} + 1$ and have successfully placed all $\nadd$ new vertices, we know that there exists an $\ell$-page stack layout of $G$ that extends \lSL and respects the assignments.
		Finally, by applying standard backtracking techniques, we can extract the spine positions of the new vertices to also obtain the layout.
		\begin{figure}
			\centering
			\includegraphics[page=2]{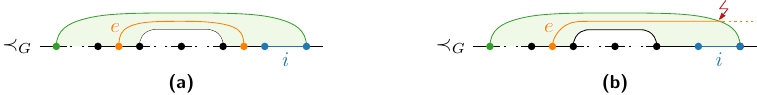}
			\caption{Illustration of advancing from the $i$th interval, marked in blue, to the interval $i + 1$.
				In~\textbf{\textsf{(a)}} and~\textbf{\textsf{(b)}}, we leave the green face and there exists an edge $e \in \Eadd$, marked in orange, with the same distance to the outer face as the green face.
				However, in~\textbf{\textsf{(a)}}, both end points of the edge $e$ have already been placed, whereas in~\textbf{\textsf{(b)}} only one has, which implies a crossing.}
			\label{fig:kappa-page-number-page-width-admissible-predecessor-sketch}
		\end{figure}
	\end{proofsketch}%
	\iftoggle{long}{}{%
		We observe that there are \BigO{\ell^{\madd} \cdot {\nadd{}}! \cdot \madd{}^{\nadd{}} \cdot \omega^{\madd{}}} different possibilities for assignments~(i)--(iv).}
	\begin{statelater}{sectionDP}
		\noindent
		Before we show \cref{lem:dp}, we first make some observations on the assignments (i)--(iv) and their immediate consequences.
		In the following, we only consider \emph{consistent} branches, i.e., we discard branches where from assignment (ii) we get $u \restrictedSpineOrder{} v$ but from assignment (iii) $\assignedSuperInterval{u} = [[a, b]]$ and $\assignedSuperInterval{v} = [[c, d]]$ with $c \prec d \preceq a \prec b$, as this implies $v \prec u$.
		{%
			We can check this in $\BigO{{\nadd}^2}$ time.
		}
		
		First, we observe that assignment (i) fully determines the page assignment $\sigma_G$.
		Thus, it allows us to check whether an edge $e\in\EaddH{}$, i.e., a new edge incident to two old vertices, crosses any old edge or another new edge from $\EaddH{}$.
		{%
			To do this efficiently, we first extend $\lSL[H]$ in $\BigO{\madd}$ time by the (at most) $\madd$ new edges $\EaddH{}$.
			Afterwards, we can check in linear time if the resulting stack layout $\langle \prec', \sigma'\rangle$ is crossing-free.
			Overall, this takes $\BigO{\madd + \Size{\instance}}$ time.%
		}
		
		Second, assignments (i) and (ii) allow us to check whether two new edges $e = ua, e' = vb \in \Eadd{}$ with $u,a,v,b\in\Vadd$ will cross each other.
		{%
			This can be done in linear time as well.
			First, we create in $\BigO{\Size{\instance}}$ time the stack layout $\langle\restrictedSpineOrder, \sigma_G'\rangle$, where $\sigma_G'(e) = \sigma_G(e)$ for all $e = uv \in \Eadd$ with $u,v \in \Vadd$. 
			Afterwards, we check in $\BigO{\Size{\instance}}$ time if $\langle\restrictedSpineOrder, \sigma_G'\rangle$ is crossing free.
		}
		
		Third, adding assignment (iii), we can also check this for new edges with some endpoints in $V(H)$, i.e., extend this to all $u,a,v,b\in V(G)$.
		Hence, assignments (ii) and (iii) together with $\prec_H$ fix the relative order among vertices incident to new edges.
		\cref{fig:crossing-super-interval} shows an example where the assignments imply a crossing among two new edges. 
		{%
			Observe that by similar arguments as above, we can perform these checks in linear time:
			We first restrict $\prec_H$ to $\Vinc$, i.e., old vertices which are endpoints of new edges, and obtain $\prec'$.
			Afterwards, we extend $\prec'$ by $\restrictedSpineOrder$ while respecting assignment~(iii) and obtain a total order $\prec''$ among all vertices that are endpoints of new edges.
			This enables us to construct in linear time a stack layout $\langle \prec'', \sigma''\rangle$ that contains all new edges \Eadd, including $\EaddH$.
			We then check if the resulting stack layout is crossing free.
		}
		\begin{figure}
			\centering
			\includegraphics{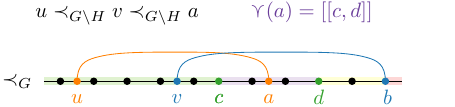}
			\caption{The crossing between the orange edge $e=ua$ and the blue edge $e' = vb$, where the vertices $a$, $u$, and $v$ are new, is already implied by the page assignment $\sigma$ (assignment~(i)), relative order among the new vertices $u$, $v$, and~$a$, i.e., by $u \prec v \prec a$ (assignment~(ii)), and their assignment to the \superIntervals{}, in particular by $\assignedSuperInterval{a}$ (assignment~(iii)). The \superIntervals{} are indicated by color, and we visualize one possible spine position of the new vertices consistent with the {%
					assignments}.}
			\label{fig:crossing-super-interval}
		\end{figure}

		All of the above checks together can be done in
		{%
			$\BigO{{\nadd{}}^2 + \Size{\instance{}}}$
		}%
		time.
		Clearly, if the assignments (i)--(iv) imply a crossing in an $\ell$-page stack layout of $G$ or contradict each other, we report that no layout exists that respects assignments (i)--(iv).
		However, if not, we still need to find concrete spine positions for the new vertices. %
		The main challenge here is to assign new vertices to faces such that adjacent new vertices are in the same face and not {%
			in} two different faces with the same distance to the outer face.
		For this, we use assignment~(iv) together with %
		{%
			a
		}%
		dynamic programming (DP) algorithm.
		{%
			Finally, recall that assignment (i) completely determines the placement of new edges from $\EaddH{}$.
			Since we have already ensured that they are crossing-free with respect to both old and new edges, we exclude them in the upcoming discussion.
			However, we will argue in the proof of \Cref{lem:dp} that re-inserting them in the end is safe.
		}

		\subparagraph*{The Intuition Behind the DP-Algorithm.}
		Recall that assignment~(iv) determines for every new edge $e = uv \in \Eadd{}$ the distance $\omega_e$ to the outer face with respect to \lSL[H]{%
			; see again \cref{fig:semantics-omega}.}%
		\begin{figure}
			\centering
			\includegraphics[page=1]{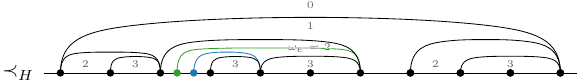}
			\caption{Visualization of the semantics of $\omega_e$. Every face of \lSL[H] has one particular distance~$\omega_e$ to the outer face, shown in gray. Both highlighted new edges are in a face with $\omega_e= 2$.}
			\label{fig:semantics-omega}
		\end{figure}
		{%
			We first observe that if $u$ and $v$ are two new vertices, $\sigma$ and the intervals in $\prec_H$ in which $u$ and $v$ are placed in uniquely determine the single face $e$ is embedded in.
			If only one of them is new, say $v$, the same holds true by using $\omega_e$ even without having access to the interval $v$ is placed in, since we know if $v \prec u$ or $u \prec v$ holds due to assignment~(iii).%
		}
		Furthermore, for every possible distance $\omega_e \in [\omega]_0$ and every interval in $\prec_H$, there is at most one face on the page $\sigma(e)$ with distance $\omega_e$ to the outer face that spans the interval.
		Hence, we address in the following for a given interval $[a, b]$ with $a \prec_H b$ the face on the page $\sigma(e)$ at the distance~$\omega_e$ to the outer face, if it exists, with $\omega_e^{[a, b]}$.
		Note that $\omega_e^{[a, b]} = 0$ always refers to the outer face, independent of the vertices $a$ and~$b$. %
		However, %
		for two different intervals $[a, b]$ and $[a', b']$ the expressions $\omega_e^{[a, b]}$ and~$\omega_e^{[a', b']}$ can identify two different faces.
		
		As a consequence of the above observation, we can decide for each interval of~$\prec_{H}$ whether we can position a new vertex $v$ there, i.e., whether $v$ sees each adjacent vertex using the respective face (in the assigned page) at the corresponding distance from the outer face.
		We now consider the ordering of the new vertices as in \restrictedSpineOrder{}, i.e., we have $v_1 \prec v_2 \prec \ldots \prec v_{\nadd{}}$.
		Furthermore, we number the intervals of $\prec_H$ from left to right and observe that there are $\Size{V(H)} + 1$ intervals.
		
		Consider a hypothetical solution \lSL which we cut vertically at the $i$th interval of~$\prec_{H}$.
		This partitions the new vertices into those that have been placed left and right of the cut.
		For new vertices placed in the $i$th interval of $\prec_{H}$, different cuts at the $i$th interval yield different partitions into left and right.
		Furthermore, some of the new edges lie completely on one side of the cut, while others span the cut.
		For $j \in [\nadd{}]$, let $G_j$ be the graph $G[V(H) \cup \{v_1, \ldots, v_j\}]$, i.e., the subgraph of~$G$ induced by the vertices of $H$ and the first $j$ new vertices.
		We will refer to edges that span the cut and thereby only have one endpoint in $G_j$ as \emph{half-edges}, and denote with $v\cdot$ a half-edge with endpoint $v\in G_j$.
		Let~$G_j^+$ be $G_j$ extended by the half-edges $E_j^+ = \{v_x\cdot \mid v_xv_y \in \Eadd{},\ 1 \leq x \leq j < y \leq \nadd{}\}$; see \cref{fig:kappa-page-number-page-width-gj}.
		\begin{figure}
			\centering
			\includegraphics[page=1]{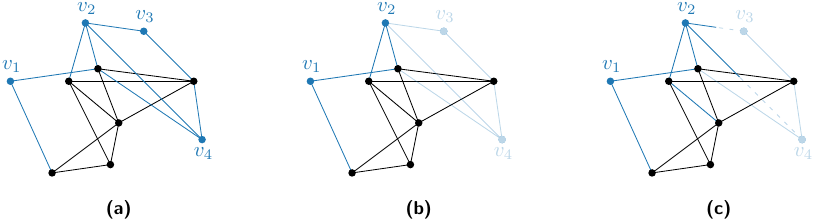}
			\caption{\textbf{\textsf{(a)}} The graph $G$ with the new vertices and edges highlighted in blue and~\textbf{\textsf{(b)}} the graphs $G_2$ and~\textbf{\textsf{(c)}} $G_2^+$. Half-edges are indicated with dashed lines.}
			\label{fig:kappa-page-number-page-width-gj}
		\end{figure}
		In the following, we denote with $e^+$ the half-edge that we create for the edge~$e$.
		For the half-edge $e^+ = v_x\cdot$ and the edge $e = v_xv_y$, we call the vertex~$v_x$ \emph{inside} $G_j$, denoted as \vertexInside{V_j}{e}, and the vertex $v_y$ \emph{outside} $G_j$, denoted as \vertexOutside{V_j}{e}.
		
		Consider again the hypothetical solution \lSL and its vertical cut at the $i$th interval.
		Assume that $j$ vertices have been placed in \lSL left of the cut.
		The stack layout \lSL witnesses the existence of a stack layout for $G_j$ that extends \lSL[H] and uses only the first~$i$ intervals.
		Furthermore, for every half-edge $e^+$ in $G_j^+$ the face where $e$ is placed in gives us a set of candidate intervals for the vertex \vertexOutside{V_j}{e}, namely those incident to that face.
		Hence, we can describe a partial solution of \lSL by a tuple $(i, j)$.

		\subparagraph*{The DP-Algorithm.}
		Let $D$ be an $(\Size{V(H)} + 1) \times (\nadd{} + 1) \times 2$ binary table.
		In the following, we denote an entry $(i, j, r)$ for $i \in [\Size{V(H)} + 1]$, $j \in [\nadd{} + 1]$, and $r \in \{0, 1\}$ as a \emph{state} of the algorithm.
		A state $(i, j, r)$ for $i \in [\Size{V(H)} + 1]$, $j \in [\nadd{} + 1]$, and $r \in \{0, 1\}$ is called \emph{feasible} if and only if there exists, in the current branch, an extension of \lSL[H] for the graph $G_j$ with the following properties.
		\begin{enumerate}[(FP~1)]
			\item The $j$ new vertices are positioned in the first $i$ intervals and their placement respects assignment~(iii){%
				, i.e., their assignment to a super interval}.\label[FP]{dp:feasible-j-in-i}
			\item If $r = 1$, the last vertex $v_j$ has been placed in the $i$th interval.
			Otherwise, i.e., if $r = 0$, the last vertex $v_j$ has been placed in some interval $i'$ with $i' < i$.\label[FP]{dp:feasible-r}
			\item For every half-edge $e^+$ of $G_j^+$, the face in which we placed (the first endpoint of) $e^+$ spans the $i$th interval.\label[FP]{dp:feasible-half-edges}
		\end{enumerate}

		Note that for \cref{dp:feasible-half-edges}, we neither require that there exists some interval $i'$ with $i \leq i'$ for the vertex \vertexOutside{V_j}{e} that is incident to $\omega_e^i$ nor that this $i'$th interval is part of \assignedSuperInterval{\vertexOutside{V_j}{e}}, i.e., the \superInterval{} for \vertexOutside{V_j}{e} according to assignment~(iii).
		However, we will ensure all of the above points in the DP  when placing \vertexOutside{V_j}{e}.
		Furthermore, while the last (binary) dimension is technically not necessary, it simplifies our following description.
		Finally, note that we (correctly, as required in some solutions) allow positioning multiple vertices in the same interval of $\prec_H$.
		
		Our DP will mark a state $(i, j, r)$ for $i \in [\Size{V(H)} + 1]$, $j \in [\nadd{} + 1]$, and $r \in \{0, 1\}$ as feasible by setting $D[i, j, r] = 1$.
		Before we can show in \cref{lem:kappa-page-number-page-width-dp-recurrence-correctness} that our DP indeed captures this equivalence, let us first relate different states and thus also partial solutions to each other.
		We observe that if we have a solution for $i \in [\Size{V(H)}]$, $j \in [\nadd{}]_0$, and $r \in \{0,1\}$, and for every half-edge $e^+$ of $G_j^+$ the face $\omega_e^i$ also spans the $i + 1$th interval, then we also have a solution in the state $(i + 1, j, 0)$.
		Otherwise, we cannot find an interval for the vertex \vertexOutside{V_j}{e} for some half-edge $e^+$ of $G_j^+$; see \cref{fig:kappa-page-number-page-width-admissible-predecessor}b.
		\begin{figure}
			\centering
			\includegraphics[page=1]{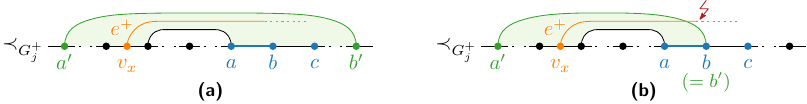}
			\caption{Illustration of some partial solution $(i, j)$ with the half-edge $e^+ = v_x\cdot$ in $G_j^+$ that we place in the green face. 
				The $i$th interval is $[a,b]$, highlighted in blue, and $\omega_e^i$ is identified by the green edge $a'b'$.
				In~\textbf{\textsf{(a)}}, the partial solution $(i, j, r)$ is an admissible predecessor of $(i + 1, j, 0)$, whereas in~\textbf{\textsf{(b)}} it is not, as we can no longer find an interval for the vertex \vertexOutside{V_j}{e} that does not introduce a crossing between $e$ and $a'b'$. In particular, observe that the face $\omega_e^i$ does not span the interval $[b, c]$. \iftoggle{long}{}{Extended version of \cref{fig:kappa-page-number-page-width-admissible-predecessor-sketch}.}}
			\label{fig:kappa-page-number-page-width-admissible-predecessor}
		\end{figure}
		More formally, let $[a, b]$ be the $i$th interval and $[b, c]$ the $i + 1$th interval.
		Assume that we have the new edge $e = v_xv_y$ in~$G$ with $1 \leq x \leq j < y \leq k$, which is the half-edge~$v_x\cdot$ in $G_j^+$.
		Let $a'b'$ be the edge (on page $\sigma_G(e)$) that bounds the face identified by $\omega_e^i$ upwards with $a' \prec b \preceq b'$.
		We call $(i, j, r)$ an \emph{admissible predecessor} of $(i + 1, j, 0)$ if $b \neq b'$; see \cref{fig:kappa-page-number-page-width-admissible-predecessor}a.
		
		If we decide to place the vertex $v_j$ in the $i$th interval, say $[a, b]$, we have to be more careful.
		In particular, we have to ensure the following criteria before we can conclude that $(i, j, 1)$ %
		{%
			is feasible.}
		\begin{enumerate}[(EC~1)]
			\item The interval $[a, b]$ is part of the \superInterval{} \assignedSuperInterval{v_j}, i.e., we place the new vertex in {%
				its assigned} \superInterval{}.\label[EC]{dp:extension-super-interval}
			\item For {%
				each new edge} $e = v_{j}u$ with $u \in V(H)$ the face $\omega_e^i$ exists on the page $\sigma_G(e)$ and both the interval {%
				$[a, b]$} (and thus $v_j$) and $u$ are incident to the face $\omega_e^i$.\label[EC]{dp:extension-old-vertex}%
			\item For {%
				each new edge} $e = v_jv_{q}$ incident to two new vertices the face $\omega_e^i$ exists on the page~$\sigma_G(e)$ and the interval~{%
				$[a, b]$} (and thus $v_j$) is incident to it.\label[EC]{dp:extension-new-vertex}
		\end{enumerate}
		{%
			Recall that we only consider assignments where no two new edges cross. 
			Note that this includes edges from \EaddH, although they are not considered in this DP.
			We observe that the} 
		second criterion ensures that edges incident to $v_j$ and an old vertex can be inserted without introducing a crossing.
		The third criterion ensures for a new edge~$e$ incident to two new vertices that the half-edge $e^+$ has been completed to a full edge without introducing a crossing {%
			(if $q < j$)} or is placed in the face $\omega_e^i$ that spans the $i$th interval (if $j < q$).
		{%
			Note that we only advance to the interval $[a,b]$ if the face containing the first vertex of the half-edge $e^+$ still spans the interval; recall \cref{dp:feasible-half-edges}.
			Hence, if $v_j$ is placed in the face $\omega_e^i$, so is $v_q$ (for the case $q < j$); see also the proof of \cref{lem:kappa-page-number-page-width-dp-recurrence-correctness}.
		}
		
		Similar to before, we call $(i, j - 1, r)$ an \emph{admissible predecessor} of $(i, j, 1)$ if all of the above criteria are met.
		Finally, note that if we decide to place the vertex $v_j$ in the $i$th interval, then the state $(i, j, 0)$ is not feasible due to \cref{dp:feasible-r}.
		Our considerations up to now are summarized by the recurrence relation in \cref{def:kappa-page-number-page-width-dp-recurrence}.
		In \cref{lem:kappa-page-number-page-width-dp-recurrence-correctness}, we show that the recurrence relation identifies exactly all feasible states. 
		\begin{definition}
			\label{def:kappa-page-number-page-width-dp-recurrence}
			We have the following relation for all $i \in [2,\Size{V(H)} + 1]$ and $j \in [\nadd{}]_0$.
			\begin{align}
				D[1, 0, 0] &= 1\label{eq:kappa-page-number-page-width-dp-recurrence-base-0}\\
				D[1, 0, 1] &= 0\label{eq:kappa-page-number-page-width-dp-recurrence-base-1}\\
				D[i, j, 0] &= \begin{cases}
					1, &\ \text{if}\ (i - 1, j, r)\ \text{is an admissible predecessor of}\ (i, j, 0)\\
					&\ \text{and}\ D[i - 1, j, r] = 1\ \text{for some}\ r \in \{0,1\}\\
					0, &\ \text{otherwise}
				\end{cases} \label{eq:kappa-page-number-page-width-dp-recurrence-recurrence-r-0}\\
				D[i, j, 1] &= \begin{cases}
					1, &\ \text{if}\ (i, j - 1, r)\ \text{is an admissible predecessor of}\ (i, j, 1)\\
					&\ \text{and}\ D[i, j - 1, r] = 1\ \text{for some}\ r \in \{0,1\}\\
					0, &\ \text{otherwise}
				\end{cases} \label{eq:kappa-page-number-page-width-dp-recurrence-recurrence-r-1}
			\end{align}
		\end{definition}
		
		\begin{restatable}%
			{lemma}{lemmaDPRecurrenceCorrectness}
			\label{lem:kappa-page-number-page-width-dp-recurrence-correctness}
			For all $i \in [\Size{V(H)} + 1]$, $j \in [\nadd{}]_0$, and $r \in \{0,1\}$ we have $D[i, j, r] = 1$ if and only if the state $(i, j, r)$ is feasible.
			Furthermore, evaluating {%
				one step from the} recurrence relation of \cref{def:kappa-page-number-page-width-dp-recurrence} takes \BigO{\madd{}} time.
		\end{restatable}
		\begin{proof}%
			We first use induction over $i$ and $j$ to show correctness of the recurrence relation, and later argue the time required to evaluate it.
			In the following, we let $[a, b]$ be the $i$th interval.
			
			\proofsubparagraph*{Base Case.}
			If we have $i = 1$ and $j = 0$, we are in the first interval and have not placed any new vertex.
			Thus, $G_j = G_j^+ = H$ holds and $G_j$ has clearly a solution, namely \lSL[H].
			Furthermore, there are no half-edges in $G_j^+$ and thus $(1, 0, r)$ is a feasible state if and only if $r = 0$.
			Note that $(1, 0, 1)$ contradicts \cref{dp:feasible-r} and thus cannot be a feasible state.
			Hence, \cref{eq:kappa-page-number-page-width-dp-recurrence-base-0,eq:kappa-page-number-page-width-dp-recurrence-base-1} are correct and serve as our base case.
			
			In our inductive hypothesis, we assume that the table $D$ has been correctly filled up until some value $i' \in [\Size{V(H)} + 1]$ and $j' \in [\nadd{}]_0$.
			
			\proofsubparagraph*{Inductive Step for $\boldsymbol{i}$ ($\boldsymbol{i = i' + 1}$ and $\boldsymbol{j = j'}$).}
			First, note that by moving one interval to the right, having $r = 1$ is not possible and hence we focus on \cref{eq:kappa-page-number-page-width-dp-recurrence-recurrence-r-0} in this step.
			We consider the cases $D[i, j, 0] = 1$ and $D[i, j, 0] = 0$ separately.
			
			For $D[i, j, 0] = 1$, there exists by the definition of \cref{eq:kappa-page-number-page-width-dp-recurrence-recurrence-r-0} an admissible predecessor $(i', j, r)$ for some $r \in \{0,1\}$ with $D[i', j, r] = 1$.
			By our inductive hypothesis, this means that the state $(i', j, r)$ is feasible.
			In particular, it has a solution for $G_j$ that places all new vertices in the first $i'$ intervals.
			Clearly, the same solution positions them also in the first $i$ intervals.
			Furthermore, every half-edge~$e^+$ of~$G_j^+$ is assigned to a face that spans the $i'$th interval.
			As $(i', j, r)$ is an admissible predecessor, we know that for every {%
				half-edge~$e^+$ of $G_j^+$}, the face $\omega_e^{i'}$ also spans the $i$th interval.
			Hence, the state $(i, j, 0)$ is feasible, i.e., $D[i, j, 0] = 1$ correctly holds.
			
			For $D[i, j, 0] = 0$, there are two cases to consider by the definition of \cref{eq:kappa-page-number-page-width-dp-recurrence-recurrence-r-0}.
			Either $(i, j, 0)$ does not have an admissible predecessor, or for all admissible predecessors $(i', j, r)$ of $(i, j, 0)$ we have $D[i', j, r] = 0$.
			Observe that only states of the form $(i', j, r)$, for some $r \in \{0,1\}$, can be admissible predecessors.
			In the former case, there exists by our definition of admissible predecessor some half-edge $e^+$ in $G_j^+$ that is assigned to the face $\omega_e^{i'}$ which does not span the $i$th interval.
			Hence, $(i, j, 0)$ is not feasible by \cref{dp:feasible-half-edges}.
			In the latter case, we know by our inductive hypothesis that $D[i', j, r] = 0$ implies that $(i', j, r)$ is not feasible, i.e., there does not exist a solution for the graph $G_j$ in which we place the new vertices in the first $i'$ intervals.
			As we do not place the vertex $v_j$ in the $i$th interval, no solution can exist for the state $(i, j, 0)$ either and $D[i, j, 0] = 0$ correctly holds.
			This concludes the inductive step for $i$.
			
			\proofsubparagraph*{Inductive Step for $\boldsymbol{j}$ ($\boldsymbol{i = i'}$ and $\boldsymbol{j = j' + 1}$).}
			Analogous to before, we note that by placing~$v_j$ in the $i$th interval, $r = 0$ is not possible and hence we focus on \cref{eq:kappa-page-number-page-width-dp-recurrence-recurrence-r-1} in this step.
			We again consider the cases $D[i, j, 1] = 1$ and $D[i, j, 1] = 0$ separately.
			
			For $D[i, j, 1] = 1$, there exists an admissible predecessor $(i, j', r)$ for some $r \in \{0,1\}$ with $D[i, j', r] = 1$ by the definition of \cref{eq:kappa-page-number-page-width-dp-recurrence-recurrence-r-1}. 
			By our inductive hypothesis, this means that the state $(i, j', r)$ is feasible and has a solution \lSL[G_{j - 1}] that places all new vertices of $G_{j - 1}$ in the first $i$ intervals.
			We now construct a solution \lSL[G_j] by setting~$\sigma_{G_j}$ according to assignment~(i), and extending the spine order $\prec_{G_{j - 1}}$ by placing~$v_j$ in the $i$th interval.
			More concretely, we take~$\prec_{G_{j - 1}}$, set $a \prec v_j \prec b$ {%
				and $v_{j - 1} \prec v_j$ (if $j > 1$)}, and take the transitive closure to obtain a linear order on the vertices of~$G_j$.
			As $(i, j', r)$ is an admissible predecessor, $v_j$ is placed within \assignedSuperInterval{v_j} in \lSL[G_j] and \lSL[G_j] extends \lSL[H].
			So it remains to show that \lSL[G_j] is crossing-free.
			As we discard all assignments~(i)--(iv) that imply a crossing among two new edges, a new edge $e$ incident to~$v_j$ could only cross an old edge~$e'$.
			However, as $(i, j', r)$ is an admissible predecessor, we have that the vertices incident to $e$ lie in the same face $\omega_e^i$.
			Every old edge induces a face of \lSL[H].
			Therefore, we deduce that a crossing between an old and a new edge whose endpoints lie in the same face is impossible. %
			Consequently, also $e$ and $e'$ cannot cross and \lSL[G_j] is a solution for the graph~$G_j$.
			Finally, it is clear that for every half-edge $e^+$ of $G_j^+$, the face assigned to its associated edge $e$ spans the $i$th interval.
			For half-edges that already existed in $G_{j'}^+$, this holds as $(i, j', r)$ is feasible.
			For half-edges introduced in $G_j^+$, this holds by the definition of admissible predecessor; see \cref{dp:extension-old-vertex,dp:extension-new-vertex}.
			Thus, $(i, j, 1)$ is feasible and $D[i, j, 1] = 1$ correctly holds.
			
			For $D[i, j, 1] = 0$, there are again two cases to consider by the definition of \cref{eq:kappa-page-number-page-width-dp-recurrence-recurrence-r-1}.
			Either $(i, j, 1)$ does not have an admissible predecessor, or for all admissible predecessors $(i, j', r)$ of $(i, j, 0)$ we have $D[i, j', r] = 0$.
			Again, we observe that it suffices to consider only states of the form $(i, j', r)$ for some $r \in \{0,1\}$ as potential admissible predecessors.
			For the former case, clearly, if both such states $(i, j', r)$ are not admissible predecessors, then $(i, j, 1)$ cannot be feasible:
			Either, we violate \cref{dp:extension-super-interval} by placing $v_j$ outside \assignedSuperInterval{v_j}, which contradicts assignment~(iii), or an edge $e$ incident to~$v_j$ crosses an old edge as one of its endpoints is not incident to $\omega_e^i$; see \cref{dp:extension-old-vertex,dp:extension-new-vertex}.
			Note that we can assume that all relevant faces span the $i$th interval, as we have already shown that the inductive step for $i$ is correct.
			In both cases, $(i, j, 1)$ is clearly not feasible; see \cref{dp:feasible-j-in-i} for the former and observe that a crossing contradicts the existence of an extension for the latter case.
			We now consider the case where all admissible predecessor $(i, j', r)$ of $(i, j, 1)$ have $D[i, j', r] = 0$.
			Using proof by contradiction, we show that in this case $(i, j, 1)$ cannot be feasible either.
			Assume that $(i, j, 1)$ would be a feasible state.
			Then, there exists a solution \lSL[G_j] for the graph $G_j$.
			Using \lSL[G_j], we can create a solution \lSL[G_{j'}] for $G_{j'}$ by removing $v_j$ from $\prec_{G_j}$ and all its incident edges from $\sigma_{G_j}$.
			Clearly, \lSL[G_{j'}] respects the assignments~(i)--(iv).
			Furthermore, for every half-edge $e^+$ in $G_j^+$, the assigned face spans the $i$th interval as $(i, j, 1)$ is feasible; see \cref{dp:feasible-half-edges}.
			Hence, if $(i, j, 1)$ is feasible and $(i, j', r)$ is an admissible predecessor, then $(i, j', r)$ is also feasible.
			However, this contradicts the inductive hypothesis, as we have $D[i, j - 1, r] = 0$.
			Thus, the state $(i, j, 1)$ cannot be feasible and $D[i, j, 1] = 0$ correctly holds.
			This concludes the inductive step for~$j$.
			
			\proofsubparagraph*{Evaluation Time of Recurrence Relation.}
			We observe that, apart from checking whether a state is an admissible predecessor of $(i, j, r)$, the steps we perform in order to evaluate the recurrence relation take constant time.
			In the following, we assume that we can access a look-up table that stores the faces in \lSL[H] that span a given interval in $\prec_{H}$  on a given page $p \in [\ell]$ ordered from %
			{%
				inside out, i.e., starting with the innermost face and ending with the outer face.
			}%
			We will account for this in our proof of \cref{lem:dp}.
			In \cref{eq:kappa-page-number-page-width-dp-recurrence-recurrence-r-0}, we ensure for every half-edge $e^+$ in $G_j^+$ that the face $\omega_e^{i - 1}$ does not end at the interval $i - 1$.
			As there are at most \madd{} half-edges, we can do this in \BigO{\madd{}} time.
			For \cref{eq:kappa-page-number-page-width-dp-recurrence-recurrence-r-1}, we have to ensure that the $i$th interval $[a, b]$ is part of \assignedSuperInterval{v_j}, which takes constant time.
			Furthermore, we have to check for every new edge $e = v_ju$ that the face~$\omega_e^i$ exists on the page $\sigma(e)$ and that~$v_j$ is incident to it.
			For a single edge, this takes constant time, as we can look up the faces that span the $i$th interval and~$v_j$ is always incident to the bottom-most face.
			Furthermore, if $u \in V(H)$ holds, we also have to check if $u$ is incident to $\omega_e^i$.
			If $u \preceq a \prec b$, we can access the faces that span the interval $[u, \successorSpine{\prec_{H}}{u}]$ and check if $u$ is incident to the face~$\omega_e^i$.
			An analogous check can be made if we have $a \prec b \preceq u$.
			This takes constant time per edge $e$.
			Hence, we can evaluate \cref{eq:kappa-page-number-page-width-dp-recurrence-recurrence-r-1} in \BigO{\madd{}} time.
			Combining all, the stated running time follows.
		\end{proof}
		\noindent
		\subparagraph*{Putting Everything Together.}
		With our DP at hand, we are now ready to prove \cref{lem:dp}.
		\begin{proof}[Proof of \Cref{lem:dp}]
			First, recall that we can check in 
			{%
				$\BigO{{\nadd{}}^2 + \Size{\instance{}}}$
			}%
			time whether assignments (i)--(iv) are consistent and do not imply  a crossing.
			Thus, we assume for the remainder of the proof that they are.
			{%
				Furthermore, recall that there are $\Size{V(H)} + 1$ intervals in \lSL[H] and observe that we have $G_k = G_k^+ = G'$, where $G' = (V(G), E(G) \setminus \EaddH)$; recall that we ignore \EaddH in the DP.
				Therefore, $G_k^+$ does not contain any half-edge and the feasibility property \cref{dp:feasible-half-edges} is trivially satisfied.
				Hence, as a consequence of \cref{lem:kappa-page-number-page-width-dp-recurrence-correctness}, we deduce that there exists an $\ell$-page stack layout $\lSL[G']$ of $G'$ that extends \lSL and respects the assignments (i)--(iv) if and only if $D[\Size{V(H)}  + 1, \nadd{}, r] = 1$ for $r = 0$ or $r = 1$.
				Furthermore, by applying standard backtracking techniques, we can also determine the concrete spine positions for every new vertex, i.e., compute such a stack layout.
				
				What is now missing from a stack layout of $G$ are the edges $\EaddH$.
				Recall that for every edge $uv \in \EaddH$, we have $u,v \in V(H)$ and $\sigma_G(uv)$ fixed by assignment~(i).
				Hence, $\lSL[H]$ and assignment~(i) fully specify their position in a solution \lSL[G].
				Consequently, we take $\lSL[G']$ and extend it to a stack layout $\langle \prec_{G'}, \sigma_G\rangle$ of $G$.
				To see that it is crossing free, observe that if otherwise, the crossing must involve a new edge $uv \in \EaddH$.
				However, we only consider assignments where no two new edges cross and where no new edge $uv \in \EaddH$ crosses an old edge.
				Hence, $\langle \prec_{G'}, \sigma_G\rangle$ is crossing free and extends, by construction, the layout \lSL[H] of $H$.
			}
			
			We now bound the running time of the DP.
			For that, we first observe that the DP-table $D$ has \BigO{\nadd{}\cdot\Size{V(H)}} entries.
			We have seen in \cref{lem:kappa-page-number-page-width-dp-recurrence-correctness} that the time required to evaluate the recurrence relation is in \BigO{\madd{}}.
			However, we assumed that we have access to a lookup-table that stores for each interval and each page the faces that span it.
			{%
				To compute this lookup-table, we construct, in a pre-processing step, for each page $p \in [\ell]$ the dual graph of the plane drawing of the graph $G'$ with $V(G') = V(G)$ and $E(G') = \{e \in E(G) \mid \sigma(e) = p\}$, where the edges are drawn as (circular) arcs in the direction of the spine.
				Note that the dual graph is a tree $T_p$, which we root at the outer face and augment with one leaf per interval connected to the single face it is incident to.
				To construct the (augmented) tree $T_p$, we iterate, from left to right, over the spine order~$\prec_H$ and whenever we visit the second vertex $v$ of an edge $uv$ with $u \prec v$, we extend $T_p$ with a node at the corresponding position representing the face associated with $uv$.
				For a single page, this can be done in linear time and for all pages together, this takes $\BigO{\ell \cdot \Size{\instance}}$ time.
				Note that in total, the augmented dual trees use $\BigO{\ell \cdot \Size{\instance}}$ space.
				Since this data structure only depends on \lSL[H], we can re-use it in different invocations of the DP-Algorithm.
				Consequently, we compute it only once in the beginning and will neglect it for the running time of the DP (but consider it in the overall running time of the algorithm).
				Hence, the running time of the DP is \BigO{\nadd{} \cdot \Size{V(H)} \cdot \madd{}}.
				Since $\BigO{{\nadd}^2} \in \BigO{\nadd \cdot \Size{\instance}}$, together with the initial checks, this amounts to \BigO{\nadd{}\cdot\madd{}\cdot \Size{\instance{}}} time.
			}%
		\end{proof}
	\end{statelater}%
	\iftoggle{long}{%
		Finally, we observe that for assignment~(i), i.e., $\sigma_G$, there are \BigO{\ell^{\madd}} different possibilities, for assignment~(ii), i.e., \restrictedSpineOrder{}, there are \BigO{{\nadd{}}!} possibilities, for assignment~(iii), i.e., the assignment of new vertices to \superIntervals{}, there are \BigO{\madd{}^{\nadd{}}} possibilities, and for assignment~(iv), i.e., the distance to the outer face, there are \BigO{\omega^{\madd{}}} different possibilities.
		This gives us overall \BigO{\ell^{\madd} \cdot {\nadd{}}! \cdot \madd{}^{\nadd{}} \cdot \omega^{\madd{}}} different possibilities for assignments~(i)--(iv). 
	}{}%
	\noindent%
	\iftoggle{long}{%
		Applying \cref{lem:dp} to each of these, we get the desired theorem{%
			, where we observe that the $\BigO{\ell \cdot \Size{\instance}}$-time pre-processing step is dominated by the final running time:} %
		\theoremPageWidthFPT*
	}{
		Applying \cref{lem:dp} to each of these, we get the following theorem.%
		\stateFPTTheorem
	}

	\newcommand{\stateIndependentSetTheorem}{\begin{restatable}\restateref{thm:kappa-page-number-fpt}{theorem}{theoremIndependentSetFPT}
			\label{thm:kappa-page-number-fpt}
			Let $\instance = \instanceLong$ be an instance of \SLEShort{} where $G[\Vadd{}]$ is an independent set.
			We can find an $\ell$-page stack layout of $G$ that extends \lSL or report that none exists in 
			{%
				\BigO{\ell^{\madd{}} \cdot \nadd{}! \cdot {\madd{}}^{\nadd{} + 1} \cdot {\Size{\instance{}}}^2}}
			time{%
				, where $\nadd$ and $\madd$ denote the number of new vertices and edges, respectively.}
	\end{restatable}}
	
	\iftoggle{long}{
		\section{Towards a Tighter Fixed-Parameter Algorithm for \textsc{SLE}}
		\label{sec:fpt-independent-set}
		As a natural next step, we would like to generalize \cref{thm:kappa-page-number-page-width-fpt} by considering only $\kappa$ and~$\ell$ as parameters.
		However, the question of whether one can still achieve fixed-parameter tractability for \SLEShort when parameterizing by $\kappa+\ell$ is still open.
		Nevertheless, as our final result, we show that strengthening \cref{thm:kappa-page-number-page-width-fpt} is indeed possible at least in the restricted case where no two missing vertices are adjacent, as we can then greedily assign the first ``possible'' interval to each vertex that complies with assignment~(i)--(iii).
		\stateIndependentSetTheorem
	}{}
	\begin{prooflater}{ptheoremIndependentSetFPT}
		Observe that $G[\Vadd{}]$ being an independent set removes the need for synchronizing the position of adjacent new vertices to ensure that they are incident to the same face.
		We propose a fixed-parameter algorithm that loosely follows the ideas introduced in  \cref{sec:fpt} and adapts them to the considered setting.
		The following claim will become useful.
		\begin{claim}
			\label{claim:greedy}
			Given an instance $\instance = \instanceLong$ of \SLEShort{} where $G[\Vadd{}]$ is an independent set,
			\textbf{(i)} a page assignment~$\sigma_G$ for all edges,
			\textbf{(ii)} an order \restrictedSpineOrder{} in which the new vertices will appear along the spine, and
			\textbf{(iii)} for every new vertex $v \in \Vadd{}$ an assignment to a \superInterval{}.
			In %
			\BigO{\madd{} \cdot {\Size{\instance{}}}^2} 
			time, we can compute an $\ell$-page stack layout of~$G$ that extends \lSL and respects the given assignments (i)--(iii) or report that no such layout exists.
		\end{claim}
		Towards showing the claim, we first note that we only miss assignment~(iv) from \cref{lem:dp}.
		Hence, by the same arguments as in the proof of \cref{lem:dp}, we can check in %
		{%
			$\BigO{{\nadd{}}^2 + \Size{\instance{}}}$
		}%
		time whether assignments (i)--(iii) are consistent and do not imply a crossing.
		For the remainder of the proof, we assume that they are, as we can otherwise immediately return that there does not exist an $\ell$-page stack layout of $G$ that respects the assignments.
		
		We still need to assign concrete spine positions to new vertices.
		However, in contrast to \cref{thm:kappa-page-number-page-width-fpt}, there is no need to ensure that adjacent new vertices are in the same face, because no two new vertices are adjacent by assumption.
		This allows us to use a greedy variant of the DP from \cref{thm:kappa-page-number-page-width-fpt}.
		Note that in contrast to the approach behind \Cref{thm:kappa-page-number-page-width-fpt}, we now immediately add all new edges in \EaddH, i.e., consider them from now on as old, since their placement is fully defined by assignment~(i).
		
		\proofsubparagraph{The Greedy Algorithm.}
		We maintain a counter $j$ initialized at $j = 1$ and consider the $i$th interval $[a, b]$.
		If $[a, b] \in \assignedSuperInterval{v_j}$, i.e., if $[a, b]$ is part of the \superInterval{} for $v_j$, we check the following for every new edge $e = v_ju \in \Eadd{}$ incident to~$v_j$.
		Assuming that $v_j$ would be placed in $[a, b]$, we check whether $v_j$ sees $u$ on the page $\sigma(e)$.
		These checks can be done in \BigO{\madd{}\cdot \Size{\instance{}}} time.
		If this is the case, we place $v_j$ in the interval $[a, b]$ and increase the counter by one, otherwise we continue with the next interval $[b, c]$.
		We stop once we have $j = k + 1$ as we have assigned an interval to all new vertices.
		To obtain the $\ell$-page stack layout of $G$, we can store in addition for each vertex the interval we have placed it in.
		If after processing the last interval there are still some new vertices that have not been placed, we can return that there does not exist an $\ell$-page stack layout of $G$ that extends \lSL[H] and respects the assignments~(i)--(iii).
		The greedy algorithm runs in \BigO{\madd{} \cdot{\Size{\instance{}}}^2} time.

		\proofsubparagraph{Correctness of the Greedy Algorithm.}
		In the following, we show that if there is an $\ell$-page stack layout of $G$ that extends \lSL[H] and respects assignments (i)--(iii), then our greedy algorithm finds also {%
			one}.
		To this end, we assume that there is such a stack layout $\langle\prec^*_{G}, \sigma^*_{G}\rangle$.
		As we only assign intervals and thus spine positions to the new vertices, it suffices to show that if $\langle\prec^*_{G}, \sigma^*_{G}\rangle$ is a solution, so is $\langle\prec, \sigma^*_G\rangle$, where $\prec$ is the spine order we obtain with our greedy algorithm.
		Observe that we must find some tuple \lSL[G] that extends \lSL[H] and respects the assignments (i)--(iii), as we only ensure that no new edge incident to a new vertex crosses old edges.
		As this is clearly not the case in $\langle\prec^*_{G}, \sigma^*_{G}\rangle$, there must be a feasible interval for each new vertex.
		For the remainder of the proof, we assume that $\prec^*_{G}$ and $\prec$ differ only in the position of some new vertex $v$.
		This is without loss of generality, as we can apply the following arguments for all new vertices iteratively from left to right according to~$\prec^*_{G}$, until all of them are placed as in the greedy solution.
		As we assign the intervals greedily, we assume that $v \prec^*_{G} u$ implies $v \prec u$ for %
		{%
			every old vertex} $u$, i.e., $v$ appears in $\prec$ earlier than in~$\prec^*_{G}$.
		Clearly, $\langle\prec, \sigma^*_G\rangle$ extends \lSL[H] as $\langle\prec^*_{G}, \sigma^*_G\rangle$ does.
		Therefore, we only need to show that $\langle\prec, \sigma^*_G\rangle$ does not contain crossings.
		
		Towards a contradiction, assume that $\langle\prec, \sigma^*_G\rangle$ contains a crossing among the edges $e = vu$ and $e' = ab$.
		We assume that $e$ is a new edge incident to the new vertex $v$ and observe that~$u$ must be an old vertex.
		Furthermore, we assume without loss of generality that $b$ is also an old vertex.
		As already argued in the beginning, we have checked, when placing~$v$, that $e$ does not cross an old edge.
		Hence, we observe that $e'$ cannot be an old edge.
		As we also treated all new edges incident to two old vertices as old edges, we conclude that $e'$ must be a new edge incident to a new vertex~$a$ and an old vertex $b$.
		Furthermore, we assume $a \prec^*_{G} b$, $a \prec^*_{G} v$, and $v \prec^*_G u$.
		This is without loss of generality, as in any other case the arguments will be symmetric{%
			; in particular, observe that the other cases correspond to different assignments~(ii) and (iii).}
		As $\langle\prec^*_{G}, \sigma^*_G\rangle$ is crossing free, we have $a \prec^*_{G} b \prec^*_{G} v \prec^*_G u$ or $a \prec^*_{G} v \prec^*_{G} u \prec^*_{G} b$.
		Furthermore, as our greedy algorithm positions $v$ in $\prec$ further to the left compared to~$\prec^*_{G}$ and $e$ and $e'$ now cross, we must have $a \prec v \prec b \prec u$ or $v \prec a \prec u \prec b$, respectively.
		{%
			Observe that the relative order between the vertices $a$ and $b$, and between $u$ and $v$ must be identical in~$\prec^*_{G}$ and $\prec$, as we would otherwise immediately get a contradiction with assignment~(iii); observe that $u$ and $b$ define boundaries of super intervals.
			However, also neither of the above two orders is possible, as we show next.
		}%

		For the former case, i.e., when we turn $a \prec^*_{G} b \prec^*_{G} v \prec^*_G u$ into $a \prec v \prec b \prec u$, we observe that $b$ is an old vertex incident to a new edge, i.e., it defines a \superInterval{}, see also \cref{fig:fpt-is-correctness}a.
		Hence, having $b \prec^*_{G} v$ and $v \prec b$ implies that the \superInterval{} for $v$ differs between $\langle\prec^*_{G}, \sigma^*_G\rangle$ and $\langle\prec, \sigma^*_G\rangle$, which violates assignment~(iii) and thus contradicts our assumption on the existence of $\langle\prec^*_{G}, \sigma^*_{G}\rangle$ and $\langle\prec_{G}, \sigma^*_{G}\rangle$.
		
		In the latter case, when we turn $a \prec^*_{G} v \prec^*_{G} u \prec^*_{G} b$ into $v \prec a \prec u \prec b$, we move~$v$ left of the new vertex~$a$.
		Hence, we change the relative order among the new vertices, see \cref{fig:fpt-is-correctness}b.
		This is a contradiction to assignment~(ii) and thus our assumption on the existence of $\langle\prec^*_{G}, \sigma^*_{G}\rangle$ and $\langle\prec_{G}, \sigma^*_{G}\rangle$.
		
		As we obtain in all cases a contradiction, we conclude that $e$ and $e'$ cannot cross.
		Applying the above arguments inductively, we derive that our algorithm must find a solution if there exists one, i.e., is correct.
		
		\proofsubparagraph*{Putting Everything Together.}
		We can now branch over all possible assignments~(i)--(iii).
		{%
			Using the same reasoning as for \cref{thm:kappa-page-number-page-width-fpt}, we conclude that there are overall \BigO{\ell^{\madd} \cdot {\nadd{}}! \cdot \madd{}^{\nadd{}}} different possibilities for assignments~(i)--(iii).}
		Applying our greedy algorithm to each of them, we obtain the {%
			statement from the theorem}.
	\end{prooflater}
	\begin{statelater}{figFPTISCorrectness}
		\begin{figure}
			\centering
			\includegraphics[page=1]{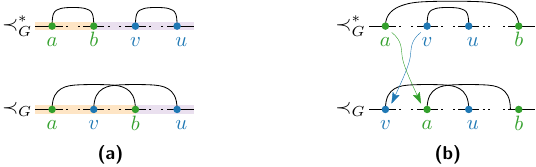}
			\caption{Illustration of the correctness arguments in the proof of \cref{thm:kappa-page-number-fpt}.
				In~\textbf{\textsf{(a)}}, we change the assignment of $v$ to \superIntervals{}, highlighted with the colors, and in~\textbf{\textsf{(b)}}, we change the relative order among the new vertices $a$ and $v$.
				Both scenarios lead to a contradiction.}
			\label{fig:fpt-is-correctness}
		\end{figure}
	\end{statelater}
	
	\section{Concluding Remarks}
	\label{sec:conclusion}
	Our results provide the first investigation of the drawing-extension problem for stack layouts through the lens of parameterized algorithmics. We show that the complexity-theoretic behavior of the problem is surprisingly rich and differs from that of previously studied drawing-extension problems. One prominent question left for future work is whether one can still achieve fixed-parameter tractability for \SLEShort when parameterizing by $\kappa+\ell$, thus generalizing \cref{thm:kappa-page-number-page-width-fpt}\iftoggle{long}{ and \cref{thm:kappa-page-number-fpt}.}{. As our final result, we show that this is indeed possible at least in the restricted case where no two missing vertices are adjacent, as we can then greedily assign the first ``possible'' interval to each vertex that complies with assignment~(i)--(iii).
		
		\stateIndependentSetTheorem}
	{%
		Note that parameterizing by $\omega + \ell$ cannot lead to a fixed-parameter algorithm for \SLEShort as recognizing graphs with stack number two is \NP-complete.}
	
	A further natural and promising direction for future work is to consider generalizing the presented techniques to other types of linear layouts. %
	{%
		Studying the problem of finding or extending a stack layout for two or more graphs on the same vertex set simultaneously could also lead to fruitful insights.}
	Finally, future work could also investigate the following generalized notion of extending linear layouts:
	Given a graph $G$, the spine order for some subset of its vertices and the page assignment for some subset of its edges, does there exist a linear layout of $G$ that extends both simultaneously?

	\bibliography{references-gd.bib}

@Article{Bhore2020,
  author  = {Sujoy Bhore and Robert Ganian and Fabrizio Montecchiani and Martin N{\"{o}}llenburg},
  journal = {Journal of Graph Algorithms and Applications},
  title   = {{P}arameterized {A}lgorithms for {B}ook {E}mbedding {P}roblems},
  year    = {2020},
  number  = {4},
  pages   = {603--620},
  volume  = {24},
  doi     = {10.7155/jgaa.00526},
}

@InProceedings{BhoreGMN19,
  author    = {Sujoy Bhore and Robert Ganian and Fabrizio Montecchiani and Martin N{\"{o}}llenburg},
  booktitle = {Proc. 27h International Symposium on Graph Drawing and Network Visualization (GD'19)},
  title     = {{P}arameterized {A}lgorithms for {B}ook {E}mbedding {P}roblems},
  year      = {2019},
  editor    = {Daniel Archambault and Csaba D. T{\'{o}}th},
  pages     = {365--378},
  publisher = {Springer},
  series    = {Lecture Notes in Computer Science},
  volume    = {11904},
  doi       = {10.1007/978-3-030-35802-0_28},
}

@Article{Chung1987,
  author  = {Chung, Fan R. K. and Leighton, Frank Thomson and Rosenberg, Arnold L.},
  journal = {SIAM Journal on Algebraic Discrete Methods},
  title   = {{E}mbedding {G}raphs in {B}ooks: {A} {L}ayout {P}roblem with {A}pplications to {VLSI} {D}esign},
  year    = {1987},
  number  = {1},
  pages   = {33--58},
  volume  = {8},
  doi     = {10.1137/0608002},
}

@Book{Cygan2015,
  author    = {Marek Cygan and Fedor V. Fomin and Lukasz Kowalik and Daniel Lokshtanov and D{\'{a}}niel Marx and Marcin Pilipczuk and Michal Pilipczuk and Saket Saurabh},
  publisher = {Springer},
  title     = {Parameterized {A}lgorithms},
  year      = {2015},
  isbn      = {978-3-319-21274-6},
  doi       = {10.1007/978-3-319-21275-3},
}

@Article{Haslinger1999,
  author  = {Haslinger, Christian and Stadler, Peter F.},
  journal = {Bulletin of Mathematical Biology},
  title   = {{RNA} {S}tructures with {P}seudo-knots: {G}raph-theoretical, {C}ombinatorial, and {S}tatistical {P}roperties},
  year    = {1999},
  number  = {3},
  pages   = {437--467},
  volume  = {61},
  doi     = {10.1006/bulm.1998.0085},
}

@Article{Stoehr.1988,
  author  = {Stöhr, Elena},
  journal = {Information and Computation},
  title   = {A {T}rade-off between {P}age {N}umber and {P}age {W}idth of {B}ook {E}mbeddings of {G}raphs},
  year    = {1988},
  number  = {2},
  pages   = {155--162},
  volume  = {79},
  doi     = {10.1016/0890-5401(88)90036-3},
}

@InProceedings{Unger.1992,
  author    = {Walter Unger},
  booktitle = {Proc. 9th Symposium on Theoretical Aspects of Computer Science (STACS'92)},
  title     = {The {C}omplexity of {C}olouring {C}ircle {G}raphs ({E}xtended {A}bstract)},
  year      = {1992},
  editor    = {Alain Finkel and Matthias Jantzen},
  pages     = {389--400},
  publisher = {Springer},
  series    = {Lecture Notes in Computer Science},
  volume    = {577},
  doi       = {10.1007/3-540-55210-3_199},
}

@Book{Garey.1979,
  author    = {Michael R. Garey and David S. Johnson},
  publisher = {W. H. Freeman},
  title     = {{C}omputers and {I}ntractability: {A} {G}uide to the {T}heory of {NP}-{C}ompleteness},
  year      = {1979},
}

@InProceedings{Eiben2020,
  author    = {Eduard Eiben and Robert Ganian and Thekla Hamm and Fabian Klute and Martin N{\"{o}}llenburg},
  booktitle = {Proc. 47th International Colloquium on Automata, Languages and Programming (ICALP'20)},
  title     = {{E}xtending {P}artial {1}-{P}lanar {D}rawings},
  year      = {2020},
  editor    = {Artur Czumaj and Anuj Dawar and Emanuela Merelli},
  pages     = {43:1--43:19},
  publisher = {Schloss Dagstuhl - Leibniz-Zentrum f{\"{u}}r Informatik},
  series    = {LIPIcs},
  volume    = {168},
  doi       = {10.4230/LIPICS.ICALP.2020.43},
}

@Article{Liu2021,
  author  = {Yunlong Liu and Jie Chen and Jingui Huang and Jianxin Wang},
  journal = {Theoretical Computer Science},
  title   = {On {P}arameterized {A}lgorithms for {F}ixed-{O}rder {B}ook {T}hickness with respect to the {P}athwidth of the {V}ertex {O}rdering},
  year    = {2021},
  pages   = {16--24},
  volume  = {873},
  doi     = {10.1016/J.TCS.2021.04.021},
}

@Article{Bhore2023,
  author  = {Bhore, Sujoy and Da Lozzo, Giordano and Montecchiani, Fabrizio and Nöllenburg, Martin},
  journal = {European Journal of Combinatorics},
  title   = {On the {U}pward {B}ook {T}hickness {P}roblem: {C}ombinatorial and {C}omplexity {R}esults},
  year    = {2023},
  pages   = {103662},
  volume  = {110},
  doi     = {10.1016/j.ejc.2022.103662},
}

@Article{Dujmovic2004,
  author  = {Dujmović, Vida and Wood, David R.},
  journal = {{Discrete Mathematics \& Theoretical Computer Science}},
  title   = {{O}n {L}inear {L}ayouts of {G}raphs},
  year    = {2004},
  volume  = {Vol. 6 no. 2},
  doi     = {10.46298/dmtcs.317},
}

@InProceedings{Ollmann1973,
  author    = {Ollmann, Taylor L.},
  booktitle = {Proc. 4th Southeastern Conference on Combinatorics, Graph Theory and Computing},
  title     = {On the {B}ook {T}hicknesses of {V}arious {G}raphs},
  year      = {1973},
  pages     = {459},
  volume    = {8},
}

@Article{Bilski1992,
  author  = {Bilski, Tomasz},
  journal = {IEE Proceedings E (Computers and Digital Techniques)},
  title   = {Embedding graphs in books: a survey},
  year    = {1992},
  number  = {2},
  pages   = {134},
  volume  = {139},
  doi     = {10.1049/ip-e.1992.0021},
}

@Article{Yannakakis1989,
  author  = {Yannakakis, Mihalis},
  journal = {Journal of Computer and System Sciences},
  title   = {{E}mbedding {P}lanar {G}raphs in {F}our {P}ages},
  year    = {1989},
  number  = {1},
  pages   = {36--67},
  volume  = {38},
  doi     = {10.1016/0022-0000(89)90032-9},
}

@Article{Bernhart1979,
  author  = {Frank Bernhart and Paul C. Kainen},
  journal = {Journal of Combinatorial Theory, Series B},
  title   = {The {B}ook {T}hickness of a {G}raph},
  year    = {1979},
  number  = {3},
  pages   = {320--331},
  volume  = {27},
  doi     = {10.1016/0095-8956(79)90021-2},
}

@Article{KaufmannBKPRU20,
  author  = {Michael A. Bekos and Michael Kaufmann and Fabian Klute and Sergey Pupyrev and Chrysanthi N. Raftopoulou and Torsten Ueckerdt},
  journal = {Journal of Computational Geometry},
  title   = {{F}our {P}ages {A}re {I}ndeed {N}ecessary for {P}lanar {G}raphs},
  year    = {2020},
  number  = {1},
  pages   = {332--353},
  volume  = {11},
  doi     = {10.20382/JOCG.V11I1A12},
}

@Article{Angelini2015,
  author  = {Angelini, Patrizio and Di Battista, Giuseppe and Frati, Fabrizio and Jelínek, Vít and Kratochvíl, Jan and Patrignani, Maurizio and Rutter, Ignaz},
  journal = {ACM Transactions on Algorithms},
  title   = {{T}esting {P}lanarity of {P}artially {E}mbedded {G}raphs},
  year    = {2015},
  number  = {4},
  pages   = {1--42},
  volume  = {11},
  doi     = {10.1145/2629341},
}

@Article{lbf-eupgd-20,
  author  = {Da Lozzo, Giordano and Di Battista, Giuseppe and Frati, Fabrizio},
  journal = {Computational Geometry Theory and Applications},
  title   = {Extending {U}pward {P}lanar {G}raph {D}rawings},
  year    = {2020},
  pages   = {101668},
  volume  = {91},
  doi     = {10.1016/j.comgeo.2020.101668},
}

@InProceedings{Patrignani05a,
  author    = {Maurizio Patrignani},
  booktitle = {Proc. 13th International Symposium on Graph Drawing and Network Visualization (GD'05)},
  title     = {{O}n {E}xtending a {P}artial {S}traight-{L}ine {D}rawing},
  year      = {2005},
  editor    = {Patrick Healy and Nikola S. Nikolov},
  pages     = {380--385},
  publisher = {Springer},
  series    = {Lecture Notes in Computer Science},
  volume    = {3843},
  doi       = {10.1007/11618058_34},
}

@Article{p-epsd-06,
  author  = {Patrignani, Maurizio},
  journal = {International Journal of Foundations of Computer Science},
  title   = {{O}n {E}xtending a {P}artial {S}traight-{L}ine {D}rawing},
  year    = {2006},
  number  = {5},
  pages   = {1061--1070},
  volume  = {17},
  doi     = {10.1142/S0129054106004261},
}

@Article{jkr-ktppeg-13,
  author  = {Jelínek, Vít and Kratochvíl, Jan and Rutter, Ignaz},
  journal = {Computational Geometry Theory and Applications},
  title   = {A {Kuratowski}-{T}ype {T}heorem for {P}lanarity of {P}artially {E}mbedded {G}raphs},
  year    = {2013},
  number  = {4},
  pages   = {466--492},
  volume  = {46},
  doi     = {10.1016/j.comgeo.2012.07.005},
}

@InProceedings{Eiben2020a,
  author    = {Eiben, Eduard and Ganian, Robert and Hamm, Thekla and Klute, Fabian and Nöllenburg, Martin},
  booktitle = {Proc. 45th Mathematical Foundations of Computer Science (MFCS'20)},
  title     = {{E}xtending {N}early {C}omplete 1-{P}lanar {D}rawings in {P}olynomial {T}ime},
  year      = {2020},
  editor    = {Javier Esparza and Daniel Kr{\'{a}}l'},
  pages     = {31:1--31:16},
  publisher = {Schloss Dagstuhl – Leibniz-Zentrum für Informatik},
  series    = {LIPIcs},
  volume    = {170},
  doi       = {10.4230/LIPICS.MFCS.2020.31},
}

@Article{Angelini2021,
  author  = {Patrizio Angelini and Ignaz Rutter and T. P. Sandhya},
  journal = {Journal of Graph Algorithms and Applications},
  title   = {{E}xtending {P}artial {O}rthogonal {D}rawings},
  year    = {2021},
  number  = {1},
  pages   = {581--602},
  volume  = {25},
  doi     = {10.7155/jgaa.00573},
}

@InProceedings{BhoreLMN21,
  author    = {Bhore, Sujoy and Da Lozzo, Giordano and Montecchiani, Fabrizio and Nöllenburg, Martin},
  booktitle = {Proc. 29th International Symposium on Graph Drawing and Network Visualization (GD'21)},
  title     = {{O}n the {U}pward {B}ook {T}hickness {P}roblem: {C}ombinatorial and {C}omplexity {R}esults},
  year      = {2021},
  editor    = {Helen C. Purchase and Ignaz Rutter},
  pages     = {242--256},
  publisher = {Springer},
  series    = {Lecture Notes in Computer Science},
  volume    = {12868},
  doi       = {10.1007/978-3-030-92931-2_18},
}

@InProceedings{Bhore2023a,
  author    = {Bhore, Sujoy and Ganian, Robert and Khazaliya, Liana and Montecchiani, Fabrizio and Nöllenburg, Martin},
  booktitle = {Proc. 39th International Symposium on Computational Geometry (SoCG'23)},
  title     = {{E}xtending {O}rthogonal {P}lanar {G}raph {D}rawings {I}s {F}ixed-{P}arameter {T}ractable},
  year      = {2023},
  editor    = {Erin W. Chambers and Joachim Gudmundsson},
  pages     = {18:1--18:16},
  publisher = {Schloss Dagstuhl – Leibniz-Zentrum für Informatik},
  series    = {LIPIcs},
  volume    = {258},
  doi       = {10.4230/LIPICS.SOCG.2023.18},
}

@TechReport{w-chcpmpg-82,
  author      = {Widgerson, Avi},
  institution = {Princeton University},
  title       = {The {C}omplexity of the {H}amiltonian {C}ircuit {P}roblem for {M}aximal {P}lanar {G}raphs},
  year        = {1982},
  number      = {298},
}

@InProceedings{Unger88,
  author    = {Walter Unger},
  booktitle = {Proc. 5th Symposium on Theoretical Aspects of Computer Science (STACS'88)},
  title     = {On the k-{C}olouring of {C}ircle-{G}raphs},
  year      = {1988},
  editor    = {Robert Cori and Martin Wirsing},
  pages     = {61--72},
  publisher = {Springer},
  series    = {Lecture Notes in Computer Science},
  volume    = {294},
  doi       = {10.1007/BFB0035832},
}

@InProceedings{gk-icl-06,
  author    = {Gansner, Emden R. and Koren, Yehuda},
  booktitle = {Proc. 14th International Symposium on Graph Drawing and Network Visualization (GD'06)},
  title     = {{I}mproved {C}ircular {L}ayouts},
  year      = {2006},
  editor    = {Kaufmann, Michael and Wagner, Dorothea},
  pages     = {386--398},
  publisher = {Springer},
  series    = {Lecture Notes in Computer Science},
  volume    = {4372},
  doi       = {10.1007/978-3-540-70904-6_37},
}

@InProceedings{bb-crcl-04,
  author    = {Michael Baur and Ulrik Brandes},
  booktitle = {Proc. 30th International Workshop on Graph-Theoretic Concepts in Computer Science (WG'04)},
  title     = {{C}rossing {R}eduction in {C}ircular {L}ayouts},
  year      = {2004},
  editor    = {Hromkovič, Juraj and Nagl, Manfred and Westfechtel, Bernhard},
  pages     = {332--343},
  publisher = {Springer},
  series    = {Lecture Notes in Computer Science},
  volume    = {3353},
  doi       = {10.1007/978-3-540-30559-0_28},
}

@InProceedings{w-dvss-02,
  author    = {Wattenberg, Martin},
  booktitle = {Proc. 8th IEEE Information Visualization Conference (InfoVis'02)},
  title     = {Arc {D}iagrams: {V}isualizing {S}tructure in {S}trings},
  year      = {2002},
  editor    = {Pak Chung Wong and Keith Andrews},
  pages     = {110-116},
  publisher = {{IEEE} Computer Society},
  doi       = {10.1109/INFVIS.2002.1173155},
}

@InProceedings{Brueckner2017,
  author    = {Brückner, Guido and Rutter, Ignaz},
  booktitle = {Proc. 28th ACM-SIAM Symposium on Discrete Algorithms (SODA'17)},
  title     = {{P}artial and {C}onstrained {L}evel {P}lanarity},
  year      = {2017},
  editor    = {Philip N. Klein},
  pages     = {2000--2011},
  publisher = {{SIAM}},
  doi       = {10.1137/1.9781611974782.130},
}

@InProceedings{Ganian2021,
  author    = {Ganian, Robert and Hamm, Thekla and Klute, Fabian and Parada, Irene and Vogtenhuber, Birgit},
  booktitle = {Proc. 48th International Colloquium on Automata, Languages and Programming (ICALP'21)},
  title     = {{C}rossing-{O}ptimal {E}xtension of {S}imple {D}rawings},
  year      = {2021},
  editor    = {Nikhil Bansal and Emanuela Merelli and James Worrell},
  pages     = {72:1--72:17},
  publisher = {Schloss Dagstuhl – Leibniz-Zentrum für Informatik},
  series    = {LIPIcs},
  volume    = {198},
  doi       = {10.4230/LIPICS.ICALP.2021.72},
}

@Misc{Pupyrev.2023,
  author = {Sergey Pupyrev},
  note   = {Last accessed: 2024-05-20},
  title  = {A {C}ollection of {E}xisting {R}esults on {S}tack and {Q}ueue {N}umbers},
  year   = {2023},
  url    = {https://spupyrev.github.io/linearlayouts.html},
}

@Book{Diestel2012,
  author    = {Reinhard Diestel},
  publisher = {Springer},
  title     = {{G}raph {T}heory, {4th} {E}dition},
  year      = {2012},
  isbn      = {978-3-642-14278-9},
  series    = {Graduate {T}exts in {M}athematics},
  volume    = {173},
}

@Book{DowneyF13,
  author    = {Rodney G. Downey and Michael R. Fellows},
  publisher = {Springer},
  title     = {{F}undamentals of {P}arameterized {C}omplexity},
  year      = {2013},
  isbn      = {978-1-4471-5558-4},
  series    = {Texts in Computer Science},
  doi       = {10.1007/978-1-4471-5559-1},
}

@InProceedings{AngeliniRP20,
  author    = {Patrizio Angelini and Ignaz Rutter and T. P. Sandhya},
  booktitle = {Proc. 28th International Symposium on Graph Drawing and Network Visualization (GD'20)},
  title     = {{E}xtending {P}artial {O}rthogonal {D}rawings},
  year      = {2020},
  editor    = {David Auber and Pavel Valtr},
  pages     = {265--278},
  publisher = {Springer},
  series    = {Lecture Notes in Computer Science},
  volume    = {12590},
  doi       = {10.1007/978-3-030-68766-3_21},
}

@InProceedings{GMOPR24,
  author    = {Robert Ganian and Haiko M{\"{u}}ller and Sebastian Ordyniak and Giacomo Paesani and Mateusz Rychlicki},
  booktitle = {Proc. 51st International Colloquium on Automata, Languages and Programming (ICALP'24)},
  title     = {{A} {T}ight {S}ubexponential-{T}ime {A}lgorithm for {T}wo-{P}age {B}ook {E}mbedding},
  year      = {2024},
  editor    = {Karl Bringmann and Martin Grohe and Gabriele Puppis and Ola Svensson},
  pages     = {68:1--68:18},
  publisher = {Schloss Dagstuhl - Leibniz-Zentrum f{\"{u}}r Informatik},
  series    = {LIPIcs},
  volume    = {297},
  doi       = {10.4230/LIPICS.ICALP.2024.68},
}

@article{Stohr91,
  author       = {Elena St{\"{o}}hr},
  title        = {The pagewidth of trivalent planar graphs},
  journal      = {Discrete Mathematics},
  volume       = {89},
  number       = {1},
  pages        = {43--49},
  year         = {1991},
  doi          = {10.1016/0012-365X(91)90398-L}
}

@Misc{ARXIV,
  author = {Depian, Thomas and Fink, Simon D. and Ganian, Robert and Nöllenburg, Martin},
  title  = {{T}he {P}arameterized {C}omplexity of {E}xtending {S}tack {L}ayouts},
  year   = {2024},
  doi    = {10.48550/ARXIV.2409.02833},
}

@InProceedings{GD-PAPER,
  author    = {Depian, Thomas and Fink, Simon D. and Ganian, Robert and Nöllenburg, Martin},
  booktitle = {Proc. 32nd International Symposium on Graph Drawing and Network Visualization (GD'24)},
  title     = {The {P}arameterized {C}omplexity {O}f {E}xtending {S}tack {L}ayouts},
  year      = {2024},
  editor    = {Stefan Felsner and Karsten Klein},
  pages     = {12:1--12:17},
  publisher = {Schloss Dagstuhl – Leibniz-Zentrum für Informatik},
  series    = {LIPIcs},
  volume    = {320},
  doi       = {10.4230/LIPICS.GD.2024.12},
}

@InProceedings{THICKNESS-SOFSEM,
  author    = {Thomas Depian and Simon D. Fink and Alexander Firbas and Robert Ganian and Martin N{\"{o}}llenburg},
  booktitle = {Proc. 50th Conference on Current Trends in Theory and Practice of Computer Science (SOFSEM'25)},
  title     = {{P}athways to {T}ractability for {G}eometric {T}hickness},
  year      = {2025},
  editor    = {Rastislav Kr{\'{a}}lovic and Vera Kurkov{\'{a}}},
  pages     = {209--224},
  publisher = {Springer},
  series    = {Lecture Notes in Computer Science},
  volume    = {15538},
  doi       = {10.1007/978-3-031-82670-2_16},
}

@InProceedings{QLE,
  author    = {Thomas Depian and Simon D. Fink and Robert Ganian and Martin N{\"{o}}llenburg},
  booktitle = {Proc. 51st International Workshop on Graph-Theoretic Concepts in Computer Science (WG'25)},
  title     = {The {P}eculiarities of {E}xtending {Q}ueue {L}ayouts},
  year      = {2026},
  editor = {Henning Fernau and Philipp Kindermann},
  volume = {16124},
  pages = {177--191},
  series = {Lecture Notes in Computer Science},
  publisher = {Springer},
  doi = {10.1007/978-3-032-11835-6_13}
}

	\iftoggle{long}{
		
		\appendix
		
		\iftoggle{long}{}{
			\newpage
			\section{Omitted Proofs from Section~\ref{sec:only-edges}}
			
			\lemmaOnlyEdgesRemoveEasyEdges*
			\label{lem:only-edges-remove-easy-edges*}
			\plemmaOnlyEdgesRemoveEasyEdges
			
			\theoremOnlyEdgesFPT*
			\label{thm:only-edges-fpt*}
			\ptheoremOnlyEdgesFPT
			
			\newpage
			\section{Omitted Details from Section~\ref{sec:paranp-hardness}}
			
			\sectionParaNPBaseLayout
			
			\subsection{Omitted Proofs}
			\lemmaFixationGadgetProperties*
			\label{lem:fixation-gadget-properties*}
			\plemmaFixationGadgetProperties
			
			\theoremParaNPHardness*
			\label{thm:paranp-hardness*}
			\ptheoremParaNPHardness

			\newpage
			\section{Omitted Details from Section~\ref{sec:w-1}}
			\label{app:w-1}
			
			\sectionWOneBaseLayout
			\sectionWOneEdges
			\sectionWOneFixationGadget

			\subsection{Omitted Proofs}
			\theoremWOne*
			\label{thm:w-1*}
			\ptheoremWOne
			
			\newpage
			\section{Omitted Details from Section~\ref{sec:fpt}}
			
			The main task left open in \cref{sec:fpt} was the proof of the following lemma.
			\lemmaDP*
			\label{lem:dp*}
			In the following, we use the following notation related to \superIntervals{}.
			\notationSuperInterval\space
			\sectionDP
			
			\theoremPageWidthFPT*
			\label{thm:kappa-page-number-page-width-fpt*}
			\begin{proof}
				We observe that for assignment~(i), i.e., $\sigma_G$, there are \BigO{\ell^{\madd}} different possibilities, for assignment~(ii), i.e., \restrictedSpineOrder{}, there are \BigO{{\nadd{}}!} possibilities, for assignment~(iii), i.e., the assignment of new vertices to \superIntervals{}, there are \BigO{\madd{}^{\nadd{}}} possibilities, and for assignment~(iv), i.e., the distance to the outer face, there are \BigO{\omega^{\madd{}}} different possibilities.
				This gives us overall \BigO{\ell^{\madd} \cdot {\nadd{}}! \cdot \madd{}^{\nadd{}} \cdot \omega^{\madd{}}} different possibilities for assignments~(i)--(iv). 
				The theorem thus follows by applying \cref{lem:dp} to each of these.
			\end{proof}
			
			\newpage
			\section{Omitted Details from Section~\ref{sec:conclusion}}
			
			In \cref{sec:conclusion}, we stated the following theorem, which we now want to prove.
			\theoremIndependentSetFPT*
			\label{thm:kappa-page-number-fpt*}
			\ptheoremIndependentSetFPT
			\figFPTISCorrectness
		}
		
		\newpage
		\section{Removing Multi-Edges}
		\label{app:removing-multi-edges}
		In the following, we will describe how one can adapt the \NP- and \W[1]-hardness reductions to not rely on multi-edges.
		While the basic idea is always the same, namely to introduce several auxiliary vertices in order to distribute the multiple edges among them, the concrete implementation of this idea depends on the (part of) the reduction we are currently discussing.
		
		\subsection{Fixation Gadget (Section~\ref{sec:fixation-gadget})}
		\label{app:multi-edges-fixation-gadget}
		The fixation gadget gave us a lot of power in the hardness reduction but relied on multi-edges.
		As our graphs $H$ and $G$ are assumed to be simple, it is now time to remove these multi-edges by distributing them over several (additional) auxiliary vertices.
		
		In order to do that, we no longer introduce the $2(F + 1)$-many vertices $b_1, \ldots, b_{F + 1}$ and $a_1, \ldots, a_{F + 1}$, but introduce for each page $p \neq p_d$ $2(F + 1)$ vertices $b_1^p, \ldots, b_{F + 1}^p$ and $a_1^p, \ldots, a_{F + 1}^p$.
		Without loss of generality, we can assume that the page $p_d$ is the $\ell$th page.
		We update the spine order $\prec_{H}$ by setting $b_i^p \prec v_i \prec a_i^p$ for $i \in [F + 1]$ and $p \in [\ell - 2]$.
		Furthermore, we fix the order among the newly introduced vertices by enforcing $b_i^{p + 1} \prec b_i^{p}$ and $a_i^p \prec a_i^{p + 1}$ for every $i \in [F + 1]$ and $p \in [\ell - 2]${%
			, and $a_i^{\ell - 1} \prec b_{i + 1}^{\ell - 1}$ for every $i \in [F]$.}
		The linear order $\prec$ is then obtained by taking the transitive closure of the above (partial) orders.
		Observe that the above spine order places for each page $p$ one vertex \emph{\textbf{b}}efore ($b_i^{p}$) and \emph{\textbf{a}}fter {%
			($a_i^{p}$)} the vertex $v_{i}$.
		
		Next, we adapt the edges that we created in \cref{sec:fixation-gadget} and their page assignment $\sigma_H$.
		To that extend, recall that we introduced for every $i \in [F + 1]$ and every page $p \neq p_d$ the edge $e(b_i, a_i, p) = b_ia_i$ with $\sigma(e(b_i, a_i, p)) = p$.
		Now, this edge should be incident to vertices created ``for the page $p$'', i.e., we create instead the edge $b_i^pa_i^p$ and set $\sigma(b_i^pa_i^p) = p$.
		Although the remaining edges created in \cref{sec:fixation-gadget} are already simple, we have to adapt (some of) them, as their incident vertices no longer exist, i.e., have been replaced.
		In particular, for every $i \in [F + 1]$, we now no longer introduce the edges $b_iv_i$ and $v_ia_i$ but the edges $b_i^{\ell - 1}v_i$ and $v_ia_i^{\ell - 1}$ and set $\sigma(b_i^{\ell - 1}v_i) = \sigma(v_ia_i^{\ell - 1}) = p_d$.
		Finally, we create instead of the edge $b_1a_{F + 1}$ the edge $b_1^{\ell - 1}a_{F + 1}^{\ell - 1}$ and set $\sigma(b_1^
		{\ell - 1}a_{F + 1}^{\ell - 1}) = p_d$.
		\cref{fig:adapted-fixation-gadget-example} is an adapted version of \cref{fig:fixation-gadget-example} and shows the updated construction.
		\begin{figure}[t]
			\centering
			\includegraphics[page=5]{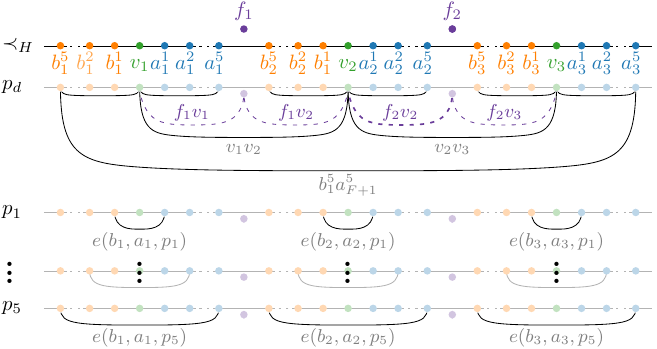}
			\caption{An example for an adapted fixation gadget for $F = 2$ with five other pages, i.e., $\ell - 1 = 5$. Adapted version of \cref{fig:fixation-gadget-example} that shows how we can remove multi-edges by distributing them over additional vertices.}
			\label{fig:adapted-fixation-gadget-example}
		\end{figure}
		We end this section with showing \cref{lem:adapted-fixation-gadget-properties}, which is an adapted version of \cref{lem:fixation-gadget-properties}.
		\begin{lemma}
			\label{lem:adapted-fixation-gadget-properties}
			Let $\instance = \instanceLong$ be an instance of \SLEShort{} that contains the \emph{adapted} fixation gadget (without multi-edges) on $F$ vertices $\{f_1, \ldots, f_{F}\}$.
			In any solution \lSL[G] to~\instance{} and for every $i \in [F]$, we have $v_i \prec f_i \prec v_{i + 1}$ and $\sigma(f_iv_i) = \sigma(f_{i}v_{i + 1}) = p_d$.
			Furthermore, the \emph{adapted} fixation gadget contributes $2F\ell + 2\ell - 1$ vertices and $(\ell + 4)F + \ell + 2$ edges to the size of \instance.
		\end{lemma}
		\begin{proof}
			Let \lSL[G] be a solution to \instance{}.
			Regarding the two properties of any possible solution to~\instance{}, i.e., that $v_i \prec f_i \prec v_{i + 1}$ and $\sigma(f_iv_i) = \sigma(f_{i}v_{i + 1}) = p_d$ holds for every $i \in [F]$, it suffices to make the following observations in the proof of \cref{lem:fixation-gadget-properties}.
			
			We first consider the argument we made to show $v_i \prec f_i \prec v_{i + 1}$.
			There, we first assumed that $f_i \prec v_i$ would hold for an $i \in [F]$.
			Using the observation on the presence of the edges $e(b_{i + 1}, a_{i + 1}, p)$ for every page $p \neq p_d$ and the spine order $b_{i + 1} \prec_{H} v_{i + 1} \prec_{H} a_{i + 1}$, we concluded that $f_{i}$ can see $v_{i + 1}$ only on the page $p_d$.
			This ultimately led to a contradiction to the assumption on the existence of a solution with $f_i \prec v_i$.
			Now, for the modified fixation gadget we can make a similar observation.
			Consider any page $p \neq p_d$.
			The graph $H$ contains the edge $b_{i + 1}^pa_{i + 1}^p$ and we have $f_i \prec b_{i + 1}^p \prec_{H} v_{i + 1} \prec_{H} a_{i + 1}^p$.
			Hence, $f_{i}$ can still see $v_{i + 1}$ only on the page $p_d$, i.e., we still must have $\sigma(f_iv_{i + 1}) = p_d$.
			And by the very same arguments as in the proof of \cref{lem:fixation-gadget-properties} this leads to a contradiction.
			
			Let us now re-visit the argument to show $\sigma(f_iv_i) = \sigma(f_{i}v_{i + 1}) = p_d$.
			There, we assumed the existence of a solution \lSL with $\sigma(f_iv_i) \neq p_d$ for some $i \in [F]$, i.e., we assumed that $\sigma(f_iv_i) = p$ holds for some page $p \neq p_d$.
			Under this assumption, we (again) make the observation that we have the edge $b_i^pa_i^p$ with $\sigma_H(b_i^pa_i^p) = p$, which allows us to strengthen above result to $v_i \prec f_i \prec a_i^p$.
			As in the proof of \cref{lem:fixation-gadget-properties}, we deduce from $\sigma_{H}(v_ia_i^{\ell - 1}) = p_d$ and $\sigma_{H}(b_{i + 1}^{p'}a_{i + 1}^{p'}) = p'$ for any page $p' \neq p_d$ that there does not exist a feasible page assignment for the edge $f_iv_{i + 1}$, which leads to a contradiction as in the proof of \cref{lem:fixation-gadget-properties}.
			
			Regarding the size of the adapted fixation gadget, we first observe that we have not introduced any additional edge, but rather re-distributed existing edges to additional vertices.
			For the number of vertices in the adapted fixation gadget, we recall that $\mathcal{F}$ consists of $F$ vertices and we have $F + 1$ vertices of the form $v_i$.
			As we introduce $2(F + 1)$ vertices per page $p \neq p_d$, this amounts to $2(F + 1)(\ell - 1)$ additional vertices.
			Combining all, we conclude that we have $2F\ell + 2\ell - 1$ vertices in the adapted fixation gadget. 
		\end{proof}
		
		\subsection{\textsf{NP}-Hardness Reduction (Section~\ref{sec:paranp-hardness})}
		\label{app:multi-edges-paranp-hardness}
		Our reduction from \cref{sec:paranp-hardness} that we used to show that \SLEShort{} with two new vertices is \NP-complete, see \cref{thm:paranp-hardness}, creates a graph $H$ with several multi-edges.
		More concretely, we created multi-edges both in the base layout from \cref{sec:paranp-hardness-base-layout}, for example to block visibility to a vertex, and in the fixation gadget.
		For the latter part of the reduction, i.e., the fixation gadget, we have discussed in \cref{app:multi-edges-fixation-gadget} how to remove the multi-edges.
		This section is devoted to removing the multi-edges in the former part, i.e., the base layout.
		At the end of this section, we argue that our reduction remains correct.
		
		\subparagraph*{Removal of the Multi-Edges.}
		Let $\varphi = (\mathcal{X}, \mathcal{C})$ be an instance of \ThreeSat{} with $N = \Size{\mathcal{X}}$ variables and $M = \Size{\mathcal{C}}$ clauses.
		Recall that we introduced $N + M + 1$ dummy vertices~$d_q$ in \cref{sec:paranp-hardness-base-layout} and distributed them in $\prec_H$ on the spine, i.e., we set $d_i \prec x_i \prec d_{i + 1}$ and $d_{N + j} \prec c_j \prec d_{N + j + 1}$ for every $i \in [N]$ and $j \in [M]$.
		We now create $2N(N + M + 1)$ dummy vertices $d_{q}^{p}$ instead with $q \in [N + M + 1]$ and $p \in [2N]$, i.e., $N + M + 1$ dummy vertices for every page $p \in [\ell] \setminus \{p_d\}$ associated to a variable.
		These new dummy vertices are ordered on the spine as follows.
		We set $d_i^p \prec x_i \prec d_{i + 1}^p$ and $d_{N + j}^p \prec c_j \prec d_{N + j + 1}^p$ for every $i \in [N]$, $j \in [M]$, and $p \in [2N]$.
		Furthermore, we set $d_q^{p_{i}} \prec d_q^{p_{\lnot i}} \prec d_q^{p_{i + 1}}$ for every $q \in [N + M + 1]$ and $i \in [N - 1]$.
		We obtain the linear order~$\prec_H$ by taking the transitive closure of the above relative orders; see also \cref{fig:adapted-paranp-hardness-base-layout-edges}.
		\begin{figure}[t]
			\centering
			\includegraphics[page=8]{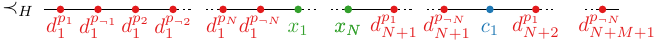}
			\caption{Spine order $\prec_H$ of the adapted base layout of \cref{sec:paranp-hardness-base-layout}.}
			\label{fig:adapted-paranp-hardness-base-layout-edges}
		\end{figure}
		
		Next, we redistribute the multi-edges over the new dummy vertices.
		Recall that we created for every pair of variables $x_i, x_j \in \mathcal{X}$ with $i \neq j$ the edges $e(x_i, p_j) = d_id_{i + 1}$ and $e(x_i, p_{\lnot j}) = d_id_{i + 1}$.
		Now, we create the edges $e(x_i, p_j) = d_{i}^{p_j}d_{i + 1}^{p_j}$ and $e(x_i, p_{\lnot j}) = d_{i}^{p_{\lnot j}}d_{i + 1}^{p_{\lnot j}}$ instead.
		We leave the page assignment $\sigma_H$ as it is, i.e., we have $\sigma(e(x_i, p_j)) = p_j$ and $\sigma(e(x_i, p_{\lnot j})) = p_{\lnot j}$.
		Furthermore, we created for every combination of a clause $c_j \in \mathcal{C}$ and a variable $x_i \in \mathcal{X}$ edges depending on the occurrence of $x_i$ in $c_j$.
		We now also re-distribute these edges as follows.
		If $x_i$ does not appear in $c_j$, we create the edges $e(c_j, p_i) = d_{N + j}^{p_{i}}d_{N + j + 1}^{p_{i}}$ and $e(c_j, p_{\lnot i}) = d_{N + j}^{p_{\lnot i}}d_{N + j + 1}^{p_{\lnot i}}$.
		We set $\sigma(e(c_j, p_i)) = p_i$ and $\sigma(e(c_j, p_{\lnot i})) = p_{\lnot i}$.
		If~$x_i$ appears in $c_j$ without negation, we create the edge $e(c_j, p_i) = d_{N + j}^{p_{i}}d_{N + j + 1}^{p_{i}}$ and set $\sigma(e(c_j, p_i)) = p_{i}$.
		Symmetrically, if $x_i$ appears negated in $c_j$, we create the edge $e(c_j, p_{\lnot i}) = d_{N + j}^{p_{\lnot i}}d_{N + j + 1}^{p_{\lnot i}}$ and set $\sigma(e(c_j, p_{\lnot i})) = p_{\lnot i}$.
		
		Finally, we set $a_3^{p_{\ell - 1}} \prec d_{1}^{p_1}$ to ensure that the (in \cref{app:multi-edges-fixation-gadget} adapted) fixation gadget is placed at the very beginning of the spine.
		Furthermore, we add the edge $d_1^{p_1}d_{N +M + 1}^{p_{\lnot N}}$ and set $\sigma(d_1^{p_1}d_{N +M + 1}^{p_{\lnot N}}) = p_d$ to ensure that our adapted construction still has \cref{prop:fixation-gadget-dummy-page}.
		
		\subparagraph*{Correctness of the Reduction.}
		First, observe that we added a polynomial number of additional vertices to $H$ and no new edges.
		Hence, the size of $H$ (and $G$) increased by a factor polynomial in the size of $\varphi$ and the size of \instance{} remains polynomial in the size of~$\varphi$.
		Furthermore, due to \cref{lem:adapted-fixation-gadget-properties}, it still holds that in any solution \lSL[G] to \instance{} we have $s \prec v \prec d_{1}^{p_1} \prec x_i \prec c_j \prec d_{N + M + 1}^{p_{\lnot N}}$ for every $i \in [N]$ and $j \in [M]$.
		Similarly, the relative order among the dummy vertices on the spine in the adapted reduction did not change compared to the reduction from \cref{sec:paranp-hardness}, i.e., previously we had $d_i \prec x_i \prec d_{i + 1}$ and $d_{N + j} \prec c_j \prec d_{N + j + 1}$ and now we have $d_{i}^p \prec x_i \prec d_{i + 1}^p$ and $d_{N + j}^p \prec c_j \prec d_{N + j + 1}^p$ for every $i \in [N]$ and $j \in [M]$ (and $p \in [2N]$).
		Furthermore, the edges span over the same vertices.
		Hence, we conclude that the proof of \cref{thm:paranp-hardness} readily carries over, i.e., \cref{thm:paranp-hardness} also holds for our modified construction that does not have multi-edges.
		
		\subsection{\textsf{W}[1]-Hardness Reduction (Section~\ref{sec:w-1})}
		\label{app:multi-edges-w-1}
		{%
			Finally, we discuss the removal of multi-edges in the \W[1]-hardness reduction from \cref{sec:w-1}.}
		Although we can already remove some multi-edges when incorporating the adapted fixation gadget from \cref{app:multi-edges-fixation-gadget}, some multi-edges are also introduced in \cref{sec:w-1-layer}.
		In particular, consider the case where we have two edges $u_{\alpha}^iu_{\beta}^j, u_{\alpha}^{i'}u_{\beta}^{j'} \in E(G_C)$.
		Then, for a $\gamma \in [k] \setminus \{\alpha, \beta\}$, we would create twice the edge $u_{\gamma}^1u_{\gamma}^{n_{\gamma} + 1}$ but assign them to different pages.
		In this section, we remove these multi-edges by using the additional (dummy) vertices that we obtain from the adapted fixation gadget, 
		However, note that this will not effect our intended equivalence from \cref{eq:w-1-equivalence} between a solution \lSL[G] to \SLEShort{} and a solution $\mathcal{C}$ to \MCC.
		
		In the following, we extend in \cref{app:multi-edges-w-1-base-layout} the base layout of our reduction to accommodate the additional vertices for the adapted fixation gadget.
		In \cref{app:multi-edges-w-1-layers}, we describe how to distribute the multi-edges from \cref{sec:w-1-layer} across the new vertices.
		Finally, we argue in \cref{app:multi-edges-w-1-correctness} that our reduction is still correct.
		
		\subsubsection{Adapting the Base Layout of Our Reduction}
		\label{app:multi-edges-w-1-base-layout}
		We now extend the base layout from \cref{sec:w-1-base} by additional vertices that in the end will be identified with their respective ``partner'' in the (adapted) fixation gadget.
		For each edge $e \in E(G_C)$ and color $\alpha \in [k + 1]$, we create the vertices $b_{\alpha}^{e}$ and $a_{\alpha}^e$.
		To order these additional vertices on the spine, we assume that the edges $E(G_C)$ are ordered, i.e., that we have $E(G_C) = \{e_1, \ldots, e_M\}$.
		We use this ordering to extend the spine order $\prec_{H}$ as follows, where we assume $u_{0}^{n_{0} + 1} = u_0^0$ for ease of notation.
		For all $\alpha \in [k + 1]$ and $i \in [M]$, we set $u_{\alpha-1}^{n_{\alpha - 1} + 1} \prec b_{\alpha}^{e_i} \prec u_{\alpha}^0 \prec a_{\alpha}^{e_i} \prec u_{\alpha}^1$.
		Furthermore, for $j \in [M - 1]$, we also set $b_{\alpha}^{e_j + 1} \prec b_{\alpha}^{e_{j}} \prec a_{\alpha}^{e_j} \prec a_{\alpha}^{e_{j + 1}}$.
		Observe that, as for the adapted fixation gadget, we place for each edge $e \in E(G_C)$ and each $\alpha \in [k + 1]$ one vertex \emph{\textbf{b}}efore ($b_e^{\alpha}$) and one \emph{\textbf{a}}fter ($a_e^{\alpha}$) {%
			the vertex} $u_{\alpha}^0$.
		The vertices before $u_{\alpha}^0$ are ordered decreasingly by the index of their respective edge and the vertices after $u_{\alpha}^0$ are ordered increasingly by the index of their respective edge.
		To obtain the adapted linear order $\prec_H$ we take the transitive closure of the above relation and the relations from \cref{sec:w-1-base}.
		
		Finally, we (re-)introduce the (adapted) fixation gadget on $F = k$ vertices $\mathcal{F} = \mathcal{X}$ by identifying the following vertices for $\alpha \in [k + 1]$. Note that we use $i = \alpha$ to distinguish between the vertices of the fixation gadget and hardness reduction.
		As in \cref{sec:w-1-fixation-gadget}, we identify $v_{i} = u_{\alpha}^0$.
		For the other vertices, we identify $b_{i}^p = b_{i}^{e_p}$, and $a_{i}^p = a_{\alpha}^{e_p}$, i.e., the vertices for the $p$th edge $e_p \in E(G_C)$ are identified with the vertices for page $p$ in the fixation gadget.
		The edges of the fixation gadget are adapted accordingly.

		\subsubsection{Redistributing the Multi-Edges}
		\label{app:multi-edges-w-1-layers}
		While incorporating the adapted fixation gadget removes some of the multi-edges, we still have to deal with the multi-edges introduced in \cref{sec:w-1-layer}, where we encoded the adjacencies from $G_C$ into $\lSL[H]$.
		We do this now by redistributing the edges over the new vertices introduced in \cref{app:multi-edges-w-1-base-layout}
		\cref{fig:adapted-w-1-layer} visualizes this process.
		\begin{figure}
			\centering
			\includegraphics[page=6]{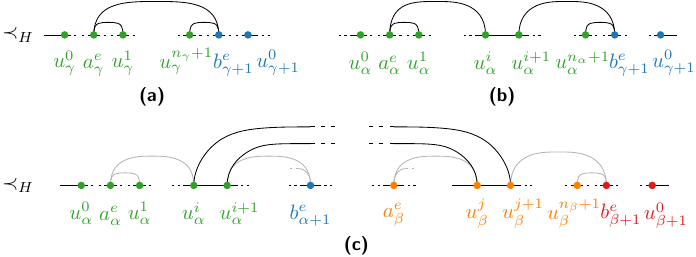}
			\caption{The edges of $H$ on page $p_e$ that model the adjacency of the edge $e = v_{\alpha}^{i}v_{\beta}^j \in E(G_C)$ in our adapted reduction of \cref{sec:w-1-layer}. We show the respective edges of $H$~\textbf{\textsf{(a)}} for a color $\gamma \in [\kappa] \setminus \{\alpha, \beta\}$,~\textbf{\textsf{(b)}} for the color $\alpha$, and~\textbf{\textsf{(c)}} that create a tunnel that connects $v_{\alpha}^{i}$ with $v_{\beta}^j$. Note that the gray edges in~\textbf{\textsf{(c)}} are those from~\textbf{\textsf{(a)}} and~\textbf{\textsf{(b)}}. This figure is an adapted version of \cref{fig:w-1-layer}.}
			\label{fig:adapted-w-1-layer}
		\end{figure}
		
		Let $e = v_{\alpha}^{i}v_{\beta}^{j} \in E(G_C)$ be an edge of $G_C$ and assume $\alpha < \beta$.
		Recall that we created in \cref{sec:w-1-layer} a set of edges in $H$ dedicated to $e$ and assigned them with $\sigma_{H}$ to the page~$p_e$.
		In particular, we created for every $\gamma \in [k] \setminus \{\alpha, \beta\}$ the edge $u_{\gamma}^{1}u_{\gamma}^{n_{\gamma} + 1}$.
		This edge is replaced by the {%
			path on the} edges $a_{\gamma}^eu_{\gamma}^{1}$, $a_{\gamma}^eb_{\gamma + 1}^e$, and $u_{\gamma}^{n_{\gamma} + 1}b_{\gamma + 1}^e$, as in \cref{fig:adapted-w-1-layer}a.
		Furthermore, we created the edges $u_{\alpha}^{1}u_{\alpha}^{i}$ and $u_{\alpha}^{i + 1}u_{\alpha}^{n_{\alpha} + 1}$.
		Instead, we now create a {%
			path on the} edges $a_{\alpha}^{e}u_{\alpha}^{1}$ and $a_{\alpha}^{e}u_{\alpha}^{i}$, and $u_{\alpha}^{i + 1}b_{\alpha + 1}^{e}$ and $u_{\alpha}^{n_{\alpha} + 1}b_{\alpha + 1}^{e}$ as shown in \cref{fig:adapted-w-1-layer}b.
		Similarly, we also create a {%
			path on the} edges $a_{\beta}^{e}u_{\beta}^{1}$, $a_{\beta}^{e}u_{\beta}^{j}$, $u_{\beta}^{j + 1}b_{\beta + 1}^{e}$, and $b_{\beta + 1}^{e}u_{\beta}^{n_{\beta} + 1}$.
		All of these edges are assigned to the page $p_e$.
		Observe that for $i \in \{1, n_{\alpha}\}$ or $j \in \{1, n_{\beta}\}$, above edges would become {%
			self-loops and} multi-edges (on the same page), which can easily be avoided.
		Recall that we created in \cref{sec:w-1-layer} a tunnel on the page $p_e$ by adding the edges $u_{\alpha}^{i}u_{\beta}^{j + 1}$ and $u_{\alpha}^{i + 1}u_{\beta}^{j}$, see also \cref{fig:adapted-w-1-layer}c.
		We do not need to adapt these edges, as they can only result in multi-edges if $G_C$ would contain them.
		One can readily verify that we no longer introduce multi-edges while still ensuring that the page $p_e$ is crossing free.
		
		Finally, recall that with the original version of the fixation gadget, we had for $\alpha \in [k + 1]$ the edges $\predecessor{u_{\alpha}^0}u_{\alpha}^0$ and $u_{\alpha}^0\successor{u_{\alpha}^0}$ that were assigned to the page $p_d$.
		In particular, these edges corresponded for $\alpha \in [k]$ to the edges $u_{\alpha}^0u_{\alpha}^1$ and $u_{\alpha}^{n_{\alpha} + 1}u_{\alpha + 1}^0$.
		However, these edges are no longer present in the adapted reduction, since no vertex of the (adapted) fixation gadget is identified with the vertices $u_{\alpha}^1$ or $u_{\alpha}^{n_{\alpha} + 1}$.
		As {%
			these} edges facilitated the arguments {%
			for showing} correctness of our approach, we re-introduce them in $\lSL[H]$ as follows.
		For every $\alpha \in [k]$, we add the edges $u_{\alpha}^0u_{\alpha}^1$ and $u_{\alpha}^{n_{\alpha} + 1}u_{\alpha + 1}^0$ and set $\sigma(u_{\alpha}^0u_{\alpha}^1) = \sigma(u_{\alpha}^{n_{\alpha} + 1}u_{\alpha + 1}^0) = p_d$.
		As the above edges span the respective edges $u_{\alpha}^0a_{\alpha}^{p_{e_M}}$ and $b_{\alpha + 1}^{p_{e_M}}u_{\alpha + 1}^0$, they do not introduce crossings on the page $p_d$.

		\subsubsection{Showing Correctness of the Modified Reduction}
		\label{app:multi-edges-w-1-correctness}
		This completes the adaptations we need to make to our reduction {%
			to remove all multi-edges} and it remains to show that \cref{thm:w-1} still holds.
		Regarding the size of the created instance, it is sufficient to observe that we introduce $2(k + 1)M$ additional vertices (for the adapted fixation gadget, see also \cref{lem:adapted-fixation-gadget-properties}, and for each edge of $G_C$ we introduce a constant number of additional edges to the already existing ones.
		Thus, the size of the instance \instance{} remains polynomial in the size of $G_C$ and we still have $\kappa = 3k + \binom{k}{2}$.
		So it remains to show the correctness of the reduction.
		
		For that, we can, on the one hand, {%
			observe that the transformation of a solution to an instance of \MCC{} to a solution to the created instance of \SLEShort{} as described in the proof of \cref{thm:w-1} is unaffected.}
		One way to see this is that all the additional edges that we introduced run between vertices placed between $u_{\alpha}^0$ and $u_{\alpha}^1$ or $u_{\alpha}^{n_{\alpha} + 1}$ and $u_{\alpha + 1}^0$ for the corresponding $\alpha \in [k + 1]$.
		Furthermore, if we changed existing edges, then we moved their incident vertices from $u_{\alpha}^1$ or $u_{\alpha}^{n_{\alpha} + 1}$ to a vertex in the above range.
		However, we created the spine order $\prec_G$ such that $u_{\alpha}^1 \prec x_{\alpha} \prec u_{\alpha}^{n_{\alpha} + 1}$ holds.
		Hence, the relative order $\prec_G$ among two new vertices $x_{\alpha}$ and $x_{\beta}$ with $\alpha, \beta \in [k]$, or a new vertex $x_{\alpha}$ and an old vertex $u_{\beta}^i$ with $\alpha, \beta \in [k]$, $\alpha \neq \beta$, and $i \in \{0, 1, n_{\beta} + 1\}$ remains unchanged.
		{%
			Therefore, we conclude that the ``$(\Rightarrow)$-direction'' still holds.}
		
		On the other hand, for the ``$(\Leftarrow)$-direction'', we used the fact that the created instance~$\instance{}$ of \SLEShort{} fulfills \cref{property:w-1-x-i-v-i,property:w-1-layer} to construct the solution $\mathcal{C}$.
		Hence, if we can convince ourselves that \instance{} still fulfills said properties, then the arguments we gave in \cref{thm:w-1} will readily carry over.
		Recall that \cref{property:w-1-x-i-v-i} is defined as follows.
		\mccXiViProperty*
		\noindent
		To see that we still have this property, we can observe that we incorporated in \cref{app:multi-edges-w-1-base-layout} the adapted fixation gadget on $k$ vertices into our construction.
		As we identify, for $\alpha \in [k]$, $v_{i} = u_{\alpha}^0$ and $f_{i} = x_{\alpha}$ for $i = \alpha$, \cref{property:w-1-x-i-v-i} follows directly from \cref{lem:adapted-fixation-gadget-properties}.
		
		Recall that we introduced at the end of \cref{app:multi-edges-w-1-layers} for every $\alpha \in [k]$ the edges $u_{\alpha}^0u_{\alpha}^1$ and $u_{\alpha}^{n_{\alpha} + 1}u_{\alpha + 1}^0$ on the dummy page $p_d$.
		We now use these edges to make the following observation.
		From \cref{lem:adapted-fixation-gadget-properties}, we get that we have in any solution \lSL[G] and for every $\alpha \in [k]$ that $\sigma(x_{\alpha}u_{\alpha}^0) = \sigma(x_{\alpha}u_{\alpha + 1}^0) = p_d$ holds.
		From \cref{property:w-1-x-i-v-i}, that still holds in our construction, we get $u_{\alpha}^0 \prec x_{\alpha} \prec u_{\alpha + 1}^0$ for every $\alpha \in [k]$.
		Using the above-mentioned edges on page $p_d$, we observe that we cannot have $u_{\alpha}^0 \prec x_{\alpha} \prec {u_{\alpha}^1}$ or $u_{\alpha}^{n_{\alpha} + 1} \prec x_{\alpha} \prec u_{\alpha + 1}^0$, as this would introduce a crossing on the page~$p_d$.
		Thus, our construction not only fulfills \cref{property:w-1-x-i-v-i}, but, furthermore, \cref{prop:w-1-x-i-v-i-strong} still applies.
		We will now argue that our construction also fulfills \cref{property:w-1-layer}, which is defines as follows.
		\mccLayerProperty*
		\noindent
		To see that \cref{property:w-1-layer} still holds, we first apply \cref{prop:w-1-x-i-v-i-strong}.
		This allows us to conclude that we have $u_{\alpha}^1 \prec x_{\alpha} \prec u_{\alpha}^{n_{\alpha} + 1}$ and $u_{\beta}^1 \prec x_{\beta} \prec u_{\beta}^{n_{\beta} + 1}$.
		Then, by exchanging $u_{\alpha}^0$ with~$a_{\alpha}^{e}$ and $u_{\alpha + 1}^{0}$ with $b_{\alpha + 1}^e$ in the proof of \cref{lem:w-1-fulfills-property-layer}, we can exclude $u_{\alpha}^1 \preceq x_{\alpha} \preceq u_{\alpha}^{i}$ and $u_{\alpha}^{i + 1} \preceq x_{\alpha} \preceq u_{\alpha + 1}^{0}$.
		Therefore, we derive that $x_{\alpha}$ must be placed in $\intervalPlacing{v_{\alpha}^i}$ and analogously must $x_{\beta}$ be placed in $\intervalPlacing{v_{\beta}^i}$, i.e., our construction still fulfills \cref{property:w-1-layer}.
		
		We use this now to argue that the ``$(\Leftarrow)$-direction'' of our reduction is still correct.
		\subparagraph*{Correctness of the ``$\boldsymbol{(\Leftarrow)}$-Direction'' in the Reduction.}
		We perform the arguments as in the ``$(\Leftarrow)$-direction'' of the proof of \cref{thm:w-1}.
		Most of the arguments are implied by \cref{property:w-1-x-i-v-i,property:w-1-layer}, which are also fulfilled in our adapted reduction.
		However, in the proof of \cref{thm:w-1}, we needed to argue that the pre-requisites for \cref{property:w-1-layer} are fulfilled.
		In the following, we perform the same argument for our adapted construction.
		Let us again assume that $x_{\alpha}$ and $x_{\beta}$ are placed in \intervalPlacing{v_{\alpha}^i} and \intervalPlacing{v_{\beta}^j}, respectively.
		We again consider the edge $x_{\alpha}x_{\beta} \in E(G)$ and the page~$p$ it is placed in the solution \lSL[G].
		By the very same arguments as in the proof of \cref{thm:w-1}, we can exclude $p = p_d$.
		Furthermore, we have $\sigma_{H}(a_{\alpha}^{e}b_{\alpha + 1}^{e}) = p_e$ and $a_{\alpha}^e \prec u_{\alpha}^1 \prec x_{\alpha} \prec u_{\alpha}^{n_{\alpha} + 1} \prec b_{\alpha + 1}^{e} \prec x_{\beta}$ for an edge $e  = uv \in E(G_C)$ with $u \not\in V_{\alpha}$ and $v \neq V_{\alpha}$.
		Analogous arguments also hold for $\beta$ and thus we get that $p = p_e$ for an edge $e \in E(G_C) \cap (V_{\alpha} \times V_{\beta})$ must hold.
		This shows that all prerequisites for \cref{property:w-1-layer} are fulfilled.
		Hence, we obtain that the edge $x_{\alpha}x_{\beta}$ is placed in the page $p_e$ that we created for the edge $v_{\alpha}^{i}v_{\beta}^{j}$ and, therefore, are the vertices $v_{\alpha}^i$ and $v_{\beta}^j$ adjacent in~$G_C$.
		Consequently, the proof of the ``$(\Leftarrow)$-direction'' in \cref{thm:w-1} carries readily over.
		
		Combining all, we conclude that \cref{thm:w-1} also holds for our modified construction that does not have multi-edges.
		
	}{}

\end{document}